\documentclass[english,10pt]{article}

\usepackage[english]{babel}
\usepackage[utf8]{inputenc}
\usepackage[T1]{fontenc}
\usepackage[formats]{listings}
\usepackage{geometry}
\usepackage{enumitem}
\usepackage{authblk}
\usepackage{csquotes}

\usepackage{graphicx}
\usepackage[colorinlistoftodos]{todonotes}
\usepackage{setspace}
\usepackage{placeins}
\usepackage{enumitem}
\usepackage{titlesec}
\usepackage{comment}
\usepackage{caption}
\usepackage{subcaption}
\usepackage{enumitem}

\usepackage[colorlinks=true, allcolors=blue]{hyperref}

\usepackage{amsmath}
\usepackage{amsthm}
\usepackage{amssymb}
\usepackage{amsfonts}
\usepackage{bbm}
\usepackage{bm}
\usepackage{nicefrac}
\usepackage{mathtools}

\usepackage{tikz}

\usepackage{graphicx}
\usepackage{fullpage}
\usepackage{float}
\usepackage{wrapfig}
\usepackage{caption}

\usepackage{algpseudocode}
\usepackage{algorithm}
\usepackage{listings}
\usepackage{enumitem}

\usepackage{amsthm}
\usepackage{dsfont}
\usepackage{array}
\usepackage{mathrsfs}
\usepackage{comment}
\usepackage{mathrsfs}
\usepackage{physics}
\usepackage{relsize}
\usepackage{enumitem}
\usepackage{xcolor}

\makeatother

\usepackage{babel}
\usepackage{color}
\usepackage{centernot}

\usepackage{babel}
\usepackage{sansmath}

\usepackage[
backend=bibtex,
style=alphabetic,
citestyle=alphabetic,
maxbibnames=10
]{biblatex}
\usepackage{hyperref}
\hypersetup{
	colorlinks=true,
	linkcolor=black,
	filecolor=blue,
	citecolor = blue,      
	urlcolor=cyan,
}

\usepackage[nameinlink]{cleveref}

\newcommand{\mmse}{\mathrm{mmse}}

\newcommand{\E}{\mathbb{E}}
\newcommand{\Xstar}{\mathsf{X}^\star}
\newcommand{\D}{\mathsf{D}}
\newcommand{\Ls}{\mathsf{L}}
\newcommand{\Ps}{\mathsf{P}}
\newcommand{\Es}{\mathsf{E}}
\newcommand{\Fs}{\mathsf{F}}

\newcommand{\df}{\mathfrak{d}}

\newcommand{\Dbar}{{\mathsf{D}}}
\newcommand{\DDbar}{{D}^\top {D}}
\newcommand{\DD}{{D}^\top {D}}

\newcommand{\R}{\mathbb{R}}
\newcommand{\As}{\mathsf{A}}
\newcommand{\Bs}{\mathsf{B}}
\newcommand{\Gs}{\mathsf{G}}
\newcommand{\Qs}{\mathsf{Q}}
\newcommand{\Rs}{\mathsf{R}}
\newcommand{\Hs}{\mathsf{H}}
\newcommand{\Us}{\mathsf{U}}
\newcommand{\Vs}{\mathsf{V}}
\newcommand{\Xs}{\mathsf{X}}
\newcommand{\Ys}{\mathsf{Y}}
\newcommand{\Zs}{\mathsf{Z}}
\newcommand{\Ss}{\mathsf{S}}
\newcommand{\Fc}{\mathcal{F}}
\newcommand{\Hc}{\mathcal{H}}
\newcommand{\Ac}{\mathcal{A}}
\newcommand{\Bc}{\mathcal{B}}
\newcommand{\Kc}{\mathcal{K}}
\newcommand{\Oc}{\mathcal{O}}
\newcommand{\Ae}{\upsilon^A}
\newcommand{\Be}{\upsilon^B}
\newcommand{\Xe}{\mathfrak{X}^e}
\newcommand{\fe}{\mathfrak{f}^e}
\newcommand{\fet}{\mathfrak{f}^e}
\newcommand{\tmr}{\zeta}

\newcommand{\Cst}{\mathcal{V}}
\newcommand{\Cstt}{\mathcal{V}}
\newcommand{\x}{\tilde{\sigma}}
\newcommand{\s}{\sigma}
\newcommand{\z}{\tilde{\tau}}
\newcommand{\ta}{\tau}
\newcommand{\pib}{\mathfrak{C}}
\newcommand{\pil}{\mathfrak{C}}

\newcommand{\Dcp}{\mathcal{D}_+}
\newcommand{\Tmr}{\mathfrak{Z}}
\newcommand{\st}{\beta^\star}

\newcommand{\V}{\mathbb{V}}

\newcommand{\Z}{\mathcal{Z}}
\newcommand{\Aft}{\mathcal{A}}

\newcommand{\e}{\mathfrak{e}}
\newcommand{\pfrak}{\mathfrak{p}}

\newcommand{\hf}{\mathfrak{h}}

\newcommand{\cpit}{c_{\pi}}
\newcommand{\Ht}{\mathcal{H}}

\newcommand{\Ft}{\mathcal{F}}

\newcommand{\toW}{\stackrel{W}{\to}}

\newcommand{\normop}[1]{\norm{#1}_{\mathrm{op}}}
\newcommand{\Vpi}{\rho_*}
\newcommand{\Tone}{\mathrm{I}}
\newcommand{\Ttwo}{\mathrm{II}}
\newcommand{\Tthree}{\mathrm{III}}
\newcommand{\Tfour}{\mathrm{IV}}
\newcommand{\Tfive}{\mathrm{V}}
\newcommand{\Tsix}{\mathrm{VI}}
\newcommand{\Fpe}{\chi^A}
\newcommand{\Fpnu}{\chi^B}
\newcommand{\Fpom}{\chi^C}
\newcommand{\Fpet}{\chi^A}
\newcommand{\Fpnut}{\chi^{B}}
\newcommand{\Fpomt}{\chi^{C}}
\newcommand{\Fpnuta}{\chi^{B}_1}
\newcommand{\Fpomta}{\chi^{C}_1}
\newcommand{\Fpnutb}{\chi^{B}_2}
\newcommand{\Fpomtb}{\chi^{C}_2}
\renewcommand{\Tr}{\operatorname{Tr}}
\newcommand{\cE}{\mathcal{E}}
\newcommand{\Pfrak}{\mathfrak{P}}
\newcommand{\Qfrak}{\mathfrak{Q}}
\newcommand{\Rfrak}{\mathfrak{R}}
\newcommand{\cU}{\mathcal{U}}
\newcommand{\rc}{\mathrm{rescaled}}

\newtheorem{Theorem}{Theorem}[section]
\newtheorem{Proposition}[Theorem]{Proposition}
\newtheorem{Lemma}[Theorem]{Lemma}
\newtheorem{Corollary}[Theorem]{Corollary}

\theoremstyle{definition}
\newtheorem{Assumption}[Theorem]{Assumption}
\newtheorem{Definition}[Theorem]{Definition}
\newtheorem{Remark}[Theorem]{Remark}
\newtheorem*{Remark*}{Remark}

\Crefname{Theorem}{Theorem}{Theorem}
\Crefname{Proposition}{Proposition}{Proposition}
\Crefname{Lemma}{Lemma}{Lemma}
\Crefname{Corollary}{Corollary}{Corollary}
\Crefname{Assumption}{Assumption}{Assumption}
\Crefname{Remark}{Remark}{Remark}
\Crefname{Notation}{Notation}{Notation}
\Crefname{Definition}{Definition}{Definition}

\def\O{\mathbb{O}}
\def\SO{\mathbb{SO}}
\def\cG{\mathcal{G}}
\def\1{\mathbf{1}}
\DeclareMathOperator{\Haar}{Haar}
\DeclareMathOperator{\diag}{diag}

\title{Random linear estimation with rotationally-invariant designs: Asymptotics
 at high temperature}
\author[2]{Yufan Li\thanks{yufan\_li@g.harvard.edu}}
\author[1]{Zhou Fan\thanks{zhou.fan@yale.edu}}
\author[2]{Subhabrata Sen\thanks{subhabratasen@fas.harvard.edu}}
\author[1]{Yihong Wu\thanks{yihong.wu@yale.edu}}
\affil[1]{Department of Statistics and Data Science, Yale University}
\affil[2]{Department of Statistics, Harvard University}

\date{}

\begin{document}
\maketitle
\begin{abstract}
We study estimation in the linear model $y=A\st+\epsilon$, in a Bayesian
setting where $\st$ has an entrywise i.i.d.\ prior and the design $A$ is
rotationally-invariant in law. In the large system limit as dimension and
sample size increase proportionally, a set of related conjectures have been
postulated for the asymptotic mutual information, Bayes-optimal
mean squared error, and TAP mean-field equations that characterize the Bayes
posterior mean of $\st$. In this work, we prove these conjectures for a
general class of signal priors and for arbitrary
rotationally-invariant designs $A$, under a ``high-temperature'' condition that
restricts the range of eigenvalues of $A^\top A$. Our proof uses a conditional
second-moment method argument, where we condition on the iterates of a
version of the Vector AMP algorithm for solving the TAP mean-field equations.
\end{abstract}

\section{Introduction}

Consider observations $y=A\st+\epsilon \in \R^m$ from a linear model
with Gaussian noise, in a Bayesian setting where the entries of $\st \in \R^n$
are drawn i.i.d.\ from a ``signal prior''. Fundamental questions of
interest in applications spanning CDMA communication systems
\cite{tanaka2002statistical} to sparse signal recovery \cite{baron2009bayesian}
to statistical genetics \cite{guan2011bayesian} pertain to the properties of
the Bayes posterior law and posterior mean estimate for $\st$.

In the asymptotic limit as $m,n \to \infty$ and $A$ constitutes an i.i.d.\
measurement design, a rich and insightful body of literature has obtained
precise ``single-letter'' characterizations of the asymptotic mutual
information, minimum mean squared error (MMSE), and low-dimensional marginals
of the Bayes posterior law. Based initially on work of Tanaka
\cite{tanaka2002statistical} and Guo and Verd\'u \cite{guo2005randomly} using
the non-rigorous replica method of statistical physics, these characterizations
have since been proven rigorously in increasingly general contexts
\cite{montanari2006analysis,barbier2016mutual,reeves2016replica,barbier2019optimal,barbier2019adaptive,barbier2020mutual,qiu2022tap},
and are closely connected to Approximate Message Passing (AMP) algorithms
and mean-field variational approaches for Bayesian inference.

The focus of our work, and the subject of significant recent attention, is on
extensions of these results beyond i.i.d.\ designs. We study here the model
where $A$ is right-rotationally invariant in law, and where analogous
single-letter characterizations are expected to depend on $A$ only via
the spectral distribution of $A^\top A$. In this model, conjectures for
the asymptotic mutual information were derived via the replica method
for binary and Bernoulli-Gaussian signal priors by Takeda, Uda, and
Kabashima \cite{takeda2006analysis} and Tulino et
al.\ \cite{tulino2013support}, and a form for general priors was stated by
Barbier et al.\ in \cite{barbier2018mutual}. A number of iterative
Bayesian inference algorithms including Adaptive TAP
\cite{opper2001adaptive,opper2001tractable}, Expectation-Consistency
\cite{opper2005expectation}, Vector/Orthogonal AMP
\cite{ma2017orthogonal,takeuchi2017rigorous,rangan2019vector}, and long-memory
forms of AMP \cite{takeuchi2021bayes,liu2022memory} have been proposed for this
model, whose algorithmic fixed-points coincide with the
replica predictions. In \cite{maillard2019high}, the forms of the TAP mean-field
equations that characterize the posterior mean were derived for this and related
models using a Plefka expansion approach, and it was argued that the
approximations underlying these AdaTAP, EC, and VAMP/OAMP algorithms are all
equivalent to the vanishing of certain diagrammatic terms in the Plefka
expansion.

In this work, we provide a rigorous proof of the expressions for asymptotic
mutual information and MMSE and of the validity of the TAP mean-field equations
(in an $L^2$ sense) that are predicted by this replica theory, for general
rotationally-invariant designs $A$ under a restriction for the range of
eigenvalues of $A^\top A$. The centered matrix $A^\top A-d_*I$ for a
constant $d_*>0$ plays the role of a
rotationally-invariant couplings matrix in analogous models of
mean-field spin glasses \cite{marinari1994replica,parisi1995mean}, and our
restriction on the eigenvalue range is analogous to an assumption of high
temperature in such spin glass models.
Our results
are complementary to those of \cite{barbier2018mutual} that established the
asymptotic mutual information for specific designs of the form $A=BW$ where $W$
has i.i.d.\ Gaussian entries, without a high-temperature constraint.
The method of \cite{barbier2018mutual} used an adaptive interpolation method
\cite{guerra2003broken,korada2010tight,barbier2019adaptive} specific to this
factorized form, which is different from our approach.

Our proof here instead builds upon recent analyses of orthogonally-invariant
spin glasses \cite{fan2021replica,fan2022tap} via a
conditional second-moment argument \cite{bolthausen2018morita,ding2019capacity}.
We analyze the first and second moments of restrictions of the log-partition
function conditioned on iterates of a version of Vector AMP, to establish the
asymptotic mutual information and Bayes-optimal MMSE. At a technical level,
this extends analyses of \cite{fan2021replica} to encompass a general
class of prior distributions with possibly unbounded support, and to address
additional complexities of the Hamiltonian in the linear model. We
remark that in this Bayesian model, an upper bound for the asymptotic mutual
information at high temperature may alternatively be derived by combining an
analysis of the prediction mean squared error achieved by VAMP with an
``area argument'' \cite{montanari2006analysis,barbier2016mutual} that integrates
the I-MMSE relation \cite{guo2005mutual} from zero SNR to a small positive SNR.
The matching 
lower bound would still require a conditional first moment calculation in our
approach, and we take the opportunity to develop both the conditional first and
second moment analyses so as to provide an independent proof that does not rely
on the area argument. Related ideas of analyzing the conditional second moment
of a truncated log-partition function for a different application to the Ising
perceptron model with unbounded log-activation were developed recently in
\cite{bolthausen2022gardner}.

Interest in the linear model with rotationally-invariant measurement designs
has been partially motivated by the belief that
asymptotic predictions derived for such designs
may hold universally across designs whose right singular vectors are
sufficiently ``generic''. Universality statements of this form have been shown
recently for AMP and other first-order iterative algorithms in
\cite{dudeja2022universality,wang2022universality,dudeja2022spectral}. These results suggest that the asymptotic mutual information and Bayes optimal MMSE are also potentially universal---we leave this as an open question for future work.


\subsection{Model and assumptions}

Let $\st \in \mathbb{R}^{n}$ be a signal vector with coordinates
$(\st_i)_{i=1}^n \overset{iid}{\sim} \pi$ distributed
according to a known prior distribution $\pi$. We observe $m$ noisy
measurements
\begin{equation}\label{eq:linearmodel}
y=A \st+\epsilon \in \R^m
\end{equation}
where $(\epsilon_j)_{j=1}^{m} \overset{iid}{\sim} N(0,1)$ is Gaussian noise
and $A \in \mathbb{R}^{m \times n}$ is the measurement matrix.

Our main results describe the asymptotic mutual information, Bayes-optimal
mean squared error, and Bayes posterior-mean estimator for $\st$, when
$n,m \to \infty$ and $A$ is right-rotationally invariant in law.
We denote the (normalized) mutual information between $\st$ and
$y$ conditioned on $A$ as
\[i_{n}:=\frac{1}{n} I\left(\st ; y \mid A\right)
=\frac{1}{n}\mathbb{E}\left(\log \frac{p\left(y \mid \st, A\right)}{p(y \mid A)
} \;\bigg|\; A\right)\]
We write expectation with respect to the posterior distribution for $\st$
given $(y,A)$ as $\expval{\cdot}$, i.e.
\begin{equation}\label{eq:anglepost}
    \expval{f(\sigma)}:=\frac{\int f(\sigma) \exp (-\frac{1}{2}\|y-A
\s \|^{2}) \prod_{i=1}^{n} d\pi(\s_{i})}{\int \exp
(-\frac{1}{2}\|y-A \s \|^{2}) \prod_{i=1}^{n} d\pi(\s_{i})}
\end{equation}
where we will use $\s$ as the variable for a sample from this
posterior. In particular, $\expval{\s}$ is the posterior mean of
$\beta^*$. We denote its normalized mean squared error conditioned on $A$ as
\begin{equation}
    \mmse_n:=\frac{1}{n} \E [\norm{\st-\expval{\sigma}}^2\mid A].
\end{equation}

We fix a random variable $\D \geq 0$
representing the limit singular value distribution of $A$, and denote throughout
\[d_*:=\E[\D^2], \qquad
d_-:=\min(x:x \in \operatorname{supp}(\D^2)),
\qquad d_+:=\max(x:x \in \operatorname{supp}(\D^2))\]
where $\operatorname{supp}(\D^2) \subseteq [0,\infty)$ is the support of $\D^2$.
\begin{Assumption}[Singular value distribution]\label{AssumpD}
$\D^2$ has strictly positive mean and variance and compact support.
\end{Assumption}
\begin{Assumption}[Measurement matrix]\label{AssumpA}
Let $A=Q^{\top} D O$ be the singular value decomposition,
where $Q \in \mathbb{R}^{m \times m}$ and $O \in
\mathbb{R}^{n \times n}$ are orthogonal and $D \in \mathbb{R}^{m \times n}$ is
diagonal. Then $Q,D$ are deterministic, $O,\st,\epsilon$ are mutually
independent, and $O \sim \Haar(\SO(n))$ is uniformly distributed on the special
orthogonal group. As $n,m \to \infty$,
\begin{equation}\label{eq:Assump2D}
D^\top 1_{m\times 1} \stackrel{W}{\rightarrow} \D,
\qquad \min \left(\diag(\DD)\right) \to d_-,
\qquad \max \left(\diag(\DD)\right) \to d_+.
\end{equation}
\end{Assumption}

Here, $D^\top 1_{m \times 1} \stackrel{W}{\rightarrow} \D$ denotes
Wasserstein-$\pfrak$ convergence of the empirical distribution of coordinates
of $D^\top 1_{m \times 1} \in \R^n$ to $\D$
for all orders $\pfrak \geq 1$, and we review properties of this convergence
in Appendix \ref{appendix:Wasserstein}.

We may assume without loss of generality that $\pi$ has mean 0, by subtracting
from $y$ a multiple of $A1_{n \times 1}$.
Our results are then proven under the following additional 
assumptions for $\pi$ and ``high-temperature'' condition for $\D$.

\begin{Assumption}[Prior distribution]\label{AssumpPrior}
Let $\Xstar \sim \pi$. Then $\pi$ is a non-Gaussian distribution with
\[\E[\Xstar]=0, \qquad \Vpi:=\E[{\Xstar}^2]>0.\]
There is a constant $\pib>0$ for which, for any $s>0$,
\begin{equation}\label{eq:subGaussian}
\Vpi \leq \pib, \qquad \mathbb{P}[|\Xstar|>s] \leq 2e^{-s^2/(2\pib)}.
\end{equation}
Furthermore, for any $k \in \{1,2\}$, symmetric $\Gamma \in \R^{k \times k}$
satisfying $\Gamma \prec (4\pib)^{-1}I$, and $z \in \R^k$, denote
\[d\mu(x)=\frac{e^{x^\top \Gamma x+x^\top z} \prod_{i=1}^k d\pi(x_i)}
{\int e^{x^\top \Gamma x+x^\top z} \prod_{i=1}^k d\pi(x_i)},
\quad \langle f(x) \rangle_\mu=\int f(x)d\mu(x),
\quad \V_\mu[f(x)]=\langle f(x)^2 \rangle_\mu-\langle f(x) \rangle_\mu^2.\]
Then the distribution $\mu$ satisfies, for any unit vector $v \in \R^k$
and a constant $C>0$ depending only on $\pib$,
\begin{equation}\label{eq:poincare}
\V_\mu[v^\top x] \leq C, \qquad \V_\mu[(v^\top x)^2] \leq
C[1+\langle (v^\top x)^2\rangle_\mu].
\end{equation}
\end{Assumption}

\begin{Assumption}[High temperature]\label{AssumpHighTemp}
We have $\operatorname{supp}(\D^2) \subseteq [d_*-\e,d_*+\e]$
for some $\e>0$.
\end{Assumption}

The condition (\ref{eq:poincare}) of Assumption \ref{AssumpPrior} may
be understood as a Poincar\'e-type inequality for $\mu$,
and holds (for example) when
$\pi$ has a bounded support contained in $[-\sqrt{\pib},\sqrt{\pib}]$ or
a log-concave density $e^{-g(x)}$ where $g''(x) \geq 1/\pib$,
c.f.\ \Cref{prop:priorconditions}.
We will require the value $\e$ in Assumption \ref{AssumpHighTemp}
to be sufficiently small, depending only on the constant
$\pib$ in Assumption \ref{AssumpPrior}. Note that in a model $y=A\st+\sigma
\epsilon$ with general noise variance $\sigma^2>0$, such a high temperature
assumption encompasses the setting of sufficiently large $\sigma^2$ for any
fixed singular value distribution $\D$ and prior $\pi$, upon rescaling
$y$, $A$, and $\D$ all by $1/\sigma$.

\begin{Remark}
For Gaussian $\pi$, $i_n$, $\mmse_n$, and the posterior mean
$\langle \sigma \rangle$ all have explicit formulas, and our main results in fact hold at any
temperature which follow more
directly from techniques of asymptotic random matrix theory (see e.g.\
\cite[Theorem 2]{tulino2013support} and \cite[Theorem 1]{reeves2017additivity}).
We restrict henceforth to the more difficult setting of non-Gaussian $\pi$,
and this restriction is used to avoid a rank degeneracy in our subsequent
conditioning arguments.
\end{Remark}

\begin{Remark}\label{remark:On}
Our results extend directly to the more commonly studied
setting where $O \sim \Haar(\O(n))$ is
uniform over the full orthogonal group, and also where $Q,D$ are random and
independent of $O,\st,\epsilon$ such that (\ref{eq:Assump2D}) holds almost
surely as $n,m \to \infty$.
We discuss this further in Appendix \ref{appendix:On}.
\end{Remark}

\subsection{Scalar channel and fixed point equation}

The asymptotic characterization of the model (\ref{eq:linearmodel}) is
described by a ``single-letter'' scalar channel
\begin{equation}\label{eq:scalarBayes}
\Ys=\Xstar+\Zs/\sqrt{\gamma}
\end{equation}
with signal $\Xstar \sim \pi$, independent Gaussian noise $\Zs \sim N(0,1)$, and
noise variance $\gamma^{-1}>0$. We denote the Bayes posterior-mean denoiser
in this model as
\begin{equation}\label{eq:deff}
f(y,\gamma)=\E[\Xstar \mid \Ys=y]
\end{equation}
and the signal-observation mutual information and Bayes-optimal
mean squared error as
\begin{equation}\label{eq:scalarimmse}
i(\gamma)=I(\Xstar;\Ys), \qquad \mmse(\gamma)=\E[\V[\Xstar \mid \Ys]].
\end{equation}

Under Assumption \ref{AssumpD}, let $G:(-d_-,\infty) \to (0,\infty)$
and $R:(0,G(-d_-)) \to (-\infty,0)$ be the Cauchy- and
R-transforms of the law of $-\D^2$, defined by
\begin{equation}\label{eq:CauchyR}
G(z)=\E\left[\frac{1}{z+\D^2}\right], \qquad R(z)=G^{-1}(z)-\frac{1}{z},
\end{equation}
where $G^{-1}(\cdot)$ is the functional inverse of $G(\cdot)$, and
we set $G(-d_-)=\lim_{z \to -d_-} G(z)$. Lemma \ref{lem:cauchy} shows
that these functions are well-defined and reviews several additional
properties under the high-temperature condition of
Assumption \ref{AssumpHighTemp}.
The noise variance $\gamma^{-1}$ that relates this scalar channel to the
model (\ref{eq:linearmodel}) is a solution of the fixed-point equations
\begin{equation}\label{eq:fix}
\eta^{-1}=\mmse(\gamma), \qquad \gamma=-{R}(\eta^{-1}).
\end{equation}
The second equation can be written equivalently as $\eta^{-1}=G(\eta-\gamma)$ by the definition of the R-transform. 

The following ensures that this fixed-point system has a unique
solution when $\e$ in Assumption \ref{AssumpHighTemp} is sufficiently small;
see Appendix \ref{append:FPE} for its proof. 
\begin{Proposition}\label{prop:uniquefix}
Under Assumptions \ref{AssumpD}, \ref{AssumpPrior}, and
\ref{AssumpHighTemp}, there exists a constant $\e_0=\e_0(\pib)>0$ such that if
$\e<\e_0$, then (\ref{eq:fix}) has a unique solution
$\left(\eta_*^{-1},\gamma_{*}\right)$ in the domain $(0,G(-d_-)) \times \R_+$.
Furthermore $\eta_*^{-1}\leq \Vpi$ and $\eta_*-\gamma_*\ge \Vpi^{-1}>0$
where $\Vpi$ is the prior variance in Assumption \ref{AssumpPrior}.
\end{Proposition}

\subsection{Main results}

Let us denote $(\eta_*^{-1},\gamma_*)$ as the unique fixed point of
(\ref{eq:fix}). Our first result describes the asymptotic mutual information
in the model (\ref{eq:linearmodel}).
We define the replica symmetric potential following \cite{barbier2018mutual},
\begin{equation}\label{eq:iRS}
i_{\mathrm{RS}}\left(\eta^{-1}, \gamma\right)=i(\gamma)
-\frac{1}{2} \int_{0}^{\eta^{-1}} {R}(z) d z-\frac{\gamma}{2 \eta}
\end{equation}
where $i(\gamma)$ is the above mutual information in the scalar channel.
\begin{Theorem}[Mutual information]\label{thm:maintheorem}
	Under Assumptions \ref{AssumpD}--\ref{AssumpHighTemp}, there exists a constant $\e_0=\e_0(\pib)>0$ such that if
$\e<\e_0$, then almost surely
	\begin{equation}\label{eq:singleletter}
		\lim_{n,m \rightarrow \infty} i_{n}=i_{\mathrm{RS}}\left(\eta_{*}^{-1}, \gamma_{*}\right)
	\end{equation}
\end{Theorem}

\begin{Remark}
Without the high-temperature condition of Assumption \ref{AssumpHighTemp},
the general conjecture \cite{barbier2018mutual} is
that $\lim_{n,m \rightarrow \infty} i_{n}=\inf i_{\mathrm{RS}}(\eta^{-1},
\gamma)$ where the infimum ranges over all $(\eta^{-1}, \gamma)$ that solves
\eqref{eq:fix}.
\end{Remark}

Next, we characterize the limiting minimum mean squared error.

\begin{Theorem}[MMSE]\label{thm:MMSE}
Under Assumptions \ref{AssumpD}--\ref{AssumpHighTemp},
there exists a constant $\e_0=\e_0(\pib)>0$ such that if
$\e<\e_0$, then almost surely
	\begin{equation}\label{eq:mmse}
		\lim_{n,m \rightarrow \infty} \mmse_n=\eta_*^{-1}
	\end{equation}
\end{Theorem}

Finally, we show that the posterior mean $\expval{\sigma}$ for $\st$ in the
model (\ref{eq:linearmodel}) approximately
satisfies a system of mean-field equations predicted by the Plefka expansion
\cite[Eqs.\ (128-129)]{maillard2019high}. These equations are an analogue of the
Thouless-Anderson-Palmer (TAP) equations for the Sherrington-Kirkpatrick
model \cite{thouless1977solution}, and of their generalization to
orthogonally-invariant spin glass models in
\cite{parisi1995mean,opper2001adaptive}.

\begin{Theorem}[TAP equations]\label{thm:TAP}
Under Assumptions \ref{AssumpD}--\ref{AssumpHighTemp}, there exists a constant
$\e_0=\e_0(\pib)>0$ such that if $\e<\e_0$, then almost surely
    \begin{equation}\label{eq:TAPeq}
        \lim_{n,m \rightarrow \infty} \frac{1}{n}
\mathbb{E}\left[\Big\|\langle\sigma\rangle-f\left(-\gamma_{*}^{-1}\left(A^{\top}
A\langle\sigma\rangle-\gamma_{*}\langle\sigma\rangle-A^{\top}
y\right),\;\gamma_*\right)\Big\|^{2} \;\bigg|\; A\right]=0
    \end{equation}
where $f(\cdot,\gamma)$ is the posterior-mean denoiser
from \eqref{eq:deff} applied entrywise to its first argument.
\end{Theorem}

\begin{Remark}\label{remark:scalinginv}
	For the following two special subclasses of priors satisfying Assumption \ref{AssumpPrior}: (i) $\pi$ has a bounded support contained in $[-\sqrt{\pib},\sqrt{\pib}]$ or (ii) $\pi$ admits a log-concave density $e^{-g(x)}$ with $g''(x)\ge 1/\pib$, an explicit choice of $\e_0(\pib)$ in Theorems \ref{thm:maintheorem}, \ref{thm:MMSE} and \ref{thm:TAP} is $\e_0(\pib)=\frac{\mathfrak{a}}{\pib}$ for some absolute constant $\mathfrak{a}>0$. See \Cref{append:EHTC} for a justification. 
\end{Remark}

\subsection*{Notation}
Denote by 
$\|\cdot\|$ the $\ell_{2}$-norm for vectors and $\ell_{2} \rightarrow \ell_{2}$ operator norm for matrices.
For scalars $x_1,\ldots,x_k \in \R$, we write
$(x_{1}, \ldots, x_{k}) \in \mathbb{R}^k$ to denote a column vector
with these entries. For vectors $x_1,\ldots,x_k \in \R^n$, we write
$(x_1,\ldots,x_k) \in \R^{n\times k}$ as the matrix containing these columns. $1_{m\times n}\in \R^{m\times n}$ denotes the all-1's matrix,
$I_{n \times n}$ denotes the identity matrix, $\succ$ and $\succeq$ denote the
positive-definite ordering for matrices,
$\E$ and $\V$ denote the expectation and variance of a random variable,
and $\R_+=(0,\infty)$ is the positive real line.

Throughout, we treat $\pib$ in Assumption \ref{AssumpPrior} as constant, and we
write $x=O(y)$ to mean $|x| \leq Cy$ for a constant $C>0$ depending only on
$\pib$. (In particular, this constant does not depend on $d_*$ or on
the small parameter
$\e$ in Assumption \ref{AssumpHighTemp}, and we will explicitly track the
dependence of various quantities on $d_*,\e$.)

\section{Vector AMP}\label{section:VAMP}

We first review a version of the VAMP algorithm proposed by
\cite{rangan2019vector}. Let $r_2^1 \in \R^n$ be an initialization vector
such that there exist random variables $(\Rs_2^1,\Xstar)$ for which,
almost surely as $n,m \to \infty$,
\begin{equation}\label{eq:VAMPinit}
(r_2^1,\st) \toW (\Rs_2^1,\Xstar),
\qquad \gamma_{2,1}^{-1}:=\E[(\Rs_2^1-\Xstar)^2]>0.
\end{equation}
Define from $\gamma_{2,1}$ a sequence of state evolution parameters
for $t=1,2,3,\ldots$
\begin{equation}\label{eq:SEparams}
\eta_{2,t}=1/G(\gamma_{2,t}), \quad \gamma_{1,t}=\eta_{2,t}-\gamma_{2,t},
\quad \eta_{1,t+1}=1/\mmse(\gamma_{1,t}), \quad
\gamma_{2,t+1} =\eta_{1,t+1}-\gamma_{1,t}
\end{equation}
where $\mmse(\cdot)$ is the scalar channel MMSE from (\ref{eq:scalarimmse})
and $G(\cdot)$ is the Cauchy-transform of $-\D^2$ from (\ref{eq:CauchyR}).
Now consider the sequence of iterates in $\R^n$, for $t=1,2,3,\ldots$
\begin{subequations}\label{eq:iteupdate}
    \begin{align}
	r_1^t&=\frac{1}{\gamma_{1,t}}\left[\eta_{2,t}\left(A^{\top}
A+\gamma_{2,t}I\right)^{-1}\left(A^{\top}
y+\gamma_{2,t}r_2^t\right)-\gamma_{2,t}r_2^t\right]\label{eq:itupdate1}\\
	r_2^{t+1}&=\frac{1}{\gamma_{2,t+1}}\Big[\eta_{1,t+1}
f\left(r_1^t,\gamma_{1,t}\right)-\gamma_{1,t}r_1^t\Big]\label{eq:itupdate2}
\end{align}
\end{subequations}
where $f(\cdot,\gamma)$ is the posterior-mean denoiser from
(\ref{eq:deff}) applied entrywise. This coincides with the VAMP algorithm in
\cite{rangan2019vector}, specialized to the setting with matched MMSE denoiser,
and replacing empirically estimated versions of the parameters
$\gamma_{1,t},\eta_{1,t},\gamma_{2,t},\eta_{2,t}$ with their large system limits
as defined by (\ref{eq:SEparams}).
The following statement is implied by \cite[Theorems 1 and 2]{rangan2019vector};
we check the conditions needed for these results in Appendix \ref{appendix:VAMP}.

\begin{Theorem}[\cite{rangan2019vector}]\label{thm:RanganSE}
Suppose Assumptions \ref{AssumpD}--\ref{AssumpPrior} hold, and $r_2^1$
is independent of $(A,\epsilon)$ and satisfies (\ref{eq:VAMPinit}).
Then each value $\eta_{2,t},\gamma_{1,t},\eta_{1,t+1},\gamma_{2,t+1}$
for $t \geq 1$ is well-defined by (\ref{eq:SEparams}) and strictly positive.
For any 2-pseudo-Lipschitz test function $g:\R^2 \to \R$ and each fixed $t \geq
1$, almost surely
\begin{equation}\label{eq:RanganSE}
\lim_{n,m \to \infty} \frac{1}{n}\sum_{i=1}^n g((r_1^t)_i,\st_i)
=\E[g(\mathsf{R}_1^t,\Xstar)]
\end{equation}
where $\mathsf{R}_1^t=\Xstar+\Zs/\sqrt{\gamma_{1,t}}$ and $\Zs \sim N(0,1)$ is
independent of $\Xstar$. Furthermore, set
\[\hat{\beta}_2^t=\left(A^{\top}A+\gamma_{2,t}I\right)^{-1}\left(A^{\top}
y+\gamma_{2,t}r_2^t\right), \qquad
\hat{\beta}_1^{t+1}=f(r_1^t,\gamma_{1,t}).\]
Then for each fixed $t \geq 1$, almost surely
\begin{equation}\label{eq:RanganMSE}
\lim_{n,m \to \infty} \frac{1}{n}\|\hat{\beta}_2^t-\st\|^2=\eta_{2,t}^{-1},
\qquad \lim_{n,m \to \infty}
\frac{1}{n}\|\hat{\beta}_1^{t+1}-\st\|^2=\eta_{1,t+1}^{-1}.
\end{equation}
\end{Theorem}

Our proofs will use an extended state evolution for a version of this
algorithm in a reparametrized form. Letting $(\eta_*^{-1},\gamma_*)$ be a
fixed point of (\ref{eq:fix}), it is computationally convenient to specialize
to a ``stationary'' initialization of this algorithm given by
\begin{equation}\label{eq:VAMPstatinit}
r_1^0=\st+p^0, \qquad \gamma_{1,0}=\gamma_*
\end{equation}
where $(p_i^0)_{i=1}^n \overset{iid}{\sim} N(0,\gamma_*^{-1})$ is
independent of all other randomness in the model. (The quantities
$r_2^1,\gamma_{2,1}$ in (\ref{eq:VAMPinit}) are then defined from this
initialization by (\ref{eq:SEparams}) and (\ref{eq:iteupdate}).)
In the following results, we reparametrize the algorithm
initialized by (\ref{eq:VAMPstatinit}) and
describe its state evolution; proofs are
deferred to Appendix \ref{appendix:VAMP}.

Set
\[\Lambda=\frac{\eta_*-\gamma_*}{\gamma_*}
\Big[\eta_*(D^\top D+(\eta_*-\gamma_*)I)^{-1}-I\Big] \in \R^{n \times n},\]
\begin{equation}\label{eq:Lambdadef}
\xi=Q\epsilon, \qquad
e_b=\frac{\eta_*}{\gamma_*}(D^\top D+(\eta_*-\gamma_*)I)^{-1} D^\top \xi,
\qquad e=O^\top e_b,
\end{equation}
and define from (\ref{eq:iteupdate}) the new variables in $\R^n$
\[x^t=r_2^t-\st, \qquad y^t=r_1^t-e-\st.\]
From the posterior-mean function $f(\cdot,\gamma)$ in (\ref{eq:deff}),
define $F:\R^2 \to \R$ as
\begin{equation}\label{eq:Fdef}
F(p,\beta)=\frac{\eta_*}{\eta_*-\gamma_*}
f(p+\beta,\gamma_*)-\frac{\gamma_*}{\eta_*-\gamma_*}p
-\frac{\eta_*}{\eta_*-\gamma_*}\beta.
\end{equation}

\begin{Proposition}\label{prop:VAMPequiv}
If $(\eta_*^{-1},\gamma_*)$ is a fixed point of
(\ref{eq:fix}) and $\gamma_{1,0}=\gamma_*$, then
$\eta_{1,t}=\eta_{2,t}=\eta_*$, $\gamma_{1,t}=\gamma_*$, and
$\gamma_{2,t}=\eta_*-\gamma_*$ for all $t \geq 1$. Furthermore,
the VAMP algorithm (\ref{eq:iteupdate}) initialized with
(\ref{eq:VAMPstatinit}) is equivalent to the initialization
$x^1=F(p^0,\st)$ and the iterations, for $t=1,2,3,\ldots$
\begin{equation}\label{eq:xsy}
s^t=Ox^t, \qquad y^t=O^\top \Lambda s^t, \qquad x^{t+1}=F(y^t+e,\st).
\end{equation}
\end{Proposition}

We define the important scalar parameters
\begin{equation}\label{eq:scalarparams}
\delta_*=\frac{1}{\eta_*-\gamma_*},
\quad \kappa_*=\left(\frac{\eta_*-\gamma_*}{\eta_*}\right)^2
\left[\E\frac{\eta_*^2}{(\D^2+\eta_*-\gamma_*)^2}-1\right],
\quad \sigma_*^2=\delta_*\kappa_*,
\quad b_*=\frac{1}{\gamma_*}-\frac{\kappa_*}{\eta_*-\gamma_*}.
\end{equation}
Let us collect
\begin{equation}\label{eq:Hdef}
H=\Big(\st,D^\top 1_{m \times 1}, D^\top \xi,
\diag(\Lambda),e_b,e,p^0\Big) \in \R^{n \times 7}.
\end{equation}
The following results describe the limit empirical distribution of
these quantities constituting $H$, as well as the 
state evolution of the iterates $x^t,s^t,y^t$ defined by (\ref{eq:xsy}).

\begin{Proposition}\label{prop:AMPparamconverge}
Suppose Assumptions \ref{AssumpD}--\ref{AssumpPrior} hold.
Define random variables
\[\Xi\sim N(0,1), \qquad \Xstar \sim \pi, \qquad
\mathsf{P}_0\sim N(0,\gamma_*^{-1}), \qquad \Es \sim N(0,b_*)
\]
independent of each other and of $\D$, and set
	$$
\mathsf{L}=\frac{\eta_{*}-\gamma_{*}}{\gamma_{*}}\left(\frac{\eta_*}{\Dbar^{2}+\eta_{*}-\gamma_{*}}-1\right),\qquad
\mathsf{E}_{b}=\frac{\eta_{*}}{\gamma_{*}} \frac{\Dbar
\Xi}{\Dbar^{2}+\eta_{*}-\gamma_{*}}, \qquad
\Hs=(\Xstar,\D,\D\Xi,\Ls,\Es_b,\Es,\Ps_0).$$
Then $\kappa_*=\E \Ls^2$ and $b_*=\E \Es_b^2$. Furthermore,
$H \stackrel{W}{\rightarrow} \Hs$ almost surely as $n,m\to \infty$.
\end{Proposition}

\begin{Theorem}\label{thm:ampSE}
Suppose Assumptions \ref{AssumpD}--\ref{AssumpPrior} hold.
Let $\Hs=(\Xstar,\D,\D\Xi,\Ls,\Es_b,\Es,\Ps_0)$ be as defined in
\Cref{prop:AMPparamconverge}. Set $\Xs_1=F(\Ps_0,\Xstar)$ for the function
$F(\cdot)$ from (\ref{eq:Fdef}), set
$\Delta_1=\E[\Xs_1^2] \in \R^{1 \times 1}$, and define iteratively
$\Ss_t,\Ys_t,\Xs_{t+1},\Delta_{t+1}$ for $t=1,2,3,\ldots$ such that
\[(\Ss_1,\ldots,\Ss_t) \sim N(0,\Delta_t),
\qquad (\Ys_1,\ldots,\Ys_t) \sim N(0,\kappa_* \Delta_t)\]
are Gaussian vectors independent of each other and of $\Hs$, and
\[\Xs_{t+1}=F(\Ys_t+\Es,\Xstar), \qquad
\Delta_{t+1}=\mathbb{E}\left[\left(\mathsf{X}_{1}, \ldots,
\mathsf{X}_{t+1}\right)\left(\mathsf{X}_{1}, \ldots,
\mathsf{X}_{t+1}\right)^{\top}\right] \in \R^{(t+1) \times (t+1)}.\]
Then for each $t \geq 1$, $\Delta_t \succ 0$ strictly,
$\delta_*=\E \Xs_t^2$, and $\sigma_*^2=\E \Ys_t^2$.

Furthermore, let $X_{t}=\left(x^{1}, \ldots, x^{t}\right) \in
\mathbb{R}^{n \times t}$, $S_{t}=\left(s^{1}, \ldots, s^{t}\right) \in
\mathbb{R}^{n \times t}$, and $Y_{t}=\left(y^{1}, \ldots, y^{t}\right) \in
\mathbb{R}^{n \times t}$ collect the iterates of \eqref{eq:xsy}, starting from
the initialization $x^1=F(p^0,\st)$. Then for any fixed $t \geq 1$,
almost surely as $n,m \rightarrow \infty$,
	$$
	\left(H, X_t, S_{t}, Y_t\right)
\stackrel{W}{\rightarrow}\left(\mathsf{H},
\mathsf{X}_{1}, \ldots, \mathsf{X}_t, 
\mathsf{S}_{1}, \ldots, \mathsf{S}_{t}, \mathsf{Y}_{1}, \ldots,
\mathsf{Y}_{t}\right).
	$$ 
	\end{Theorem}

\begin{Corollary}\label{cor:ampSEcor1}
In the setting of Theorem \ref{thm:ampSE}, for any fixed $t \geq 1$, almost
surely
\[\lim_{n,m \to \infty} n^{-1}(e,X_t,Y_t)^\top (e,X_t,Y_t)
=\begin{pmatrix} b_* & 0 & 0 \\ 0 & \Delta_t & 0 \\
0 & 0 & \kappa_*\Delta_t \end{pmatrix}\]
\end{Corollary}

Theorem \ref{thm:ampSE} implies
that the joint limit $(\Ss_1,\ldots,\Ss_t)$ for the
iterates $S_t=(s^1,\ldots,s^t)$ is independent of $\D$.
We highlight here the following implication, which is an analogue of
\cite[Proposition 2.4]{fan2021replica}.

\begin{Corollary}\label{cor:ampSEcor2}
In the setting of Theorem \ref{thm:ampSE},
fix any $t \geq 1$, let $f: \mathbb{R} \rightarrow \mathbb{R}$ be any
function which is continuous and bounded in a neighborhood of
$\operatorname{supp}(\D^2)$, and define $f(\DDbar) \in \R^{n \times n}$ by the
functional calculus. Then almost surely
	$$
	\lim_{n,m \to \infty} n^{-1} S_{t}^{\top} f(\DDbar ) S_{t}=\Delta_{t} \cdot \E f(\Dbar^2).
	$$
\end{Corollary}

Noting that each matrix $\Delta_t$ is the upper-left submatrix of
$\Delta_{t+1}$, let us denote the entries of these matrices as
$\Delta_t=(\delta_{rs})_{r,s=1}^t$. Theorem \ref{thm:ampSE} ensures that
$\delta_{tt}=\delta_*$ for all $t \geq 1$.
The following result then guarantees that for sufficiently small $\e$ in
Assumption \ref{AssumpHighTemp}, the state evolution of this
stationary VAMP algorithm is convergent in the sense
\begin{equation*}
	\begin{aligned}
		\lim _{\min (s, t) \rightarrow \infty}\left(\lim_{n,m \rightarrow \infty} \frac{1}{n}\left\|x^{t}-x^{s}\right\|^{2}\right) &=\lim _{\min (s, t) \rightarrow \infty}\left(\delta_{s s}+\delta_{t t}-2 \delta_{s t}\right)=0 \\
		\lim _{\min (s, t) \rightarrow \infty}\left(\lim_{n,m \rightarrow \infty} \frac{1}{n}\left\|y^{t}-y^{s}\right\|^{2}\right) &=\lim _{\min (s, t) \rightarrow \infty} \kappa_{*}\left(\delta_{s s}+\delta_{t t}-2 \delta_{s t}\right)=0.
	\end{aligned}
\end{equation*}

\begin{Proposition}\label{prop:convsmallbeta}
Under Assumptions \ref{AssumpD}--\ref{AssumpHighTemp},
there exists some constant
$\e_0=\e_0(\pib)>0$ such that if $\e<\e_0$, then
$\lim_{\min (s, t) \rightarrow \infty} \delta_{s t}=\delta_*$.
\end{Proposition}
\noindent We verify in Appendix \ref{appendix:VAMPconvergence} that,
outside the high-temperature regime of Assumption \ref{AssumpHighTemp}, this
statement of Proposition \ref{prop:convsmallbeta} continues to hold for the
stationary initialization of VAMP defined by any fixed point
$(\eta_*^{-1},\gamma_*)$ to (\ref{eq:fix}) that is a local minimizer of the
replica-symmetric potential (\ref{eq:iRS}).

\section{Analysis of the restricted partition function}

For any subset $\cU \subseteq [0,\infty)$, define a restricted partition
function for the model (\ref{eq:linearmodel}) as 
\begin{equation}\label{eq:defpartitionfunc}
\begin{aligned}
\Z(\cU) &= \int \mathbb{I}\bigg(\frac{1}{n}\|\sigma-\st\|^2 \in \cU\bigg)
\cdot \exp \bigg(-\frac{\left\|y- A \s \right\|^{2}}{2}\bigg)
\prod_{i=1}^{n} d\pi\left(\s_{i}\right).
	\end{aligned}
\end{equation}
We will ultimately analyze the unrestricted partition function
$\Z=\Z([0,\infty))$, although it is technically convenient to first analyze
$\Z(\cU)$ for bounded subsets $\cU$.

For any $t \geq 1$,
define the sigma-field (in the probability space of $O$, $\st$, and
$\epsilon$)
\begin{equation}\label{eq:defGt}
	\mathcal{G}_{t}=\mathcal{G}\left(H, x^{1}, s^{1}, y^{1}, \ldots, x^{t}, s^{t}, y^{t}\right)
\end{equation}
where $H$ consists of the quantities in (\ref{eq:Hdef}),
and $x^t,s^t,y^t$ are
the VAMP iterates of (\ref{eq:xsy}). In this section,
we provide asymptotic variational characterizations of the conditional first
and second moments $\E[\Z(\cU) \mid \cG_t]$ and $\E[\Z(\cU)^2 \mid \cG_t]$.
Together with a concentration inequality for $\log \Z(\cU)$
and a second-moment argument, these establish an unconditional
first-order limit for $\log \Z(\cU)$. Proofs of
these results are provided in Appendices \ref{appendix:variational},
\ref{appendix:analysis}, and \ref{appendix:concentration}.

For $a,b \in \R$ and $M>0$, define
\begin{equation}\label{eq:cpi}
c_{\pi}(a,b)=\int \exp \left(a x^{2}+b x\right) d\pi(x) \in (0,\infty],
\qquad c_\pi^M(a,b)=\int_{-M}^M \exp \left(a x^{2}+b x\right) d\pi(x) \in
(0,\infty).
\end{equation}
When $\pi$ has unbounded support, $c_\pi(a,b)$ may be infinite for large
positive values of $a$, and we discuss its behavior in Lemma \ref{lem:abar}.
Under Assumption \ref{AssumpHighTemp},
recall the unique fixed point $(\eta_*^{-1},\gamma_*)$ of (\ref{eq:fix})
and the prior variance $\Vpi$ from Assumption \ref{AssumpPrior},
and define the replica-symmetric free energy
	\begin{equation}\label{eq:RSdef}
		\Psi_{\mathrm{RS}}=-\frac{1}{2}-\frac{\gamma_*\rho_*}{2}
+\frac{\gamma_{*}}{2\eta_{*}}+\frac{1}{2}
\int_{0}^{\eta_{*}^{-1}} {R}(z) d z+\mathbb{E} \log
c_{\pi}\left(-\frac{\gamma_{*}}{2}, \gamma_{*} \Xstar+\sqrt{\gamma_{*}}
\Zs\right)
	\end{equation}
where the expectation is over independent variables
$\Xstar \sim \pi$ and $\Zs \sim N(0,1)$. In addition to $d_*,\Vpi$
and the variance parameters
$\delta_*,\kappa_*,\sigma_*^2,b_*$ from (\ref{eq:scalarparams}),
we introduce the auxiliary scalar parameters
\[a_*=\left(\eta_{*}-\gamma_{*}\right)\left(1-\frac{d_*}{\gamma_{*}}\right), \qquad
c_*=-\left(\eta_{*}-\gamma_{*}\right)
\kappa_{*}+\left(\frac{\eta_{*}-\gamma_{*}}{\gamma_{*}}\right)^{2}\left(d_*-\gamma_{*}\right),
\qquad e_*=1+\kappa_*\]
\begin{equation}\label{eq:auxparams}
\alpha_{*}^{A}=\frac{\eta_{*}-\gamma_{*}}{\gamma_{*}}
\frac{1}{\sqrt{\kappa_{*}}},\qquad
\alpha_{*}^{B}=\left(\alpha_{*}^{A}\right)^{2}\left(\gamma_{*}-d_*\right),
\qquad \pi_{*}=\E \Xstar F\left(\frac{\Zs}{\sqrt{\gamma_*}},\,\Xstar\right)
\end{equation}

\subsection{Conditional first moment}

Fix an iteration $t \geq 1$ for VAMP, and define tuples of primal and
dual variables
\[\Pfrak=(u,r,v,w), \qquad \Qfrak=(\tmr,U,R,V,W,\chi^A,\chi^B,\chi^C)\]
where $u>0$ and $\tmr>-d_-$ and $r,U,R,\chi^A,\chi^B,\chi^C \in \R$ and
$v,w,V,W\in \R^t$. Define
\begin{equation}\label{eq:defCst}
\Ac(\Pfrak)=u-r^2-\|v\|^2-\|w\|^2, \qquad
\Cst=\big\{\Pfrak:\Ac(\Pfrak)>0 \big\}.
\end{equation}
Define also the functions
\begin{equation}\label{eq:defFH}
\begin{aligned}
&\Hc(\tmr,\Ac)=\tmr \Ac-\E[\log (\tmr+\D^2)]-(1+\log \Ac), \qquad
\Bc(v,w)=\|v-\alpha_*^A w\|^2\\
&\lambda(x)=\frac{\eta_*-\gamma_*}{\gamma_*}\left(
\frac{\eta_{*}}{x+\eta_{*}-\gamma_{*}}-1\right),\qquad
\theta(x)=x-\frac{\alpha_{*}^{B}
\eta_{*}}{x+\eta_{*}-\gamma_{*}}+\alpha_{*}^{B}-d_*.
\end{aligned}
\end{equation}
Let $\Delta_t$ and the random variables
$(\mathsf{E}, \Xstar,\Xs_1,...,\Xs_t, \Ys_1,...,\Ys_t)$ be as
described in \Cref{prop:AMPparamconverge} and \Cref{thm:ampSE}, and define
{\small
\begin{equation}\label{eq:defPhi1t}
\begin{aligned}
\Phi_{1,t}(\Pfrak,\Qfrak)
&=-\frac{1}{2}+\E \log c_{\pi}\left(U, -2U\,\Xstar+\frac{R\,\mathsf{E}}{\sqrt{b_*}}+V^{\top}
\Delta_t^{-1/2}\left(\Xs_{1}, \ldots, \Xs_{t}\right)+\frac{ W^{\top}
\Delta_t^{-1/2}\left(\Ys_{1}, \ldots, \Ys_{t}\right)}{\sqrt{\kappa_{*}}}\right)\\
&\quad-(u-\rho_*)U-r\,R-(v+\pi_*\Delta_t^{-1/2}1_{t \times 1})^\top V-w^\top W\\
&\quad-\frac{1}{2}
\qty(\frac{e_*r^2}{b_*}-\frac{2r}{\sqrt{b_*}}+
\Tr \left[\begin{pmatrix} d_* & a_* \\ a_* &
c_* \end{pmatrix}\left(v,\frac{w}{\sqrt{\kappa_*}}\right)^\top
\left(v,\frac{w}{\sqrt{\kappa_*}}\right) \right])
+\frac{1}{2} \Hc(\tmr,\Ac(\Pfrak))\\
&\quad+\frac{1}{2}\mathbb{E}\left[\frac{\mathsf{E}_{b}^{2}}{\tmr+\D^2}\left(\left[\frac{\gamma_{*}}{\eta_{*}}-\frac{r}{\sqrt{b_*}}\right] \D^2-\Fpe\right)^{2}\right]
+\frac{1}{2}\mathbb{E}\left[\frac{1}{\tmr+\D^2}\left(\theta(\D^2)-\lambda(\D^2)
\Fpnu -\Fpom\right)^{2}\right]\Bc(v,w)
\end{aligned}
\end{equation}}
Let $\Phi_{1,t}^M(\Pfrak,\Qfrak)$ have the same definition with $c_\pi$ replaced
by $c_\pi^M$. Finally, define
\begin{equation}\label{eq:defPsi1t}
\Psi_{1,t}(\Pfrak)=\inf_{\Qfrak:\tmr>-d_-} \Phi_{1,t}(\Pfrak,\Qfrak),
\qquad \Psi_{1,t}^M(\Pfrak)=\inf_{\Qfrak:\tmr>-d_-} \Phi_{1,t}^M(\Pfrak,\Qfrak)
\end{equation}
where these may take extended real values in $[-\infty,\infty)$.

\begin{Lemma}\label{lemma:firstmt}
Fix any $K>0$. Under Assumptions \ref{AssumpD}--\ref{AssumpHighTemp},
there exists a constant $\e_0=\e_0(\pib,K)>0$ such that if $\e<\e_0$, then
for any fixed $t\ge 1$ and non-empty open subset $\cU \subseteq (0,K)$,
almost surely
	\begin{equation*}
\liminf_{n,m \rightarrow \infty} \frac{1}{n} \log
\mathbb{E}\left[\Z(\cU) \mid \mathcal{G}_{t}\right] \geq
\sup_{\Pfrak \in \Cst:\,u \in \cU} \sup_{M>0} \Psi_{1,t}^M(\Pfrak),
	\qquad \limsup_{n,m \rightarrow \infty} \frac{1}{n} \log
\mathbb{E}\left[\Z(\overline{\cU}) \mid \mathcal{G}_{t}\right]\leq
\sup_{\Pfrak \in \Cst:\,u \in \cU} \Psi_{1,t}(\Pfrak)
	\end{equation*}
where $\overline{\cU}$ is the closure of $\cU$.
\end{Lemma}

\begin{Lemma}\label{lemma:analysisf}
Fix any $K>0$. Under Assumptions \ref{AssumpD}--\ref{AssumpHighTemp},
there exists a constant $\e_0=\e_0(\pib,K)>0$ such that
if $\e<\e_0$ and $\cU \subseteq (0,K)$ is any fixed open subset containing
$2\eta_*^{-1}$, then
\begin{equation}\label{eq:analysisf1}
\liminf_{t \rightarrow \infty} \sup_{\Pfrak \in \Cst:\,u \in \cU} \sup_{M>0}
\Psi_{1,t}^M(\Pfrak) \geq \Psi_{\mathrm{RS}}, \qquad
\limsup_{t \rightarrow \infty} \sup_{\Pfrak \in \Cst:\,u \in \cU} 
\Psi_{1,t}(\Pfrak) \leq \Psi_{\mathrm{RS}}.
\end{equation}
Furthermore, there exists a universal
constant $c_0>0$ such that for any $\varsigma>0$,
\begin{equation}\label{eq:analysisf2}
\limsup_{t \rightarrow \infty}
\mathop{\sup_{\Pfrak \in \Cst:\,u \in \cU}}_{|u-2\eta_*^{-1}|>\varsigma}
\Psi_{1,t}(\Pfrak)<\Psi_{\mathrm{RS}}-c_0\e^{1/2}\varsigma^2.
\end{equation}
\end{Lemma}

\subsection{Conditional second moment}

Again fixing an iteration $t \geq 1$ for VAMP,
define tuples of primal and dual variables
\[\Pfrak=(u,r,v,w,p),
\qquad \Qfrak=(\Tmr,U,R,V,W,P,\Fpet,\Fpnut,\Fpomt)\]
where $u \in \R_+^2$, $r,U,R \in \R^2$, 
$v,w,V,W \in \R^{t \times 2}$, $p,P \in \R$, 
$\Tmr \in \R^{2 \times 2}$ is symmetric, and $\Fpet,\Fpnut,\Fpomt \in \R^2$.
We write $v=(v_1,v_2)$ where $v_1,v_2 \in \R^t$ are its columns,
and similarly for $w,V,W$. Set the domain for $\Tmr$ as
\begin{equation}\label{eq:defDcp}
    \Dcp=\qty{\Tmr\in \R^{2\times 2}: \Tmr=\Tmr^\top,\;\Tmr\succ -d_-\cdot
I_{2\times 2}},
\end{equation}
and denote the eigen-decomposition of $\Tmr$ by
\begin{equation}\label{eq:Tmreigend}
    \Tmr=\mqty(
y_{1} & y_{2}
)\mqty(
\tmr_{1} & 0 \\
0 & \tmr_{2}
)\mqty(
y_{1}^{\top} \\
y_{2}^{\top}
)
\end{equation} 
where $y_1,y_2 \in \R^2$ are unit-norm eigenvectors of $\Tmr$ and
$\tmr_1,\tmr_2$ are the corresponding eigenvalues. (The expressions below do not
depend on the signs of $y_1,y_2$ or the choice of $y_1,y_2$
when $\tmr_1=\tmr_2$.) Set
\begin{equation}\label{eq:defAft}
    \Aft(\Pfrak)=\begin{pmatrix} u_1 & p \\ p & u_2 \end{pmatrix}
-rr^\top-v^\top v-w^\top w \in \R^{2 \times 2}, \qquad
\Cstt=\qty{\Pfrak:\Aft(\Pfrak)\succ 0}.
\end{equation}
Let $\theta(x),\lambda(x)$ be as in (\ref{eq:defFH}), and define
for $a,b \in \R^2$ and $c \in \R$
\begin{align}
    \cpit\left(a, b, c\right)&=\int \exp \left(a_{1}
x^{2}_1+b_{1} x_1+a_{2} x^{2}_2+b_{2} x_2+c x_1 x_2\right) d\pi(x_1)
d\pi(x_2)\label{eq:defcpit}\\
    \Ht\big(\Tmr,\Aft\big)&=\Tr(\Tmr \Aft)-\E\Big[\log \det (\Tmr+\D^2\cdot I_{2
\times 2})\Big]-\qty(2+\log \det\Aft)\nonumber\\
    \Bc(v,w)&=(v-\alpha_*^A w)^\top (v-\alpha_*^A w)\nonumber
\end{align}
Here $c_\pi(p,a,b),\Hc(\Tmr,\Ac) \in \R$ and $\Bc(v,w) \in \R^{2 \times 2}$.
Finally, define from the above
{\footnotesize
\begin{equation}\label{eq:defPhi2t}
\begin{aligned}
&\Phi_{2,t}(\Pfrak,\Qfrak)\\
&=-1+\E \log \cpit\bigg(U,\;-2U\,\Xstar+\frac{R\,\mathsf{E}}{\sqrt{b_{*}}}+V^\top
\Delta_t^{-1/2}\left(\Xs_{1}, \ldots, \Xs_{t}\right)+\frac{W^\top \Delta_t^{-1/2}\left(\Ys_{1}, \ldots, \Ys_{t}\right)}{\sqrt{\kappa_{*}}}
-P\,\Xstar 1_{2 \times 1},\;P\bigg)\\
&\quad -(u-\rho_*1_{2 \times 1})^\top U-r^\top R-(v_1+\pi_*\Delta_t^{-1/2}1_{t
\times 1})^\top V_1-(v_2+\pi_*\Delta_t^{-1/2}1_{t \times 1})^\top V_2
-w_1^\top W_1-w_2^\top W_2-(p-\rho_*)P\\
&\quad -\frac{1}{2} \qty(\frac{e_*\|r\|^2}{b_{*}}-\frac{2r^\top 1_{2 \times
1}}{\sqrt{b_{*}}}+\Tr \begin{pmatrix} d_* & a_* \\ a_* & c_* \end{pmatrix}
\left[
\left(v_1,\frac{w_1}{\sqrt{\kappa_*}}\right)^\top
\left(v_1,\frac{w_1}{\sqrt{\kappa_*}}\right)
+\left(v_2,\frac{w_2}{\sqrt{\kappa_*}}\right)^\top
\left(v_2,\frac{w_2}{\sqrt{\kappa_*}}\right)\right])\\
&\quad+\frac{1}{2}\Ht\big(\Tmr, \Aft(\Pfrak)\big)
+\frac{1}{2} \E \qty [\qty(\left[\frac{\gamma_*}{\eta_*}1_{2 \times
1}-\frac{r}{\sqrt{b_*}}\right]\D^2-\Fpet)^\top
\Big(\mathsf{E}_b^2(\Tmr+\D^2\cdot I_{2 \times
2})^{-1}\Big)\qty(\left[\frac{\gamma_*}{\eta_*}1_{2 \times
1}-\frac{r}{\sqrt{b_*}}\right]\D^2-\Fpet)]\\
&\quad+\frac{1}{2}
\mathbb{E}\left[\frac{\left(\theta(\D^2)-\lambda(\D^2)
\Fpnuta -\Fpomta\right)^{2}}{\tmr_1+\D^2}\right]y_1^\top \Bc(v,w)y_1
+\frac{1}{2}
\mathbb{E}\left[\frac{\left(\theta(\D^2)-\lambda(\D^2)
\Fpnutb -\Fpomtb\right)^{2}}{\tmr_2+\D^2}\right]y_2^\top \Bc(v,w)y_2
\end{aligned}
\end{equation}}
and denote
\begin{equation}\label{eq:fpsff}
    \Psi_{2,t}(\Pfrak)=\inf_{\Qfrak:\Tmr \in \Dcp} \Phi_{2,t}(\Pfrak,\Qfrak).
\end{equation}

The following results are analogous to the upper bounds of Lemmas
\ref{lemma:firstmt} and \ref{lemma:analysisf}. (Lower bounds may
also be shown, but we omit these statements as we require only the
upper bounds in the subsequent proofs.)

\begin{Lemma}\label{lemma:secondmt}
Fix any $K>0$. Under Assumptions \ref{AssumpD}--\ref{AssumpHighTemp},
there exists a constant $\e_0=\e_0(\pib,K)>0$ such that if $\e<\e_0$, then
for any fixed $t\ge 1$ and non-empty open subset $\cU \subseteq (0,K)$,
almost surely
\begin{equation}
\lim_{n,m \rightarrow \infty} \frac{1}{n} \log \E\left[\Z(\overline{\cU})^{2} \mid
\cG_{t}\right] \leq \sup_{\Pfrak \in \Cstt:\,u_1,u_2 \in \cU} \Psi_{2, t}(\Pfrak)
\end{equation}
where $\overline{\cU}$ is the closure of $\cU$.
\end{Lemma}

\begin{Lemma}\label{lemma:analysisff}
Fix any $K>0$. Under Assumptions \ref{AssumpD}--\ref{AssumpHighTemp},
there exists a constant $\e_0=\e_0(\pib,K)>0$ such that if $\e<\e_0$
and $\cU \subseteq (0,K)$ is any fixed open set containing $2\eta_*^{-1}$, then
$$
\limsup_{t \rightarrow \infty} \sup_{\Pfrak \in \Cstt:\,u_1,u_2 \in \cU} \Psi_{2,
t}(\Pfrak) \leq 2\Psi_{\mathrm{RS}}.
$$
\end{Lemma}

\subsection{Limiting free energy}

Combining the preceding results with the following exponential concentration
inequality for $n^{-1} \log \Z(\cU)$, we deduce as a corollary an
unconditional first-order limit for the restricted free energy.

\begin{Lemma}\label{lemma:concentration}
Fix any $K,L>0$ and subset $\cU \subseteq [0,K]$.
Let $\cE$ denote the event where
\[\int\mathbb{I}\left(\frac{1}{n}\|\sigma-\st\|^2 \in \cU\right)
\prod_{i=1}^n d\pi(\sigma_i)>0 \qquad \text{and}
\qquad \|D^\top \xi\|^2 \leq Ln.\]
(Note that $\cE$ depends on the random quantities $(O,\st,\epsilon)$ only via $(\st,D^\top \xi)$ and is hence $\cG_t$-measurable for any $t \geq 1$.)
Under Assumptions \ref{AssumpD} and \ref{AssumpA}, there exists a constant
$C(K,L,d_+)>0$ such that for any $t \geq 1$, any $\delta>0$, and all
sufficiently large $n$,
\begin{equation}\label{eq:conditionalconcentration}
\mathbb{P}\left(\left|\frac{1}{n}\log
\Z(\cU)-\mathbb{E}\left[\frac{1}{n}\log \Z(\cU) \;\bigg|\;
\mathcal{G}_{t}\right]\right| \geq \delta\;\bigg|\; \cG_t \right)
\mathbb{I}\{\cE\} \leq 2\exp\left(\frac{-\delta^2n}{C(K,L,d_+)}\right).
\end{equation}
\end{Lemma}

\begin{Corollary}\label{cor:PsiRS}
Fix any $K>0$. Under Assumptions \ref{AssumpD}--\ref{AssumpHighTemp},
there exists a constant $\e_0=\e_0(\pib,K)>0$ such that if
$\e<\e_0$ and $2\eta_*^{-1} \in (0,K)$, then almost surely
\[\lim_{n,m \to \infty} \frac{1}{n}\log \Z([0,K])
=\Psi_{\mathrm{RS}}.\]
\end{Corollary}

\section{Proofs of the main results}\label{section:mainpf}

We use the preceding lemmas to prove
Theorems \ref{thm:maintheorem}, \ref{thm:MMSE}, and \ref{thm:TAP}.
For expositional clarity, we consider in this section the simpler
setting where $\pi$ has compact support contained in
$[-\sqrt{\pib},\sqrt{\pib}]$. We extend these proofs to the more general
condition of Assumption \ref{AssumpPrior} in Appendix
\ref{appendix:logconcavemain}.

\subsection{Mutual information}

\begin{proof}[Proof of \Cref{thm:maintheorem}, bounded support]
Letting $p(y \mid \st,A)$ and $p(y \mid A)$ be the conditional density
functions of $y \in \R^m$, direct calculation yields 
\begin{align*}
\mathbb{E}\left[\log \left((2\pi)^{n/2}p\left(y \mid \st, A\right)\right)
\;\Big|\; A\right]&=\E\left[-\frac{\|y-A\st\|^2}{2} \;\bigg|\;
A\right]=-\frac{n}{2}, \quad
\log \left((2\pi)^{n/2}p\left(y \mid A\right)\right)=\log \Z
\end{align*}
where $\Z=\Z([0,\infty))$ is the unrestricted partition function defined by
(\ref{eq:defpartitionfunc}). So the normalized mutual information in the model
(\ref{eq:linearmodel}) is
\[i_{n}=\frac{1}{n} I\left(\st ; y \mid A\right)=\frac{1}{n}
\mathbb{E}\left(\log \frac{p\left(y \mid \st, A\right)}{p(y \mid A)
} \;\bigg|\; A\right)=-\frac{1}{n}\E[\log \Z \mid A]-\frac{1}{2}.\]
Similarly, the mutual information $i(\gamma_*)=I(\Xstar;\Ys)$ in the scalar
channel \eqref{eq:scalarBayes} is
		\begin{align}
			i(\gamma_*)&=\mathbb{E}\left(\log \frac{p\left(\Ys
\mid \Xstar \right)}{p(\Ys)}\right) \nonumber\\
			&=-\frac{1}{2}-\mathbb{E} \log \int \exp
\left(-\frac{(\sqrt{\gamma_*}(\Xstar-x)+\Zs)^{2}}{2}\right) d\pi(x)\nonumber\\
			&=-\frac{1}{2}+\mathbb{E}
\frac{(Z+\sqrt{\gamma_{*}} \Xstar)^2}{2}-\mathbb{E} \log \int \exp
\left(-\frac{1}{2} \gamma_{*} x^{2}+\left(\gamma_{*} \Xstar+\sqrt{\gamma_{*}}
Z\right) x\right) d\pi(x)\nonumber\\
			&=\frac{1}{2}\gamma_{*}\rho_{*}-\mathbb{E} \log
c_{\pi}\left(-\frac{1}{2} \gamma_{*}, \gamma_{*} \Xstar+\sqrt{\gamma_{*}}
Z\right)\label{eq:Indef}
		\end{align}
where $c_\pi$ is defined in (\ref{eq:cpi}).

Suppose $\pi$ has bounded support contained in $[-\sqrt{\pib},\sqrt{\pib}]$.
Setting $K=4\pib$ and $\cU=(0,K)$, we have $n^{-1}\|\sigma-\st\|^2 \leq K$
with probability 1. Then $\Z(\overline{\cU})=\Z$, and also
$2\eta_*^{-1} \leq 2\Vpi<K$ where the first inequality is by Proposition
\ref{prop:uniquefix}. Thus Corollary \ref{cor:PsiRS} shows
$n^{-1}\log \Z \to \Psi_{\mathrm{RS}}$ almost surely. By Jensen's inequality,
\begin{equation}\label{eq:jensen}
0 \leq -\frac{1}{n}\log \Z \leq \frac{1}{n} \int
\frac{\|A(\beta^{\star}-\sigma)+\epsilon\|^{2}}{2}\prod_{i=1}^n
d\pi(\sigma_i)
\leq \left\|A^\top A\right\|_{\mathrm{op}}
\cdot \int \frac{\|\st-\sigma\|^2}{n} \prod_{i=1}^n d\pi(\sigma_i)
+\frac{1}{n}\|\epsilon\|^{2}.
\end{equation}
Then by the given assumption $\|A^\top A\|_{\mathrm{op}}=\|D^\top
D\|_{\mathrm{op}} \to d_+$, the bound $n^{-1}\|\st-\sigma\|^2 \leq K$,
uniform integrability of $\{\|\epsilon\|^2/n\}_{n \geq 1}$,
and the dominated convergence theorem, almost surely
$n^{-1} \E[\log \Z \mid A] \to \Psi_{\mathrm{RS}}$. Applying this,
(\ref{eq:Indef}), and the forms of $\Psi_{\mathrm{RS}}$ and $i_{\mathrm{RS}}$
in (\ref{eq:RSdef}) and (\ref{eq:iRS}), we obtain as desired
\[\lim_{n,m \to \infty} i_n=-\Psi_{\mathrm{RS}}-1/2=i_{\mathrm{RS}}(\eta_*^{-1},\gamma_*).\]
\end{proof}

\subsection{Bayes risk}\label{section:bayesriskproof}

\begin{Lemma}\label{lemma:harli}
Denoting by $\langle f(\sigma) \rangle$ the posterior expectation in
(\ref{eq:anglepost}),
    $$\operatorname{mmse}_{n}=\frac{1}{2 n}
\mathbb{E}\left[\left\langle\left\|\s-\st\right\|^{2}\right\rangle \;\Big|\;
A\right].$$
\end{Lemma}

\begin{proof}
Let $\sigma,\tau$ denote two replicas sampled independently from the posterior
distribution defining (\ref{eq:anglepost}). Conditional on $y$ and $A$, since
$\sigma,\tau,\st$ are independent and equal in law, we have
the Nishimori identity for any integrable function $f$
(see also \cite[Appendix A]{barbier2018mutual})
\[\E[\langle f(\sigma,\tau) \rangle \mid A]
=\E[\E[f(\sigma,\tau) \mid y,A] \mid A]
=\E[\E[f(\sigma,\st) \mid y,A] \mid A]
=\E[\langle f(\sigma,\st) \rangle \mid A].\]
Thus
\begin{equation}\label{eq:Nish}
\begin{aligned}
n \cdot \operatorname{mmse}_{n}&=\mathbb{E}[\|\st-\langle\s\rangle\|^{2} \mid A]
=\mathbb{E}[{\st}^{\top} \st+\langle\s^{\top} \ta\rangle-\langle\s^{\top} \st\rangle-\langle\ta^{\top} \st\rangle \mid A] \\
&\stackrel{(a)}{=}\mathbb{E}[\beta^{\star \top} \st-\langle\s^{\top} \st\rangle \mid A]
=\mathbb{E}[\langle(\s-\st)^{\top}(-\st)\rangle \mid A]\stackrel{(b)}{=}\mathbb{E}[\langle(\s-\st)^{\top}\s\rangle \mid A]
\end{aligned}
\end{equation}
where $(a)$ applies the Nishimori identity, and $(b)$ follows from the
exchangeability of the replicas and the Nishimori identity as follows: 
\[\mathbb{E}[\langle(\s-\st)^{\top}(-\st)\rangle \mid A]=\mathbb{E}[\langle(\s-\ta)^{\top}(-\ta)\rangle \mid A]
   =\mathbb{E}[\langle(\ta-\s)^{\top}(-\s)\rangle \mid A]
    =\mathbb{E}[\langle(\st-\s)^{\top}(-\s)\rangle \mid A].\]
Summing the last two expressions in \eqref{eq:Nish} gives
$2n \cdot \operatorname{mmse}_{n}=\mathbb{E}[\langle\|\s-\st\|^{2}\rangle
\mid A]$ as required.
\end{proof}

\begin{proof}[Proof of \Cref{thm:MMSE}, bounded support]
Suppose again that $\pi$ has support contained in $[-\sqrt{\pib},\sqrt{\pib}]$,
and set $K=4\pib$. Then again
$n^{-1}\|\sigma-\st\|^2 \leq K$ with probability 1, the
unrestricted partition function is $\Z=\Z([0,K])$, and $2\eta_*^{-1} \leq 2\rho_*<K$.
Fix any small constant $\varsigma>0$
and set $\cU=(0,K) \setminus (2\eta_*^{-1}-\varsigma,2\eta_*^{-1}+\varsigma)$.
Then
\[\expval{\mathbb{I}\qty(\qty|\frac{1}{n}\norm{\s-\st}^2-2\eta_*^{-1}|>\varsigma)}=\frac{\Z(\overline{\cU})}{\Z}.\]

Applying Lemmas \ref{lemma:firstmt}, \ref{lemma:analysisf}, and Jensen's
inequality, for a sufficiently large iteration $t \geq 1$,
almost surely for all large $n$,
$$
\frac{1}{n} \mathbb{E}\left[\log \Z(\overline{\cU}) \mid
\cG_{t}\right] \leq \frac{1}{n} \log \mathbb{E}\left[\Z(\overline{\cU}) \mid \cG_{t}\right]<\Psi_{\mathrm{RS}}-c_{0} \e^{1 / 2} \varsigma^{2}.
$$
Taking the expectation of (\ref{eq:conditionalconcentration}) from Lemma
\ref{lemma:concentration} yields the unconditional tail bound
\begin{equation}\label{eq:unconditionalconcentration}
\mathbb{P}\left(\cE \text{ holds and } \left|\frac{1}{n}\log
\Z(\overline{\cU})-\mathbb{E}\left[\frac{1}{n}\log \Z(\overline{\cU}) \;\bigg|\;
\mathcal{G}_{t}\right]\right| \geq \delta\right)
\leq 2\exp\left(\frac{-\delta^2n}{C(K,L,d_+)}\right).
\end{equation}
Applying $\|D^\top D\|_{\mathrm{op}} \to d_+$ as $n \to \infty$ and a standard
chi-squared tail bound, for a sufficiently large constant $L>0$,
the second condition $\|D^\top \xi\|^2 \leq Ln$ defining $\cE$ holds almost
surely for all large $n$. The first condition defining $\cE$ is equivalent to
$\Z(\overline{\cU}) \neq 0$. Hence by (\ref{eq:unconditionalconcentration})
applied with $\delta=(c_0/2)\e^{1/2}\varsigma^2$ and by
the Borel-Cantelli lemma, almost surely for all large $n$,
either $\Z(\overline{\cU})=0$ or
$$
\frac{1}{n} \log \mathcal{Z}(\overline{\cU})
<\frac{1}{n} \mathbb{E}\left[\log \Z(\overline{\cU}) \mid
\cG_{t}\right]+\frac{c_{0}}{2} \e^{\frac{1}{2}} \varsigma^{2}
<\Psi_{\mathrm{RS}}-\frac{c_{0}}{2} \e^{1 / 2} \varsigma^{2}.
$$
Then combining with $n^{-1}\log \Z \to \Psi_{\mathrm{RS}}$ by
Corollary \ref{cor:PsiRS}, almost surely for all large $n$,
\[\expval{\mathbb{I}\qty(\qty|\frac{1}{n}\norm{\s-\st}^2-2\eta_*^{-1}|>\varsigma
)}=\frac{\Z(\overline{\cU})}{\Z}
<\exp(-\frac{c_0}{2} \e^{1 / 2} \varsigma^{2} \cdot n).\]
Since $|n^{-1}\left\|\s-\st\right\|^{2}-2 \eta_{*}^{-1}| \leq 2K$,
this tail bound implies
\[\left|\frac{1}{n}\langle \|\sigma-\st\|^2 \rangle-2\eta_*^{-1} \right|
\leq \left\langle\left|\frac{1}{n}\left\|\s-\st\right\|^{2}-2
\eta_{*}^{-1}\right|\right\rangle
\leq \varsigma+2K \cdot\exp(-\frac{c_0}{2}\e^{1 / 2} \varsigma^{2} \cdot n)
<2\varsigma\]
almost surely for all large $n$. Since $\varsigma>0$ is arbitrary, this shows
$(2n)^{-1}\langle \|\sigma-\st\|^2 \rangle \to \eta_*^{-1}$
almost surely. Then by \Cref{lemma:harli} and dominated convergence theorem,
also $\mmse_n=(2n)^{-1} \E[\langle \|\sigma-\st\|^2 \rangle \mid A] \to \eta_*^{-1}$ almost surely.
\end{proof}

\subsection{The TAP equations}\label{section:TAPproof}

We note that the stationary initialization for VAMP in (\ref{eq:VAMPstatinit})
requires knowledge of $\st$, and hence the resulting iterates do not define
estimators of $\st$ given only $(y,A)$. Here, we consider VAMP from the
non-informative initialization
$r_2^1=0$ and $\gamma_{2,1}=\rho_*^{-1}$. We first use
Theorem \ref{thm:MMSE} already proven to show convergence of the VAMP state
evolution; a different argument for this convergence has also been given
recently in \cite{takeuchi2022convergence}.

\begin{Proposition}\label{prop:compmse}
Consider the VAMP algorithm (\ref{eq:iteupdate}) with initialization $r_2^1=0$
and $\gamma_{2,1}=\rho_*^{-1}$. Under Assumptions
\ref{AssumpD}--\ref{AssumpHighTemp}, there exists a constant
$\e_0=\e_0(\pib)>0$ such that if $\e<\e_0$, then
$(\eta_{1,t})_{t\ge 1}$, $(\gamma_{1,t})_{t\ge 1}$,
$(\eta_{2,t})_{t\ge 1}$, and $(\gamma_{2,t})_{t\ge 1}$ are monotone increasing
and converge to $\eta_*,\gamma_*,\eta_*,\eta_*-\gamma_*$
respectively. Consequently, for $\hat{\beta}_j^t$ as defined in
Theorem \ref{thm:RanganSE},
    \[\lim_{t \to \infty} \lim_{n,m \rightarrow \infty}
\frac{1}{n}\E[\|\hat{\beta}_j^t-\st\|^{2}\mid A]=\eta_*^{-1}\]
where the inner limit exists almost surely.
\end{Proposition}
\begin{proof}[Proof of Proposition \ref{prop:compmse}]
From \Cref{thm:RanganSE}, for both $j=1,2$ and each fixed $t$,
we have almost surely
\[\lim_{n,m \to \infty} \frac{1}{n}\|\hat{\beta}_j^t-\st\|^2=\eta_{j,t}^{-1}.\]
To apply the dominated convergence theorem, note that
under Assumption \ref{AssumpA}, the largest and smallest eigenvalues of
$A^\top A$ converge to $d_+,d_- \geq 0$. Then applying
$\gamma_{2,t}>0$ by \Cref{thm:RanganSE} and that $r \mapsto f(r,\gamma_{1,t})$
is Lipschitz by Proposition \ref{prop:denoiserlipschitz},
from the forms of the iterations (\ref{eq:iteupdate}), there are constants
$C_t>0$ (depending on the state evolution parameters (\ref{eq:SEparams})
and the value of $f(0,\gamma_{1,t})$) for which, for all large $n$,
\[\|r_1^t\| \leq C_t(\|y\|+\|r_2^t\|), \qquad \|r_2^{t+1}\|
\leq C_t(\sqrt{n}+\|r_1^t\|).\]
Iterating these bounds and applying the definitions of $\hat{\beta}_j^t$,
there are constants $C_t'>0$ for which, for both $j=1,2$ and all large $n$,
\[\|\hat{\beta}_j^t\|/\sqrt{n}
\leq C_t'\left(1+\left(\|y\|/\sqrt{n}\right)^t\right).\]
Then, applying $\|y\| \leq \|A\|\,\|\st\|+\|\epsilon\|$,
for some constants $C_t''>0$,
\begin{equation}\label{eq:AMPbound}
\frac{\|\hat{\beta}_j^t\|^4}{n^2}
\leq C_t''\left(1+\left(\frac{\|\st\|^2}{n}\right)^{2t}
+\left(\frac{\|\epsilon\|^2}{n}\right)^{2t}\right)
\leq C_t''\left(1+\frac{1}{n}\sum_{i=1}^n \Big({\st_i}^{4t}
+\epsilon_i^{4t}\Big)\right).
\end{equation}
For each fixed $t$,
this upper bound has finite expectation independent of $n$,
so $\{\|\hat{\beta}_j^t\|^2/n\}_{n \geq 1}$ is
bounded in $L^2$ and hence uniformly integrable. Then the dominated convergence
theorem implies
\[\lim_{n,m \to \infty} \frac{1}{n}\E[\|\hat{\beta}_j^t-\st\|^2 \mid A]
=\eta_{j,t}^{-1}.\]
Combining this with \Cref{thm:MMSE}, we must have $\eta_{j,t}^{-1} \geq
\eta_*^{-1}$ for every $t$, because each
$\hat{\beta}_j^t$ is a $(y,A)$-measurable estimator of $\st$ and hence
$n^{-1}\E[\|\hat{\beta}_j^t-\st\|^2 \mid A] \geq \mmse_n$.

It remains to show the monotonicity and convergence of $(\eta_{j,t})_{t \geq
1}$ and $(\gamma_{j,t})_{t \geq 1}$. Applying
the definition of the R-transform to write $\gamma_{1,t}=\eta_{2,t}-\gamma_{2,t}
=-R(\eta_{2,t})$, the iterations (\ref{eq:SEparams}) yield
\begin{equation}
\eta_{2,t+1}^{-1}=G\left(\mmse\left(-R\left(\eta_{2,t}^{-1}\right)\right)^{-1}+R\left(\eta_{2,t}^{-1}\right)\right).
\end{equation}
We claim that for any $x_1 \in [\eta_*^{-1},G(-d_-))$,
\begin{equation}\label{eq:x2def}
x_{2}:=G\left(\mmse\left(-R\left(x_{1}\right)\right)^{-1}+R\left(x_{1}\right)\right)
\leq x_{1}
\end{equation}
To see this, note that $R(x)$ is negative and increasing by Lemma \ref{lem:cauchy},
and $\mmse(\gamma)$ is decreasing over $\gamma>0$ (by the law of total
variance). Then $x \mapsto \mmse(-R(x))$ is increasing over $x \in (0,
G(-d_{-}))$. Proposition \ref{prop:uniquefix} implies that
this function has a unique fixed point $x=\eta_*^{-1} \in (0,G(-d_-))$.
Furthermore $\lim_{x \to 0} \mmse(-R(x))=\mmse(1)>0$ strictly, because the
prior distribution $\pi$ has strictly positive variance. Then $\mmse(-R(x))>x$
for $x<\eta_*^{-1}$ and $\mmse(-R(x))<x$ for $x>\eta_*^{-1}$, so in particular
\[\mmse\left(-R\left(x_{1}\right)\right) \leq x_{1}\]
when $x_1 \geq \eta_*^{-1}$.
Then $\mmse(-R(x_1))^{-1}+R(x_1) \geq x_1^{-1}+R(x_1)=G^{-1}(x_1)>-d_-$.
So $x_2$ is well-defined
in (\ref{eq:x2def}), and also $x_2 \leq G(G^{-1}(x_1))=x_1$
as desired because $G$ is decreasing.

Finally, since $\eta_{2,t}^{-1} \geq \eta_*^{-1}$ for all $t \geq 1$,
(\ref{eq:x2def}) implies that $(\eta_{2,t}^{-1})_{t \geq 1}$ is a
monotonically decreasing sequence, which must then converge to a fixed point of
$x \mapsto G(\mmse(-R(x))^{-1}+R(x))$. Such a fixed point satisfies
$G^{-1}(x)=\mmse(-R(x))^{-1}+R(x)$, i.e.\ $x=\mmse(-R(x))$, so it must be the
unique fixed point $\eta_*^{-1}$. Thus $(\eta_{2,t})_{t \geq 1}$ monotonically
increases to $\eta_*$. It is then straightforward to verify from their
definitions in (\ref{eq:SEparams}) that
$(\gamma_{2,t})_{t \geq 1}$, $(\gamma_{1,t+1})_{t \geq 1}$,
and $(\eta_{1,t+1})_{t \geq 1}$ also monotonically increase to
$\eta_*-\gamma_*$, $\gamma_*$, and $\eta_*$.
\end{proof}

\begin{proof}[Proof of \Cref{thm:TAP}, bounded support]
First note the following ``Pythagorean relation'': For $j=1,2$ and any $t \geq
1$,
\[\mathbb{E}[\|\st-\hat{\beta}_j^t\|^{2} \mid
A]=\mathbb{E}[\|\st-\langle\sigma\rangle\|^{2}
\mid A]+\mathbb{E}[\|\hat{\beta}_j^t-\langle\sigma\rangle\|^{2} \mid A]\]
because $\hat{\beta}_j^t$ is a $(y,A)$-measurable estimator of $\st$ and
$\langle \sigma \rangle=\E[\st \mid y,A]$.
Then by \Cref{thm:MMSE} and \Cref{prop:compmse}, for both $j=1,2$,
\begin{equation}\label{eq:conctapfkr}
\lim _{t \rightarrow \infty} \lim_{n,m \rightarrow \infty} \frac{1}{n}
\mathbb{E}[\|\hat{\beta}_j^t-\langle\sigma\rangle\|^{2} \mid A]=0
\end{equation}
where the inner limit exists almost surely.
It follows from this, uniform boundedness of $(\gamma_{2,t})_{t \geq 1}$
and $(\eta_{1,t})_{t \geq 1}$,
the convergence $\eta_{1,t}-\gamma_{2,t}=\gamma_{1,t-1}
\to \gamma_*$ from \Cref{prop:compmse}, and the triangle inequality that
\begin{equation}\label{fadsff1}
\lim _{t \rightarrow \infty} \lim_{n,m \rightarrow \infty} \frac{1}{n}
\mathbb{E}\left[\left\|\left(A^{\top}
A\langle\sigma\rangle-\gamma_{*}\langle\sigma\rangle-A^{\top}
y\right)-\left((A^{\top} A+\gamma_{2,t}I) \hat{\beta}_2^t
-\eta_{1,t} \hat{\beta}_1^t-A^{\top} y\right)\right\|^{2} \;\bigg|\; A\right]=0.
\end{equation}
From the definitions of $\hat{\beta}_j^t$ and (\ref{eq:itupdate2}),
$\gamma_{2,t}r_2^t=\left(A^{\top} A+\gamma_{2,t} I\right)
\hat{\beta}_2^t-A^{\top} y=\eta_{1,t}\hat{\beta}_1^t-\gamma_{1,t-1}
r_1^{t-1}$. Substituting this identity into (\ref{fadsff1}),
\[\lim_{t \rightarrow \infty} \lim_{n,m \rightarrow \infty}
\frac{1}{n}\mathbb{E}\left[\left\|\left(A^{\top}
A\langle\sigma\rangle-\gamma_*\langle\sigma\rangle-A^{\top} y\right)
+\gamma_{1,t-1}r_1^{t-1} \right\|^{2} \;\Big|\; A \right]=0.\]
Then dividing by $\gamma_{1,t-1}^2$, applying $\gamma_{1,t-1} \to \gamma_*$,
and applying that $r \mapsto f(r,\gamma_*)$ is Lipschitz by
Proposition~\ref{prop:denoiserlipschitz}
\begin{equation}\label{eq:almostTAP}
\lim_{t \rightarrow \infty} \lim_{n,m \rightarrow \infty}
\frac{1}{n}\mathbb{E}\left[\left\|f\left(-\gamma_*^{-1}\left(A^{\top}
A\langle\sigma\rangle-\gamma_*\langle\sigma\rangle-A^{\top}
y\right),\,\gamma_*\right)-f(r_1^{t-1},\gamma_*) \right\|^{2}
\;\Big|\; A \right]=0.
\end{equation}

Finally, applying again that $r \mapsto f(r,\gamma)$ is Lipschitz,
(\ref{eq:RanganSE}) implies that for each fixed $t$ we have almost surely
\[\lim_{n,m \to \infty} \frac{1}{n}
\left\|f(r_1^{t-1},\gamma_*)-f(r_1^{t-1},\gamma_{1,t-1})\right\|^2\\
=\E\Big[\big(f(\mathsf{R}_1^{t-1},\gamma_*)-f(\mathsf{R}_1^{t-1},\gamma_{1,t-1})\big)^2\Big]\]
Here, assuming that $\pi$ has bounded support, $f(r,\gamma)$ is bounded, so
the dominated convergence theorem yields
\begin{equation}\label{eq:TAPDCT}
\lim_{n,m \to \infty} \frac{1}{n}\E\left[\left\|f(r_1^{t-1},\gamma_*)
-f(r_1^{t-1},\gamma_{1,t-1}) \right\|^2
\;\bigg|\;A\right]=\E\Big[\big(f(\mathsf{R}_1^{t-1},\gamma_*)-f(\mathsf{R}_1^{t-1},\gamma_{1,t-1})\big)^2\Big].
\end{equation}
Representing $\mathsf{R}_1^{t-1}=\Xstar+\Zs/\sqrt{\gamma_{1,t-1}}$
and applying $\gamma_{1,t-1} \to \gamma_*$, an application of the dominated
convergence theorem shows that the right side converges to 0 as $t \to \infty$.
Then, recalling $f(r_1^{t-1},\gamma_{1,t-1})=\hat{\beta}_j^1$
and applying also the statement (\ref{eq:conctapfkr}) for $\hat{\beta}_j^1$
and the triangle inequality, this shows
\[\lim_{t \to \infty} \lim_{n,m \to \infty}
\frac{1}{n}\E\left[\left\|f(r_1^{t-1},\gamma_*)-\langle \sigma \rangle
\right\|^2\;\Big|\;A\right]=0.\]
Combining this with (\ref{eq:almostTAP}) concludes the proof.
\end{proof}

\appendix
\section{State evolution of VAMP}\label{appendix:VAMP}
We prove Theorems \ref{thm:RanganSE}, \ref{thm:ampSE} and
Propositions \ref{prop:VAMPequiv}, \ref{prop:AMPparamconverge} on Vector AMP.
These results do not require the high temperature condition of Assumption
\ref{AssumpHighTemp}, and the stationary initialization (\ref{eq:VAMPstatinit})
of VAMP and associated
scalar parameters may be defined with respect to any fixed point
$(\eta_*^{-1},\gamma_*) \in (0,G(-d_-)) \times \R_+$ of (\ref{eq:fix}).

\begin{proof}[Proof of Theorem \ref{thm:RanganSE}]
This follows from \cite[Theorems 1 and 2]{rangan2019vector}. The state
evolution (\ref{eq:SEparams}) corresponds
to the setting of matched MMSE denoising described in
\cite[Theorem 2]{rangan2019vector}, where
$\mathcal{E}_1(\gamma_{1,t})=\mmse(\gamma_{1,t})$
and $\mathcal{E}_2(\gamma_{2,t})=G(\gamma_{2,t})$ as shown in
\cite[Eq.\ (41)]{rangan2019vector}. 
(Our quantities $\eta_{j,t},\gamma_{j,t}$ defined by (\ref{eq:SEparams}) are
the asymptotic quantities $\overline{\eta}_{jt},\overline{\gamma}_{jt}$
in \cite{rangan2019vector}.) We note that under (\ref{eq:SEparams}),
$\gamma_{1,t}=\eta_{2,t}-G^{-1}(\eta_{2,t}^{-1})=-R(\eta_{2,t}^{-1})$, which is
positive for $\eta_{2,t}>0$ by Lemma \ref{lem:cauchy}(b). Furthermore,
setting $A(\gamma)=\gamma \cdot \mmse(\gamma)$, we have
$\gamma_{2,t+1}=(1-A(\gamma_{1,t}))/\mmse(\gamma_{1,t})$. The argument of
\cite[Section IV.F]{rangan2019vector} shows $A(\gamma) \in (0,1)$ for any
$\gamma>0$, hence $\gamma_{2,t+1}>0$ when $\gamma_{1,t}>0$. Thus the parameters
of (\ref{eq:SEparams}) are all positive and well-defined. 

The proof of \cite[Theorem 1]{rangan2019vector} is easily adapted to start from
an initialization $(r_2^1,\gamma_{2,1})$ instead of $(r_1^0,\gamma_{1,0})$, and
to use the deterministic state evolution parameters (\ref{eq:SEparams})
and the associated quantities $\alpha_{j,t}=\gamma_{j,t}/\eta_{j,t}$ instead of 
their empirical estimates. (We initialize with $(r_2^1,\gamma_{2,1})$
so that $\gamma_{2,1}^{-1}$ is correctly matched with the variance of
$r_2^1-\st$, without
requiring $r_1^0$ to have a limit $\mathsf{R}_1^0$ that is a Gaussian
perturbation of $\Xstar$.) The  
empirical convergence of  $D^\top 1_{m \times 1}$ and of $(r_2^1,\st)$ on Pseudo-Lipschitz test functions of order 2 (as required in
\cite{rangan2019vector}) are implied by our assumptions of empirical
Wasserstein convergence at all orders. The first
condition of \cite[Theorem 1]{rangan2019vector} follows from $A(\gamma)
\in (0,1)$ discussed above, the second condition from
the continuity of $\mmse(\cdot)$ and $G(\cdot)$, and the final uniform Lipschitz
condition for $f(\cdot,\gamma)$ from Assumption \ref{AssumpPrior}.
The convergence (\ref{eq:RanganSE}) is then shown in \cite[Theorem 1, Eq.\
(45)]{rangan2019vector}, and (\ref{eq:RanganMSE}) is shown in
\cite[Theorem 2, Eq.\ (56c)]{rangan2019vector}.
\end{proof}

\subsection{Identities for stationary VAMP}\label{append:AMP}

\begin{Proposition}\label{prop:deltabound}
Suppose Assumptions \ref{AssumpD} and \ref{AssumpPrior} hold, and let
$(\eta_*^{-1},\gamma_*) \in (0,G(-d_-)) \times \R_+$ be any fixed point of
(\ref{eq:fix}). Then $\eta_*^{-1}\leq \Vpi$ and
$\eta_*-\gamma_*\ge \Vpi^{-1}>0$.
\end{Proposition}
\begin{proof}
Let $(\Xstar,\Ys)$ be as defined in the scalar channel (\ref{eq:scalarBayes}).
The law of total variance implies $\eta_*^{-1}=\mmse(\gamma_*)
\leq \V(\Xstar)=\Vpi$. The second claim
$\eta_*-\gamma_*\ge \rho_*^{-1}$ follows from comparing the
MMSE with the error of the linear estimator $a\Ys$ with
$a=\frac{\Vpi}{\Vpi+\gamma_{\star}^{-1}}$:
\begin{equation}\label{eq:linearcompare}
\eta_{*}^{-1}=\mathbb{E}\left(\Xstar-\mathbb{E}\left[\Xstar \mid \Ys\right]\right)^{2}
		\leq
\mathbb{E}\left(\Xstar-a\left(\Xstar+\frac{1}{\sqrt{\gamma_{*}}} \Zs\right)\right)^{2}
		=(1-a)^{2}\Vpi+\frac{a^{2}}{\gamma_{*}}=\frac{1}{\Vpi^{-1}+\gamma_{*}}.
\end{equation}
Rearranging yields $\eta_*-\gamma_*\ge \rho_*^{-1}$.
\end{proof}

\begin{proof}[Proof of Proposition \ref{prop:VAMPequiv}]
Writing (\ref{eq:fix}) as $\eta_*^{-1}=\mmse(\gamma_*)$ and
$\eta_*^{-1}=G(\eta_*-\gamma_*)$, it is clear that the initialization
$\gamma_{1,0}=\gamma_*$ yields $\eta_{1,t}=\eta_{2,t}=\eta_*$,
$\gamma_{1,t}=\gamma_*$, and $\gamma_{2,t}=\eta_*-\gamma_*$ for all $t$.

Then substituting $y=A\st+\epsilon$ and $A=Q^\top DO$,
the update rule \eqref{eq:itupdate1} for $r_1^t$ can be rearranged as
\[r_1^t=O^{\top} \Lambda O r_2^t
+O^\top \bigg[\frac{\eta_*}{\gamma_*} (D^\top D+(\eta_*-\gamma_*)I)^{-1}
D^\top D\bigg]O\st+e.\]
Applying the identity $(D^{\top} D+(\eta_*-\gamma_*)I)^{-1} D^{\top} D
=I-(\eta_*-\gamma_*)(D^\top D+(\eta_*-\gamma_*)I)^{-1}
=(\gamma_*/\eta_*)(I-\Lambda)$ we obtain
\[y^t:=r_1^t-e-\st=O^\top \Lambda O(x^t+\st)
+O^\top(I-\Lambda)O\st-\st=O^\top \Lambda Ox^t\]
which may be written as $s^t=Ox^t$ and $y^t=O^\top \Lambda s^t$.
Setting $r_1^t=p^t+\st$, we have from \eqref{eq:itupdate2}
\[x^{t+1}:=r_2^{t+1}-\st=\frac{\eta_*}{\eta_*-\gamma_*}
f(p^t+\st,\gamma_{1,t})-\frac{\gamma_*}{\eta_*-\gamma_*}(p^t+\st)-\st
=F(p^t,\st).\]
For $t=0$, this gives the initialization $x^1=F(p^0,\st)$, and for
$t \geq 1$, we have $p^t=y^t+e$ so this gives the update for $x^{t+1}$.
\end{proof}

\begin{Proposition}\label{prop:denoiserprop}
Suppose Assumptions \ref{AssumpD} and \ref{AssumpPrior} hold, and let
$(\eta_*^{-1},\gamma_*)$ be any fixed point of (\ref{eq:fix}).
Let $\Ys$ denote an observation from the scalar channel (\ref{eq:scalarBayes})
with variance $\gamma_*^{-1}$. Then the functions
$y \mapsto f(y,\gamma_*)$ and $(p,\beta) \mapsto F(p,\beta)$ are continuously
differentiable and Lipschitz. We have
	\begin{equation}\label{eq:denoiserpp1}
f'(y,\gamma_*):=\frac{\partial}{\partial y} f(y,\gamma_*)=\gamma_*\,
\V(\Xstar \mid \Ys=y), \quad
		F^{\prime}(p, \beta):=\frac{\partial}{\partial p} F(p,
\beta)=\frac{\eta_{*}}{\eta_{*}-\gamma_{*}}
f^{\prime}(p+\beta,\gamma_*)-\frac{\gamma_{*}}{\eta_{*}-\gamma_{*}},
	\end{equation}
and these are non-constant in $y$ and $p$.
For $\mathsf{P} \sim N(0, \gamma_*^{-1})$ and $\Xstar \sim
\pi$ independent, we also have
	\begin{equation}\label{eq:denoiserpp2}
		\mathbb{E} F\left(\mathsf{P}, \Xstar\right)=0, \quad
		\mathbb{E}F^{\prime}\left(\mathsf{P}, \Xstar\right)=0,  \quad
		\mathbb{E}\left(F\left(\mathsf{P},
\Xstar\right)^{2}\right)=\delta_*
	\end{equation}
\end{Proposition}
\begin{proof}
Proposition~\ref{prop:denoiserlipschitz} shows that $y \mapsto f(y,\gamma_*)$ is
continuously-differentiable and Lipschitz, with derivative given by
(\ref{eq:denoiserpp1}). Since $\pi$ is also non-Gaussian by
Assumption~\ref{AssumpPrior}, $y \mapsto f(y,\gamma_*)$ is non-linear, and hence
$y \mapsto f'(y,\gamma_*)$ is non-constant. Then the same properties hold for
$F$ and $F'$, and the form (\ref{eq:denoiserpp1}) for $F'$
follows from definition of $F$.

The first two identities in \eqref{eq:denoiserpp2} follow from $\E \Ps=\E
\Xstar=\E f(\Ps+\Xstar,\gamma_*)=0$ and $\E f'(\Ps+\Xstar,\gamma_*)=\gamma_*
\cdot \mmse(\gamma_*)=\gamma_*/\eta_*$.
For the last identity in \eqref{eq:denoiserpp2},
denote for simplicity $f(y)=f(y,\gamma_*)$. Note that
	\begin{equation*}
		\mathbb{E}\left(f\left(\mathsf{P}+\Xstar\right)-\Xstar\right)^{2}=\mmse(\gamma_*)=\eta_*^{-1}
	\end{equation*} and thus
	\begin{equation}\label{eq:step1denoise}
		\left(\frac{\eta_{*}}{\eta_{*}-\gamma_{*}}\right)^{2}\left(\mathbb{E}\left(f\left(\mathsf{P}+\Xstar\right)-\Xstar\right)^{2}-\frac{\gamma_{*}}{\eta_{*}^{2}}\right)=\left(\frac{\eta_{*}}{\eta_{*}-\gamma_{*}}\right)^{2}\left(\frac{1}{\eta_{*}}-\frac{\gamma_{*}}{\eta_{*}^{2}}\right)=\frac{1}{\eta_{*}-\gamma_{*}}=\delta_*.
	\end{equation}
	It follows that
	\begin{equation*}
		\begin{aligned}
			\mathbb{E}\left(F\left(\mathsf{P}, \Xstar\right)^{2}\right) &=\mathbb{E}\left(\frac{\eta_{*}}{\eta_{*}-\gamma_{*}} f\left(\mathsf{P}+\Xstar\right)-\frac{\gamma_{*}}{\eta_{*}-\gamma_{*}} \mathsf{P}-\frac{\eta_{*}}{\eta_{*}-\gamma_{*}} \Xstar\right)^{2} \\
			&=\left(\frac{\eta_{*}}{\eta_{*}-\gamma_{*}}\right)^{2}\left(\mathbb{E}\left(f\left(\mathsf{P}+\Xstar\right)-\Xstar\right)^{2}+\left(\frac{\gamma_{*}}{\eta_{*}}\right)^{2}
\mathbb{E} \mathsf{P}^{2}-\frac{2\gamma_{*}}{\eta_{*}} \mathbb{E}\left(\left(f\left(\mathsf{P}+\Xstar\right)-\Xstar\right) \mathsf{P}\right)\right) \\
			&
\stackrel{(a)}{=}\left(\frac{\eta_{*}}{\eta_{*}-\gamma_{*}}\right)^{2}
\left(\mathbb{E}\left(f\left(\mathsf{P}+\Xstar\right)-\Xstar\right)^{2}+\frac{\gamma_{*}}{\eta_{*}^{2}}-\frac{2}{\eta_{*}} \mathbb{E}\left(f^{\prime}\left(\mathsf{P}+\Xstar\right)\right)\right) \\
			&
\stackrel{(b)}{=}\left(\frac{\eta_{*}}{\eta_{*}-\gamma_{*}}\right)^{2}
\mathbb{E}\left(\left(f\left(\mathsf{P}+\Xstar\right)-\Xstar\right)^{2}-\frac{\gamma_{*}}{\eta_{*}^{2}}\right)=\delta_*
		\end{aligned}
	\end{equation*}
	where we used Gaussian integration by parts in $(a)$ and $\E
f'(\mathsf{P}+\Xstar,\gamma_*)=\gamma_*/\eta_*$ and \eqref{eq:step1denoise} in $(b)$. 
\end{proof}

\begin{Lemma}\label{lemma:dABC}
Recall $b_*,\kappa_*$ from (\ref{eq:scalarparams}) and $a_*,c_*,e_*$
from (\ref{eq:auxparams}). We have the following identities
	\begin{subequations}
		\begin{align}
			&d_{*}^{A}:=\mathbb{E} \frac{1}{\Dbar^{2}+\eta_{*}-\gamma_{*}}=\frac{1}{\eta_{*}} \label{eq:dA}\\
			&d_{*}^{B}:=\mathbb{E} \frac{\Dbar^{2}}{\Dbar^{2}+\eta_{*}-\gamma_{*}}=\frac{\gamma_{*}}{\eta_{*}} \label{eq:dB}\\
			&d_{*}^{C}:=\mathbb{E} \frac{1}{\left(\Dbar^{2}+\eta_{*}-\gamma_{*}\right)^{2}}=\frac{1}{\eta_{*}^{2}}\left(\frac{\gamma_{*}}{\eta_{*}-\gamma_{*}}\right)^{2} \kappa_{*}+\frac{1}{\eta_{*}^{2}} \label{eq:dC}\\
			&d_{*}^{D}:=\mathbb{E} \frac{\Dbar^{2}}{\left(\Dbar^{2}+\eta_{*}-\gamma_{*}\right)^{2}}=-\frac{\gamma_{*}^{2}}{\eta_{*}^{2}}\left(\frac{\kappa_{*}}{\eta_{*}-\gamma_{*}}\right)+\frac{\gamma_{*}}{\eta_{*}^{2}} \label{eq:dD}\\
			&d_{*}^{E}:=\mathbb{E} \frac{\Dbar^{4}}{\left(\Dbar^{2}+\eta_{*}-\gamma_{*}\right)^{2}}=\frac{\gamma_{*}^{2}}{\eta_{*}^{2}}\left(1+\kappa_{*}\right) \label{eq:dE}
		\end{align}
	\end{subequations}
Also, for the quantities $\Ls,\Es_b$ from Proposition
\ref{prop:AMPparamconverge}, we have
\begin{equation}\label{eq:abcde}
\E \Ls=0, \quad \E \Ls^2=\kappa_*, \quad \E \Es_b^2=b_*, \quad
\E \D^2\Ls=a_*, \quad \E \D^2\Ls^2=c_*, \quad \E \D^2\Es_b^2=e_*.
\end{equation}
\end{Lemma} 
\begin{proof}
\eqref{eq:dA} is the identity $\eta_*^{-1}=G(\eta_*-\gamma_*)$,
which is a rewriting of the second equation of (\ref{eq:fix}). This then
implies $\E\Ls=0$, as well as $\E\Ls^2=\kappa_*$.
	\eqref{eq:dB} follows from the identity
	$d_{*}^{B}+\left(\eta_{*}-\gamma_{*}\right) d_{*}^{A}=1$, and this then
implies $\E\D^2\Ls=a_*$. \eqref{eq:dC} follows from rearranging the identity
$\kappa_*=\E\Ls^2=\frac{(\eta_*-\gamma_*)^2}{\gamma_*^2}(\eta_*^2 d_*^C-1)$.
 \eqref{eq:dD} then follows from the identity
$d_{*}^{D}+\left(\eta_{*}-\gamma_{*}\right) d_{*}^{C}=d_{*}^{A}$, 
and this then implies $\E \Es_b^2=b_*$ as well as $\E \D^2\Ls^2=c_*$.
Finally, \eqref{eq:dE} follows from the identity
$d_{*}^{E}+\left(\eta_{*}-\gamma_{*}\right)^{2}
d_{*}^{C}+2\left(\eta_{*}-\gamma_{*}\right) d_{*}^{D}=1$,
and this then implies $\E \D^2\Es_b^2=e_*$.
\end{proof}

\begin{Remark}\label{remark:non0param}
This shows $b_*,\kappa_*>0$ strictly, because $\D$ has strictly
positive variance by \Cref{AssumpD}, and hence so do $\Ls$ and $\Es_b$.
Also by Proposition \ref{prop:deltabound},
we have $\eta_*-\gamma_*>0$, hence $\delta_*,\sigma_*^2>0$.
\end{Remark}

\subsection{State evolution of stationary VAMP}\label{append:ampSE}
\begin{proof}[Proof of \Cref{prop:AMPparamconverge}]
Note that $\xi=Q\epsilon \sim N(0,I_{m \times m})$. Then $D^\top \xi \in \R^n$
may be written as the entrywise product of $D^\top 1_{m \times 1} \in \R^n$
and a vector $\bar{\xi} \sim N(0,I_{n \times n})$, both when $m \geq n$
and when $m \leq n$. The almost-sure convergence $H \toW \Hs$ is then
a straightforward consequence of Propositions \ref{prop:iidW}, \ref{prop:contW},
and \ref{prop:orthoW}, where all random variables of $\Hs$ have finite moments
of all orders under Assumptions \ref{AssumpA} and \ref{AssumpPrior}.
The identities $\kappa_*=\E \Ls^2$ and $b_*=\E \Es_b^2$ were shown in
(\ref{eq:abcde}).
\end{proof}

\begin{proof}[Proof of \Cref{thm:ampSE}]
We have $\delta_{11}=\E \Xs_1^2=\delta_*$ by the last identity of
(\ref{eq:denoiserpp2}). Supposing that $\delta_{tt}=\E \Xs_t^2=\delta_*$, we
have by definition $\E \Ys_t^2=\kappa_* \delta_{tt}=\sigma_*^2=\delta_*\kappa_*$.
Since $\Ys_t$ is independent of $\Es$, we have $\Ys_t+\Es \sim N(0,\sigma_*^2+b_*)$
where this variance is $\sigma_*^2+b_*=\gamma_*^{-1}$ by the definition of
$b_*$. Then $\E \Xs_{t+1}^2=\delta_*$ by the last identity of
(\ref{eq:denoiserpp2}), so $\E \Xs_t^2=\delta_*$ and $\E \Ys_t^2=\sigma_*^2$ for
all $t \geq 1$.

Noting that $\Delta_t$ is the upper-left submatrix of $\Delta_{t+1}$, let us
denote
\[\Delta_{t+1}=\begin{pmatrix} \Delta_t & \delta_t \\ \delta_t^\top & \delta_*
\end{pmatrix}\]
We now show by induction on $t$ the following three statements:
\begin{enumerate}
\item $\Delta_t \succ 0$ strictly.
\item We have
	\begin{equation}\label{eq:extraSE}
	    \mathsf{Y}_t=\sum_{k=1}^{t-1} \mathsf{Y}_{k}\left(\Delta_{t-1}^{-1}
\delta_{t-1}\right)_k+\mathsf{U}_t, \quad \mathsf{S}_t=\sum_{k=1}^{t-1}
\mathsf{S}_{k}\left(\Delta_{t-1}^{-1} \delta_{t-1}\right)_{k}+\mathsf{U}^{\prime}_t
	\end{equation}
	where $\mathsf{U}_t,\mathsf{U}_t'$ are Gaussian variables
with strictly positive variance, independent of
$\mathsf{H}$, $\left(\mathsf{Y}_{1}, \ldots, \mathsf{Y}_{t-1}\right)$, and
$\left(\mathsf{S}_{1}, \ldots, \mathsf{S}_{t-1}\right)$.
\item $\left(H, X_{t+1}, S_{t}, Y_t\right)
\stackrel{W}{\rightarrow}\left(\mathsf{H},
\mathsf{X}_{1}, \ldots, \mathsf{X}_{t+1}, 
\mathsf{S}_{1}, \ldots, \mathsf{S}_{t}, \mathsf{Y}_{1}, \ldots,
\mathsf{Y}_{t}\right)$.
\end{enumerate}

We take as base case $t=0$, where the first two statements are vacuous,
and the third statement requires $(H,x^1) \toW
(\Hs,\Xs_1)$ almost surely as $n \to \infty$.
Recall that $x^1=F(p^0,\st)$, and that
$F(p,\beta)$ is Lipschitz by Proposition \ref{prop:denoiserprop}.
Then this third statement follows from
Propositions \ref{prop:AMPparamconverge} and \ref{prop:contW}.

Supposing that these statements hold for some $t \geq 0$,
we now show that they hold for $t+1$. To show the first statement $\Delta_{t+1}
\succ 0$, note that for $t=0$ this follows from $\Delta_1=\delta_*>0$ by Remark
\ref{remark:non0param}. For $t \geq 1$,
given that $\Delta_{t}\succ0$, $\Delta_{t+1}$ is singular if and only if there exist constants $\alpha_{1}, \ldots, \alpha_{t} \in \mathbb{R}$ such that
$$
\Xs_{t+1}=F\left(\mathsf{Y}_{t}+\mathsf{E}, \Xstar\right)=\sum_{r=1}^t \alpha_{r} \mathsf{X}_{r}
$$
with probability 1. From the induction hypothesis,
$\mathsf{Y}_{t}=\sum_{k=1}^{t-1} \mathsf{Y}_{k}\left(\Delta_{r}^{-1}
\delta_{r}\right)_{k}+\mathsf{U}_t$ where $\mathsf{U}_t$ is independent of
$\mathsf{H},\mathsf{Y}_{1}, \ldots, \mathsf{Y}_{t-1}$ and hence also of
$\Es,\Xstar,\mathsf{X}_1,...,\mathsf{X}_t$. We now show that for any realized values
$(e_{0}, x_{0}, w_{0})$ of $$\left(\mathsf{E}+\sum_{k=1}^{t-1}
\mathsf{Y}_{k}\left(\Delta_{r}^{-1} \delta_{r}\right)_{k}, \quad \Xstar, \quad
\sum_{r=1}^t \alpha_{r} \mathsf{X}_{r}\right),$$ we have that
$\mathbb{P}\left(F\left(\mathsf{U}_t+e_{0}, x_{0}\right) \neq w_{0}\right)>0$. This would imply that $\Delta_{t+1}\succ 0$. From \Cref{prop:denoiserprop},
$f'(y,\gamma_*)$ is non-constant, so there exists $y \in \R$ such that
$f^{\prime}(y,\gamma_*) \neq \gamma_*/\eta_*$.
This implies that there exists some $u_{0} \in \mathbb{R}$ such that
$$F^{\prime}\left(u_{0}+e_{0}, x_{0}\right)=\frac{\eta_{*}}{\eta_{*}-\gamma_{*}}
f^{\prime}\left(u_{0}+e_{0}+x_{0},\gamma_*\right)-\frac{\gamma_{*}}{\eta_{*}-\gamma_{*}} \neq 0.$$
Then by the inverse function theorem, $F(u+e_0,x_0)=w_0$ has at most one
solution for $u$ in an open neighborhood of $u_0$. Since $\Us_t$ is Gaussian
with strictly positive variance,
this shows $\mathbb{P}\left(F\left(\mathsf{U}_t+e_{0}, x_{0}\right)
\neq w_{0}\right)>0$ as desired.
We thus have proved the first inductive statement that $\Delta_{t+1}\succ 0$.

 To study the empirical limit of $s_{t+1}$, let $U=\left(e_{b}, S_{t},
\Lambda S_{t}\right)$ and $V=\left(e, X_{t}, Y_{t}\right)$.
(For $t=0$, this is simply $U=e_b$ and $V=e$.) By the induction
hypothesis, the independence of $(\Ss_1,\ldots,\Ss_t)$ with $(\Es_b,\Ls)$,
and the identities $\E \Es_b^2=b_*$ and $\E \Ls=0$ and $\E \Ls^2=\kappa_*$,
almost surely as $n \to \infty$,
\begin{equation*}
	\frac{1}{n}\left(e_{b}, S_{t}, \Lambda S_{t}\right)^{\top}\left(e_{b}, S_{t}, \Lambda S_{t}\right) \rightarrow\left(\begin{array}{ccc}
		b_{*} & 0 & 0 \\
		0 & \Delta_{t} & 0 \\
		0 & 0 & \kappa_{*} \Delta_{t}
	\end{array}\right)\succ0
\end{equation*}
So almost surely for sufficiently large $n$, conditional on
$(H,X_{t+1},S_t,Y_t)$, the law of $s^{t+1}$ is given by its law conditioned on
$U=OV$, which is (see \cite[Lemma B.2]{fan2022tap})
\begin{equation}\label{eq:SEconv3}
	s^{t+1}\big|_{U=OV}=Ox^{t+1}\big|_{U=OV} \stackrel{L}{=}U\left(U^{\top}
U\right)^{-1} V^{\top} x^{t+1}+ \Pi_{U^{\perp}} \tilde{O} \Pi_{V^{\perp}}^{\top} x^{t+1}
\end{equation}
where $\tilde{O} \sim \Haar(\mathbb{SO}(n-(2 t+1)))$ and $\Pi_{U^{\perp}},
\Pi_{V^{\perp}} \in \mathbb{R}^{n \times(n-(2 t+1))}$ are matrices with
orthonormal columns spanning the orthogonal complements of the column spans of
$U,V$ respectively. We may replace $s^{t+1}$ by the right side of
\eqref{eq:SEconv3} without affecting the joint law of $\left(H, X_{t+1}, S_t,
Y_{t},s^{t+1}\right)$.

For $t=0$, we have $\E \Xs_1\Es=0$ since $\Xs_1$ is independent of $\Es$.
For $t \geq 1$, by the definition of $\Xs_{t+1}$, the
condition $\E F'(\Ps,\Xstar)=0$ from
(\ref{eq:denoiserpp2}), and Stein's lemma, we have $\E \Xs_{t+1}\Es=0$
and $\E \Xs_{t+1}\Ys_r=0$ for each $r=1,\ldots,t$. Then
by the induction hypothesis, almost surely as $n \rightarrow \infty$,
$$
\left(n^{-1}U^{\top} U\right)^{-1} \rightarrow\left(\begin{array}{ccc}
	b_{*} & 0 & 0 \\
	0 & \Delta_{t} & 0 \\
	0 & 0 & \kappa_{*} \Delta_{t}
\end{array}\right)^{-1}, \quad n^{-1}V^{\top} x_{t+1} \rightarrow\left(\begin{array}{c}
	0 \\
	\delta_{t} \\
	0
\end{array}\right).
$$
Then by (\ref{eq:SEconv3}) and
Propositions \ref{prop:combW} and \ref{prop:orthoW}, it follows that
$$
\left(H, X_{t+1}, S_{t}, Y_{t}, s^{t+1}\right) \stackrel{W}{\rightarrow}\left(\mathsf{H}, \mathsf{X}_{1}, \ldots, \mathsf{X}_{t+1}, \mathsf{S}_{1}, \ldots, \mathsf{S}_{t}, \mathsf{Y}_{1}, \ldots \mathsf{Y}_{t}, \sum_{r=1}^{t} \mathsf{S}_{r}\left(\Delta_{t}^{-1} \delta_{t}\right)_{r}+\mathsf{U}^{\prime}_{t+1}\right)
$$
where $\Us^{\prime}_{t+1}$ is the Gaussian limit of the second term on the right
side of (\ref{eq:SEconv3}) and is independent of $\mathsf{H},
\mathsf{X}_{1}, \ldots, \mathsf{X}_{t+1}, \mathsf{S}_{1}, \ldots,
\mathsf{S}_{t}, \mathsf{Y}_{1}, \ldots \mathsf{Y}_{t}$.
We can thus set $\Ss_{t+1}:=\sum_{r=1}^{t} \Ss_{r}\left(\Delta_{t}^{-1}
\delta_{t}\right)_{r}+\Us^{\prime}_{t+1}$. Then $(\Ss_1,\ldots,\Ss_{t+1})$ is
multivariate Gaussian and remains independent of $\Hs$ and
$(\Ys_1,\ldots,\Ys_t)$. Since $n^{-1}\|s^{t+1}\|^{2}=n^{-1}\|x^{t+1}\|^{2}
\rightarrow \delta_*$ almost surely as $n \rightarrow \infty$ by the induction
hypothesis, we have $\E\mathsf{S}_{t+1}^2=\delta_*$.
From the form of $\Ss_{t+1}$, we may check also
$\E\Ss_{t+1}(\Ss_1,\ldots,\Ss_t)=\delta_t$, so 
$(\Ss_1,\ldots,\Ss_{t+1})$ has covariance $\Delta_{t+1}$ as desired.
Furthermore $\sum_{r=1}^{t} \mathsf{~S}_{r}\left(\Delta_{t}^{-1}
\delta_{t}\right)_{r} \sim N\left(0, \delta_{t}^{\top} \Delta_{t}^{-1}
\delta_{t}\right)$. From $\Delta_{t+1} \succ 0$ 
and the Schur complement formula, $\delta_*-\delta_{t}^{\top} \Delta_{t}^{-1}
\delta_{t}>0$ strictly.
Then $\Us^{\prime}_{t+1}$ has strictly positive variance, since the variance of $\sum_{r=1}^{t} \mathsf{S}_{r}\left(\Delta_{t}^{-1} \delta_{t}\right)_{r}$ is less than the variance of $\Ss_{t+1}$. 
This proves the second equation in \eqref{eq:extraSE} for $t+1$.

Now, we study the empirical limit of $y^{t+1}$. Let $U=\left(e, X_{t+1},
Y_{t}\right)$, $V=\left(e_{b}, S_{t+1}, \Lambda S_{t}\right)$. Similarly
by the induction hypothesis and the empirical convergence of $(H,S_{t+1})$
already shown, almost surely as $n \rightarrow \infty$,
$$
\frac{1}{n}\left(e_b, S_{t+1}, \Lambda S_{t}\right)^{\top}\left(e_b, S_{t+1},
\Lambda S_{t}\right) \rightarrow\left(\begin{array}{ccc}
	b_{*} & 0 & 0 \\
	0 & \Delta_{t+1} & 0 \\
	0 & 0 & \kappa_{*} \Delta_{t}
\end{array}\right)\succ 0.
$$
Then the law of $y^{t+1}$ conditional on $(H,X_{t+1},S_{t+1},Y_t)$ is given by
its law conditioned on $U=O^\top V$, which is
\begin{equation}\label{eq:SEconv4}
	y^{t+1}\big|_{U=O^\top V}=O^{\top}\Lambda s^{t+1}\big|_{U=O^{\top} V}  \stackrel{L}{=}
U\left(V^{\top} V\right)^{-1} V^{\top} \Lambda s^{t+1}+\Pi_{U^{\perp}} \tilde{O}
\Pi_{V^{\perp}}^{\top} \Lambda s^{t+1}
\end{equation}
where $\tilde{O} \sim \Haar(\mathbb{SO}(n-(2 t+2)))$. From the convergence of
$(H,S_{t+1})$ already shown, almost surely as $n \rightarrow \infty$,
$$
\left(n^{-1} V^{\top} V\right)^{-1} \rightarrow\left(\begin{array}{ccc}
	b_{*} & 0 & 0 \\
	0 & \Delta_{t+1} & 0 \\
	0 & 0 & \kappa_{*} \Delta_{t}
\end{array}\right)^{-1}, \quad n^{-1} V^{\top} \Lambda s_{t+1} \rightarrow\left(\begin{array}{c}
	0 \\
	0 \\
	\kappa_{*} \delta_{t}
\end{array}\right).
$$
Then by (\ref{eq:SEconv4}) and Propositions \ref{prop:combW} and
\ref{prop:orthoW},
$$
\left(H, X_{t+1}, S_{t+1}, Y_{t}, y^{t+1}\right) \stackrel{W}{\rightarrow}\left(\mathsf{H}, \mathsf{X}_{1}, \ldots, \mathsf{X}_{t+1}, \mathsf{~S}_{1}, \ldots, \mathsf{S}_{t+1}, \mathsf{Y}_{1}, \ldots \mathsf{Y}_t, \sum_{r=1}^{t} \mathsf{Y}_{r}\left(\Delta_{t}^{-1} \delta_{t}\right)_{r}+\mathsf{U}_{t+1}\right)
$$
where $\Us_{t+1}$ is the limit of the second term on the right side of
(\ref{eq:SEconv4}), which is
Gaussian and independent of $ \mathsf{H},\Ss_{1}, \ldots,
\Ss_{t+1}, \Ys_{1}, \ldots \Ys_{t}$. Setting
$\Ys_{t+1}:=\sum_{r=1}^{t} \Ys_{r}\left(\Delta_{t}^{-1}
\delta_{t}\right)_{r}+\Us_{t+1}$, it follows that $(\Ys_1,\ldots,\Ys_{t+1})$
remains independent of $\Hs$ and $(\Ss_1,\ldots,\Ss_{t+1})$.
We may check that $\E \Ys_{t+1}(\Ys_1,\ldots,\Ys_t)=\kappa_* \delta_t$, and
we have also $n^{-1}\|y^{t+1}\|^{2}=n^{-1}\|\Lambda
s^{t+1}\|^{2} \rightarrow \kappa_{*}\delta_*$ so $\E\mathsf{Y}_{t+1}^2
=\kappa_{*}\delta_*$. From $\Delta_{t+1} \succ 0$ and the Schur complement
formula, note that $\sum_{r=1}^{t}
\mathsf{Y}_{r}\left(\Delta_{t}^{-1} \delta_{t}\right)_{r}$ has variance
$\kappa_{*} \delta_{t}^{\top} \Delta_{t}^{-1} \delta_{t}$ which is strictly
smaller than $\kappa_* \delta_*$, so $\Us_{t+1}$ has strictly positive variance.
This proves the first equation in \eqref{eq:extraSE} for $t+1$, and completes
the proof of this second inductive statement.

Finally, recall $x^{t+2}=F\left(y^{t+1}+e, \st\right)$ where $F$ is Lipschitz.
Then by Proposition \ref{prop:contW}, almost surely
$$
\left(H, X_{t+2}, S_{t+1}, Y_{t+1}\right)
\stackrel{W}{\rightarrow}\left(\mathsf{H}, \mathsf{X}_{1}, \ldots,
\mathsf{X}_{t+2}, \mathsf{~S}_{1}, \ldots, \mathsf{S}_{t+1}, \mathsf{Y}_{1},
\ldots, \mathsf{Y}_{t+1}\right)
$$
where $\mathsf{X}_{t+2}=F\left(\mathsf{Y}_{t+1}+\mathsf{E}, \Xstar\right)$,
showing the third inductive statement and completing the induction.
\end{proof}

\begin{proof}[Proof of \Cref{cor:ampSEcor1}]
This follows from the empirical Wasserstein convergence of $(e,X_t,Y_t)$
guaranteed by \Cref{thm:ampSE}.
The statements $n^{-1}X^\top e \to 0$ and $n^{-1} X^\top Y \to 0$ follow from
the identity $\E F'(\Ps,\Xstar)=0$ in (\ref{eq:denoiserpp2}) and Stein's
lemma, and the remaining statements follow directly from the independence of
$(\mathsf{Y}_1,\ldots,\mathsf{Y}_t)$ with $\mathsf{E}$ and from their specified
Gaussian laws.
\end{proof}

\begin{proof}[Proof of \Cref{cor:ampSEcor2}]
	This follows directly from \Cref{thm:ampSE}, the independence of
$(\Ss_1,\ldots,\Ss_t)$ and $\D$, and our definition of empirical Wasserstein
convergence.
\end{proof}

\section{Properties in high temperature}
We show uniqueness of the fixed point to (\ref{eq:fix})
and convergence of the stationary VAMP state evolution, assuming
the high temperature condition of Assumption \ref{AssumpHighTemp}.

\subsection{Fixed-point equation}\label{append:FPE}


\begin{proof}[Proof of \Cref{prop:uniquefix}]
Provided that the fixed point $(\eta_*^{-1},\gamma_*)$ is unique,
the statements $\eta_*^{-1} \leq \rho_*$ and $\eta_*-\gamma_* \geq \rho_*^{-1}$
were shown in Proposition \ref{prop:deltabound}.

To show uniqueness of this fixed point,
by the law of total variance, $\mmse(\gamma) \leq \Vpi \leq \pib$ for any
$\gamma>0$. Then
for all $\e<1/(2\pib)$, we have $\Vpi<G(-d_-)$ from Lemma \ref{lem:cauchy}(a),
and $-R(\eta^{-1})>0$ for all $\eta^{-1} \in (0,\rho_*]$
from Lemma \ref{lem:cauchy}(b). Extending $-R$ by continuity to
$-R(0)=-\E[-\D^2]=d_*>0$ via (\ref{eq:Rseries}), this shows that
\begin{equation}\label{eq:fixchangevar}
h(\eta^{-1}):=\mmse({-}R(\eta^{-1}))
\end{equation}
is a well-defined continuous map from $[0,\rho_*]$ to itself. Applying Stein's
lemma, the derivative of $\gamma \mapsto \mmse(\gamma)$ may be computed to be
\begin{equation}\label{eq:mmseprime}
\mmse'(\gamma)=-\E[\V[\Xstar \mid \Ys]^2],
\end{equation}
see e.g.\ \cite[Theorem 2]{payaro2009hessian} whose scalar specialization 
($m=1$ and $G=\sqrt{\gamma}$) yields
$\frac{d}{d\sqrt{\gamma}}\mmse(\gamma)=-2\sqrt{\gamma}\,\E[\V[\Xstar \mid
\Ys]^2]$, and (\ref{eq:mmseprime}) then follows from the chain rule.
Then the map (\ref{eq:fixchangevar}) has derivative
\begin{equation}\label{eq:hprime}
h'(\eta^{-1})=\E[\V[\Xstar \mid \Ys]^2] \cdot R'(\eta^{-1})
\end{equation}
where $\E,\V$ are with respect to the scalar channel
(\ref{eq:scalarBayes}) with inverse-variance $\gamma=-R(\eta^{-1})$. By the
condition (\ref{eq:poincare}) of Assumption \ref{AssumpPrior}, for any such
channel, $\E[\V[\Xstar \mid \Ys]^2] \leq C$ where $C>0$ depends only
on $\pib$. By Lemma \ref{lem:cauchy}(b--c),
$R'(\eta^{-1})>0$ and $R'(\eta^{-1})<O(\e^2)$ for all $\eta^{-1} \in (0,\Vpi)$.
Then for $\e<\e_0$ small enough,
we have $\sup_{\eta^{-1} \in (0,\Vpi)} h'(\eta^{-1}) \in (0,1)$ strictly.
Then $h(\cdot)$ defines a contractive map on $[0,\rho_*]$, so it has a unique
fixed point $\eta_*^{-1} \in [0,\rho_*]$ by the Banach fixed-point
theorem. We have $h(0)=\mmse(d_*)>0$ strictly, because $\pi$ has strictly
positive variance. Thus there is a unique fixed point
$(\eta_*^{-1},\gamma_*)$ to \eqref{eq:fix}
where $\eta_*^{-1} \in (0,\rho_*] \subset (0,G(-d_-))$ and
$\gamma_*=-R(\eta_*^{-1})>0$.
\end{proof}

\subsection{Explicit high temperature condition}\label{append:EHTC}
\begin{proof}[Proof of \Cref{remark:scalinginv}] 
	Suppose Assumptions \ref{AssumpD}--\ref{AssumpHighTemp} hold for the linear model in \eqref{eq:linearmodel}. Consider the rescaled problem: $y=A_{\rc}\st_{\rc}+\epsilon$ where $A_{\rc}:=\sqrt{\pib} A$ and $\st_{\rc}:=\frac{1}{\sqrt{\pib}} \st$. Note that $A_\rc=Q^\top D_\rc O$ with $D_\rc:=\sqrt{\pib} D$ and $D_\rc^\top 1_{m\times 1} \stackrel{W}{\to} \D_\rc$ with $\D_\rc := \sqrt{\pib} \D$.
	
	We then have that (i) Assumptions \ref{AssumpD} and \ref{AssumpA}  hold for the rescaled problem with $D_\rc$ and $\D_\rc$ in place of $D$ and $\D$; (ii) Assumption \ref{AssumpHighTemp} holds for the rescaled problem: $\mathrm{supp} (\D^2_{\rc}) \subseteq [d_{*,\rc}-\e_\rc, d_{*,\rc}+\e_\rc]$ with $d_{*,\rc}:=\pib d_*$ and $\e_\rc:=\pib \e$; (iii) Assumption \ref{AssumpPrior} holds for the rescaled problem with parameter $\pib_\rc=1$ in place of $\pib$. This follows from   \Cref{prop:priorconditions} and that the rescaled prior $\Xstar_\rc:=\frac{1}{\sqrt{\pib}}\Xstar$ either has bounded support contained in $[-1,1]$ or admits a density function $\sqrt{\pib} \exp{-g(\sqrt{\pib} x)}$ with $\frac{d^2}{dx^2} g(\sqrt{\pib} x)\ge 1$. 
	
	Apply \Cref{thm:maintheorem}--\Cref{thm:TAP} to the rescaled problem. Along with (ii), (iii) above, we obtain that there exists an absolute constant $\mathfrak{a}:=\e_0(\pib_\rc=1)>0$ such that if $\e_\rc := \pib \e\le \mathfrak{a}$, \eqref{eq:singleletter}, \eqref{eq:mmse} and \eqref{eq:TAPeq} hold for the rescaled problem. It is straightforward to show that \eqref{eq:singleletter}, \eqref{eq:mmse} and \eqref{eq:TAPeq} hold for the original problem if and only if they hold for the rescaled problem. The proof is now finished. 
\end{proof}

\subsection{Scalar parameters}

Let us record here the leading-order behaviors of several quantities related to
the scalar parameters of (\ref{eq:scalarparams}) and (\ref{eq:auxparams})
in the small parameter $\e$ of Assumption \ref{AssumpHighTemp}.
\begin{Proposition}\label{prop:betalimit}
Suppose Assumptions \ref{AssumpD}, \ref{AssumpPrior}, and \ref{AssumpHighTemp}
hold. Let $\kappa_2=\V(\D^2)$ and let $R(z)$ be the R-transform of $-\D^2$. 
For some constant $\e_0=\e_0(\pib)>0$, if $\e<\e_0$, then for any constant $c>0$
and all $z\in (0,c)$,
	\begin{equation}\label{eq:Rexpandlim}
	    {R}(z)=-d_*+\kappa_2 z(1+z\cdot O(\e)), \quad
{R}^{\prime}(z)=\kappa_2+O(\e^3), \quad {R}^{\prime
\prime}(z)=O(\e^{3}).
	\end{equation}
Furthermore
	\begin{equation}
		\begin{aligned}
			&\kappa_2 \leq \min(d_*\e,\e^2), \quad
\gamma_{*}=d_*-\kappa_{2} \eta_{*}^{-1}\left(1+\eta_{*}^{-1} \cdot O(\e)\right), \quad
			\kappa_*=\left(\frac{\eta_*-\gamma_*}{\eta_*}\right)^2
\frac{\kappa_2}{d_*^2} \qty(1+\eta_*^{-1} \cdot O(\e)),\\
&\alpha_{*}^{A}=\frac{\eta_{*}}{\sqrt{\kappa_{2}}}\left(1+\eta_{*}^{-1}\cdot
O(\e)\right), \quad \alpha_{*}^{B}=-\eta_{*}+O(\e),
\quad b_*=\frac{1}{d_*}\left(1+\eta_*^{-1} \cdot
O\left(\frac{\kappa_2}{d_*}\right)\right),\\
			& \frac{e_*}{b_*}=d_*+O(\e), \quad
\frac{a_*}{\sqrt{\kappa_*}}=O(\e), \quad \frac{c_*}{\kappa_*}=d_*+O(\e). 
		\end{aligned}
	\end{equation}
\end{Proposition}
\begin{proof}
\eqref{eq:Rexpandlim} and the bounds for $\kappa_2$ follow from Lemma
\ref{lem:cauchy}(c) and $\E[-\D^2]=-d_*$. We will use implicitly the bound
$\eta_*^{-1} \leq \Vpi \leq \pib$ from Proposition \ref{prop:uniquefix},
and hence $\eta_*^{-1}=O(1)$, throughout the proof.

Applying \eqref{eq:Rexpandlim} to the fixed point equation $\gamma_*=-R(\eta_*^{-1})$
in \eqref{eq:fix}, we have
\begin{equation}\label{eq:gammabound}
\gamma_{*}=d_*-\kappa_{2} \eta_{*}^{-1}\left(1+\eta_{*}^{-1} \cdot O(\e)\right).
\end{equation}
For the remaining bounds, let us first show
\begin{equation}\label{eq:Gsquaredbound}
\mathbb{E}\frac{\eta_{*}^2}{(\D^2+\eta_{*}-\gamma_{*})^{2}}
=1+\kappa_2\eta_*^{-2}(1+\eta_{*}^{-1}\cdot O(\e)).
\end{equation}
Note that $|\D^2-d_*| \leq \e$, $\kappa_2 \leq \e$, and
(\ref{eq:gammabound}) together imply $|\D^2-\gamma_*|=O(\e)$. Then for all $\e
\leq \e_0(\pib)$ sufficiently small,
\begin{equation}\label{eq:hee3}
\frac{\eta_{*}^{2}}{(\D^2+\eta_{*}-\gamma_{*})^{2}}=\left(\frac{1}{1-\frac{\gamma_{*}-\D^2}{\eta_{*}}}\right)^{2}=1+\sum_{k=1}^{\infty}
(k+1)\left(\frac{\gamma_{*}-\D^2}{\eta_{*}}\right)^k
\end{equation}
which is an absolutely convergent series. Let $\mu_j=\E[(d_*-\D^2)^j]$
be the $j^\text{th}$ central moment of $-\D^2$ (where $\mu_0=1$ and $\mu_1=0$),
which has the bound $|\mu_j| \leq \e^{j-2}\kappa_2$ for all $j \geq 2$
by Lemma \ref{lem:cauchy}(c). Applying this bound and
$|\gamma_*-d_*|=O(\kappa_2)=O(\e)$, for all $k \geq 3$ and a constant
$C=C(\pib)>0$ we have
\[\left|\mathbb{E}\left(\frac{\gamma_{*}-\mathsf{D}^{2}}{\eta_*}\right)^{k}\right|=\left|\sum_{j=0}^k
\binom{k}{j}\eta_*^{-k}(\gamma_*-d_*)^{k-j}\mu_j\right|
\leq 2^k\eta_*^{-k}(C\e)^{k-2}\kappa_2.\]
Then the summation of the terms of (\ref{eq:hee3}) for $k \geq 3$
is bounded by $\kappa_2 \eta_*^{-3} \cdot O(\e)$
for all $\e \leq \e_0(\pib)$ sufficiently small. For $k=1$ and $k=2$, we have
\begin{align*}
\mathbb{E}\left(\frac{\gamma_{*}-\mathsf{D}^{2}}{\eta_*}\right)&=\frac{\gamma_*-d_*}{\eta_*}=-\kappa_{2}\eta_{*}^{-2}+\kappa_{2}\eta_{*}^{-3}  \cdot O(\e) \\
\mathbb{E}\left(\frac{\gamma_{*}-\mathsf{D}^{2}}{\eta_*}\right)^{2}&=\frac{(\gamma_*-d_*)^2+\kappa_2}{\eta_*^2}=\kappa_{2}\eta_*^{-2}+\kappa_{2}\eta_{*}^{-4} \cdot O(\e)
\end{align*}
Applying these to (\ref{eq:hee3}) gives (\ref{eq:Gsquaredbound}) as claimed.

Now applying (\ref{eq:Gsquaredbound}) to the
definition of $\kappa_*$ in (\ref{eq:scalarparams}),
\begin{equation}\label{eq:fdr3}
\kappa_{*}=\left(\frac{\eta_{*}-\gamma_{*}}{\gamma_{*}}\right)^{2}
\frac{\kappa_{2}}{\eta_{*}^{2}}\left(1+\eta_*^{-1} \cdot O(\e)\right)
=\left(\frac{\eta_{*}-\gamma_{*}}{\eta_{*}}\right)^{2}
\frac{\kappa_2}{d_*^2}\left(1+\eta_{*}^{-1} \cdot O(\e)\right),
\end{equation}
the second equality applying $\gamma_*=d_*(1+O(\e))$ as implied by
(\ref{eq:gammabound}) and $\kappa_2=O(d_*\e)$.
Applying the first equality of (\ref{eq:fdr3}) to the definition of
$\alpha_*^A$ in (\ref{eq:auxparams}), and then applying this and
(\ref{eq:gammabound}) to the definition of $\alpha_*^B$,
\[\alpha_{*}^{A}=\frac{\eta_{*}}{\sqrt{\kappa_{2}}}(1+\eta_*^{-1} \cdot O(\e)),
\qquad \alpha_{*}^{B}=\frac{\eta_*^2}{\kappa_2}
\cdot (\gamma_*-d_*) \cdot (1+\eta_*^{-1}O(\e))=-\eta_{*}+O(\e).\]
Inverting (\ref{eq:gammabound}) and applying $\kappa_2=d_*\e$, we have
$\gamma_*^{-1}=d_*^{-1}(1+\eta_*^{-1}O(\kappa_2/d_*))$. Then applying
(\ref{eq:fdr3}) and $(\eta_*-\gamma_*)/\eta_* \in (0,1)$ where this lower bound
of 0 follows from Proposition \ref{prop:uniquefix}, we obtain
from the definition of $b_*$ in (\ref{eq:scalarparams}) that
\[b_{*}=\frac{1}{\gamma_{*}}-\frac{\kappa_{*}}{\eta_{*}-\gamma_{*}}
=\frac{1}{d_*}\left(1+\eta_*^{-1} \cdot O\left(\frac{\kappa_2}{d_*}\right)
\right)+\frac{1}{d_*\eta_*} \cdot O\left(\frac{\kappa_2}{d_*}\right)
=\frac{1}{d_*}\left(1+\eta_*^{-1} \cdot O\left(\frac{\kappa_2}{d_*}\right)
\right).\]
Applying $\kappa_2/d_*=O(\e)$, we have from (\ref{eq:fdr3}) that
$\kappa_*=d_*^{-1}O(\e)$. Then
\[\frac{e_*}{b_*}=\frac{1+\kappa_{*}}{b_{*}}=
d_*(1+O(\kappa_2/d_*))(1+d_*^{-1}O(\e))=d_*+O(\e).\]
Applying $d_*\gamma_*^{-1}=1+\eta_*^{-1}O(\kappa_2/d_*)$ and (\ref{eq:fdr3}),
we have from the definition of $a_*$ in (\ref{eq:auxparams}) that
\[\frac{a_*}{\sqrt{\kappa_*}}
=\frac{(\eta_*-\gamma_*)(1-d_*\gamma_*^{-1})}{\sqrt{\kappa_*}}
=O\left(\frac{(\eta_*-\gamma_*)/(d_*\eta_*) \cdot
\kappa_2}{\sqrt{\kappa_*}}\right)=O(\sqrt{\kappa_2})=O(\e).\]
Applying $d_*\gamma_*^{-1}=1+\eta_*^{-1}O(\kappa_2/d_*)=1+O(\e)$,
$d_*-\gamma_*=\kappa_2\eta_*^{-1}(1+\eta_*^{-1}O(\e))$, and
(\ref{eq:fdr3}), we have from the definition of $c_*$ in (\ref{eq:auxparams}) that
\begin{align*}
\frac{c_{*}}{\kappa_{*}}&=-\left(\eta_{*}-\gamma_{*}\right)+\frac{1}{\kappa_{*}}\left(\frac{\eta_{*}-\gamma_{*}}{\gamma_{*}}\right)^{2}\left(d_*-\gamma_{*}\right)
=-(\eta_{*}-\gamma_{*})+\frac{\eta_*^2(d_*-\gamma_*)}{\kappa_2}
(1+\eta_*^{-1} \cdot O(\e))\\
&=-(\eta_{*}-\gamma_{*})+\eta_*(1+\eta_*^{-1} \cdot O(\e))=d_*+O(\e).
\end{align*}
\end{proof}

\subsection{Convergence of stationary VAMP}\label{appendix:VAMPconvergence}

\begin{Lemma}\label{lemma:ampcondition}
Recall the replica-symmetric potential $i_{\mathrm{RS}}(\eta^{-1},\gamma)$
from (\ref{eq:iRS}), and let $(\eta_*^{-1},\gamma_*) \in (0,G(-d_-)) \times
\R_+$ be any fixed point of (\ref{eq:fix}) for which $\gamma_*$ is a local
minimizer of
\[\gamma \mapsto \sup_{\eta^{-1} \in (0,G(-d_-))}
i_{\mathrm{RS}}(\eta^{-1},\gamma)\]
Then the conclusion $\lim_{\min(s,t)\to \infty} \delta_{st} = \delta_*$
of Proposition \ref{prop:convsmallbeta} holds for the
stationary initialization (\ref{eq:VAMPstatinit}).
\end{Lemma}
\begin{proof}
Recall that $\delta_{tt}=\delta_*$ for all $t \geq 1$ from Theorem
\ref{thm:ampSE}. Then
	$\delta_{s t}=\mathbb{E}\left[\mathsf{X}_{s} \mathsf{X}_{t}\right] \leq
\sqrt{\mathbb{E}\left[\mathsf{X}_{s}^{2}\right]
\mathbb{E}\left[\mathsf{X}_{t}^{2}\right]} =\delta_*$ for all $s,t \geq 1$.
For $s=1$ and any $t \geq 2$, observe also that
\begin{equation}\label{eq:delta1t}
	\begin{gathered}
		\delta_{1t}=\mathbb{E} \Xs_{1} \mathsf{X}_{t}=\mathbb{E}\left[F\left(\mathsf{P}_{0}, \Xstar\right) F\left(\mathsf{Y}_{t-1}+\mathsf{E}, \Xstar\right)\right]=\mathbb{E}\left[\mathbb{E}\left[F\left(\mathsf{P}_{0}, \Xstar\right) F\left(\mathsf{Y}_{t-1}+\mathsf{E}, \Xstar\right) \mid \Xstar\right]\right] \\
		=\mathbb{E}[\mathbb{E}\left[F\left(\mathsf{P}_{0}, \Xstar\right) \mid \Xstar\right]^{2}] \geq 0
	\end{gathered}
\end{equation}
where the last equality holds because
$\mathsf{P}_{0}$, $\mathsf{Y}_{t-1}+\Es$, and $\Xstar$ are independent,
with $\Ps_0$ and $\Ys_{t-1}+\Es$ equal in law
(by the identity $\sigma_*^2+b_*=\gamma_*^{-1}$).

Consider now the map $\delta_{st} \mapsto \delta_{s+1,t+1}$.
Recalling that $\E \Ys_t^2=\sigma_{*}^{2}$ and $\E \Ys_s\Ys_t=\kappa_*
\delta_{st}$, we may represent
	$$
	\left(\mathsf{Y}_{s}+\mathsf{E}, \mathsf{Y}_{t}+\mathsf{E}\right)\stackrel{L}{=}\left(\sqrt{\kappa_{*} \delta_{s t}+b_{*}} \mathsf{G}+\sqrt{\sigma_{*}^{2}-\kappa_{*} \delta_{s t}} \mathsf{G}^{\prime}, \sqrt{\kappa_{*} \delta_{s t}+b_{*}} \mathsf{G}+\sqrt{\sigma_{*}^{2}-\kappa_{*} \delta_{s t}} \mathsf{G}^{\prime \prime}\right)
	$$
	where $\mathsf{G}, \mathsf{G}^{\prime}, \mathsf{G}^{\prime \prime}$ are
jointly independent standard Gaussian variables.
Denote
$$\mathsf{P}_{\delta}^{\prime}:=\sqrt{\kappa_{*} \delta+b_{*}} \cdot \mathsf{G}+\sqrt{\sigma_{*}^{2}-\kappa_{*} \delta} \cdot \mathsf{G}^{\prime}, \quad \mathsf{P}_{\delta}^{\prime \prime}:=\sqrt{\kappa_{*} \delta+b_{*}} \cdot \mathsf{G}+\sqrt{\sigma_{*}^{2}-\kappa_{*} \delta} \cdot \mathsf{G}^{\prime \prime}$$
and define $g:\left[0,\delta_*\right] \to \mathbb{R}$ by
$g(\delta):=\mathbb{E}\left[F\left(P_\delta',\Xstar\right)
F\left(P_\delta'',\Xstar\right)\right]$. Then $\delta_{s+1,t+1}=g(\delta_{st})$.

We claim that for any $\delta \in [0,\delta_*]$, we have $g(\delta) \geq 0$,
$g'(\delta) \geq 0$, and $g''(\delta) \geq 0$. The first bound $g(\delta) \geq 0$
follows from
$$
g(\delta)=\mathbb{E}\Big[\E[F\left(\mathsf{P}_{\delta}^{\prime}, \Xstar\right)
F\left(\mathsf{P}_{\delta}^{\prime \prime}, \Xstar\right) \mid
\Xstar,\Gs]\Big]=\mathbb{E}\left[\mathbb{E}\left[F\left(\mathsf{P}_{\delta}^{\prime},
\Xstar\right) \mid \Xstar, \mathsf{G}\right]^{2}\right] \geq 0,
$$ 
because $\Ps_\delta',\Ps_\delta''$ are independent and equal in law conditional
on $\Gs,\Xstar$. Differentiating in $\delta$ and applying
Gaussian integration by parts,
	$$
	\begin{aligned}
		g^{\prime}(\delta) &=2 \mathbb{E}\left[F^{\prime}\left(\mathsf{P}_{\delta}^{\prime}, \Xstar\right) F\left(\mathsf{P}_{\delta}^{\prime \prime}, \Xstar\right)\left(\frac{\kappa_{*}}{2 \sqrt{\kappa_{*} \delta+b_{*}}} \cdot \mathsf{G}-\frac{\kappa_{*}}{2 \sqrt{\sigma_{*}^{2}-\kappa_{*} \delta}} \cdot \mathsf{G}^{\prime}\right)\right] \\
		&=\frac{\kappa_{*}}{\sqrt{\kappa_{*} \delta+b_{*}}} \mathbb{E}\left[F^{\prime}\left(\mathsf{P}_{\delta}^{\prime}, \Xstar\right) F\left(\mathsf{P}_{\delta}^{\prime \prime}, \Xstar\right) \mathsf{G}\right]-\frac{\kappa_{*}}{\sqrt{\sigma_{*}^{2}-\kappa_{*} \delta}} \mathbb{E}\left[F^{\prime}\left(\mathsf{P}_{\delta}^{\prime}, \Xstar\right) F\left(\mathsf{P}_{\delta}^{\prime \prime}, \Xstar\right) \mathsf{G}^{\prime}\right] \\
		&=\kappa_{*}\mathbb{E}\left[F^{\prime \prime}\left(\mathsf{P}_{\delta}^{\prime}, \Xstar\right) F\left(\mathsf{P}_{\delta}^{\prime \prime}, \Xstar\right)+F^{\prime}\left(\mathsf{P}_{\delta}^{\prime}, \Xstar\right) F^{\prime}\left(\mathsf{P}_{\delta}^{\prime \prime}, \Xstar\right)\right]
		-\kappa_{*}\mathbb{E}\left[F^{\prime \prime}\left(\mathsf{P}_{\delta}^{\prime}, \Xstar\right) F\left(\mathsf{P}_{\delta}^{\prime \prime}, \Xstar\right)\right] \\
		&=\kappa_{*}
\mathbb{E}\left[F^{\prime}\left(\mathsf{P}_{\delta}^{\prime}, \Xstar\right)
F^{\prime}\left(\mathsf{P}_{\delta}^{\prime \prime}, \Xstar\right)\right].
	\end{aligned}
	$$
Then $g'(\delta)=\kappa_*\E\left[\E[F'(\Ps_\delta',\Xstar) \mid
\Gs,\Xstar]^2\right] \geq 0$, and a similar argument shows $g''(\delta) \geq 0$.

Observe that at $\delta=\delta_*$, we have
$\mathsf{P}_{\delta_{*}}^{\prime}=\mathsf{P}_{\delta_{*}}^{\prime
\prime}=\sqrt{\sigma_*^2+b_{*}} \cdot \mathsf{G}=\mathsf{G}/\sqrt{\gamma_*}$
which is equal in law to $\Ps \sim N(0,\gamma_*^{-1})$. Then
$g(\delta_{*})=\mathbb{E}[F(\Ps,\Xstar)^{2}]=\delta_{*}$
by (\ref{eq:denoiserpp2}). So $g:[0,\delta_*] \to
[0,\delta_*]$ is a nonnegative, increasing, convex function with a fixed point
at $\delta_*$. We claim that
\begin{equation}\label{eq:gprimebound}
g'(\delta_*)<1
\end{equation}
This then implies that $\delta_*$ is the unique fixed point of $g(\cdot)$
over $[0,\delta_*]$, and $\lim_{t \to \infty} g^{(t)}(\delta)=\delta_*$ for
any $\delta \in [0,\delta_*]$. Observe from (\ref{eq:delta1t})
that $\delta_{1t}=\delta_{12}$ for all $t \geq 2$, so
$\delta_{t,t+s}=g^{(t-1)}(\delta_{1,1+s})=g^{(t-1)}(\delta_{12})$
for any $s \geq 1$. Then $\lim_{\min(s,t) \to \infty} \delta_{st}=\delta_*$
follows.

It remains to show (\ref{eq:gprimebound}). For this,
applying (\ref{eq:denoiserpp1}) and
$\E[F'(\Ps,\Xstar)]=0$ by (\ref{eq:denoiserpp2}), we have
	\begin{equation}\label{eq:ampcondgp}
g'(\delta_*)=\kappa_* \E[F'(\Ps_{\delta_*}',\Xstar)^2]
=\kappa_*\left(\frac{\eta_*^2}{(\eta_*-\gamma_*)^2}\E[f'(\Ps+\Xstar,\gamma_*)^2]-\frac{\gamma_*^2}{(\eta_*-\gamma_*)^2}\right).
	\end{equation}
Recall the replica-symmetric potential $i_{\mathrm{RS}}$ from (\ref{eq:iRS}),
whose gradient and Hessian in $(\gamma,\eta^{-1})$ are given by
\[\nabla i_{\mathrm{RS}}(\eta^{-1},\gamma)
=\frac{1}{2}\Big(\mmse(\gamma)-\eta^{-1} \;\; -R(\eta^{-1})-\gamma
\Big),
\qquad \nabla^2 i_{\mathrm{RS}}=\frac{1}{2} \begin{pmatrix}
\mmse'(\gamma) & -1 \\ -1 & -R'(\eta^{-1}) \end{pmatrix}\]
where we have used the I-MMSE relation $i'(\gamma)=\frac{1}{2}\mmse(\gamma)$
\cite{guo2005mutual}. The condition that $(\eta_*^{-1},\gamma_*)$ is a fixed
point of (\ref{eq:fix}) implies $\nabla
i_{\mathrm{RS}}(\eta_*^{-1},\gamma_*)=0$. Here, $\eta^{-1} \mapsto
i_{\mathrm{RS}}(\eta^{-1},\gamma_*)$ is concave because $-R'(\eta^{-1})>0$ by
Lemma \ref{lem:cauchy}(b), so $\eta_*^{-1}$ is the global maximizer of this
function. Then, the given assumption that $\gamma_*$ is a local minimizer of
$\gamma \mapsto \sup_{\eta^{-1}} i_{\mathrm{RS}}(\eta^{-1},\gamma)$ implies
the Schur-complement condition for $\nabla^2 i_{\mathrm{RS}}$
\begin{equation}\label{eq:schurcomplement}
\mmse'(\gamma)+\frac{1}{R'(\eta^{-1})} \geq 0.
\end{equation}
Differentiating $R(z)=G^{-1}(z)-z^{-1}$ where
$G(z)=\E[(z+\D^2)^{-1}]$, and applying $\kappa_*$ from (\ref{eq:scalarparams}),
\[R'(\eta_*^{-1})=\frac{1}{G'(G^{-1}(\eta_*^{-1}))}+\eta_*^2
=\frac{1}{G'(\eta_*-\gamma_*)}+\eta_*^2
=\eta_*^2-\frac{1}{\E[(\D^2+\eta_*-\gamma_*)^{-2}]}
=\eta_*^2-\frac{\eta_*^2}{(\frac{\eta_*}{\eta_*-\gamma_*})^2\kappa_*+1}.\]
Then
\[\frac{1}{R'(\eta_*^{-1})}=\frac{1}{\eta_*^2}\left(1+\frac{(\eta_*-\gamma_*)^2}{\eta_*^2\kappa_*}\right).\]
Recalling also $\mmse'(\gamma)=-\E[\V[\Xstar \mid \Ys]^2]
=-\gamma_*^{-2}\E[f'(\Ps+\Xstar,\gamma_*)^2]$ from (\ref{eq:mmseprime})
and Proposition \ref{prop:denoiserlipschitz}, the condition
(\ref{eq:schurcomplement}) may be rearranged as
\[\E[f'(\Ps+\Xstar,\gamma_*)^2] \leq 
\frac{\gamma_*^2}{\eta_*^2}\left(1+\frac{(\eta_*-\gamma_*)^2}{\eta_*^2\kappa_*}\right).\]
Substituting into (\ref{eq:ampcondgp}) gives $g'(\delta_*) \leq
\gamma_*^2/\eta_*^2<1$ where the second inequality applies $\eta_*-\gamma_*>0$
from Proposition \ref{prop:uniquefix}. This shows the desired claim
(\ref{eq:gprimebound}), concluding the proof.
\end{proof}

\begin{proof}[Proof of \Cref{prop:convsmallbeta}]
For $\e<\e_0$ sufficiently small, we have shown in the proof of
Proposition \ref{prop:uniquefix} that $h'(\eta_*^{-1})<1$ strictly where
$h'(\eta^{-1})=-\mmse'(\gamma) \cdot R'(\eta^{-1})$ is as defined in
(\ref{eq:hprime}). Rearranging this gives
$\mmse'(\gamma_*)+1/R'(\eta_*^{-1})>0$, i.e.\ $\gamma_*$ is a (strict)
local minimizer of $\gamma \mapsto \sup_{\eta^{-1}}
i_{\mathrm{RS}}(\eta^{-1},\gamma)$, so the result follows from
\Cref{lemma:ampcondition}.
\end{proof}

\section{Analysis of the conditional first moment}\label{appendix:variational}

In this appendix we prove Lemmas \ref{lemma:firstmt} and \ref{lemma:analysisf}.
The arguments are extensions of those of
\cite[Lemmas 3.1 and 3.2]{fan2021replica}, and we will refer to
\cite{fan2021replica} for some of the technical details.

\begin{Lemma}\label{lem:abar}
Let $\pi$ be any probability distribution over $\R$. Let $c_\pi(a,b)$ be as
defined in (\ref{eq:cpi}), and let
\[\bar{a}=\sup\left\{a \in \R:\int
e^{ax^2}d\pi(x)<\infty\right\} \in [0,\infty].\]
Then for all $a>\bar{a}$ and $b \in \R$,
we have $c_\pi(a,b)=\infty$. For each fixed $a<\bar{a}$,
the function $b \mapsto \log c_\pi(a,b)$ is continuous and satisfies,
for some $(a,\pi)$-dependent constant $C>0$ and for all $b \in \R$,
\[\log c_\pi(a,b) \leq C(b^2+1).\]
\end{Lemma}
\begin{proof}
For the first statement,
suppose $\bar{a}<\infty$, and consider $a>\bar{a}$. Taking $\varepsilon>0$
such that $a-\varepsilon>\bar{a}$, we have
\[c_\pi(a,b) \geq
\int_{|x| \geq |b|/\varepsilon} e^{ax^2+bx}d\pi(x)
\geq \int_{|x| \geq |b|/\varepsilon} e^{(a-\varepsilon)x^2}d\pi(x)=\infty.\]
For the second statement, consider any $\bar{a} \in [0,\infty]$ and fix
$a<\bar{a}$. Taking $\varepsilon>0$ such that $a+\varepsilon<\bar{a}$, we have
\begin{align*}
c_\pi(a,b)&=\int_{|x|<b/\varepsilon} e^{ax^2+bx}d\pi(x)
+\int_{|x| \geq b/\varepsilon} e^{ax^2+bx}d\pi(x)\\
&\leq e^{b^2/\varepsilon}\int e^{ax^2}d\pi(x)+\int e^{(a+\varepsilon)x^2}d\pi(x)
\leq C(e^{b^2/\varepsilon}+1) \leq 2Ce^{b^2/\varepsilon}
\end{align*}
for a constant $C=C(a,\varepsilon,\pi)>0$. Then $\log c_\pi(a,b) \leq
C'(b^2+1)$ for a constant $C'=C'(a,\varepsilon,\pi)>0$. In particular,
$\log c_\pi(a,b)<\infty$,
and continuity in $b$ follows from a standard application of the dominated
convergence theorem.
\end{proof}

\begin{proof}[Proof of Lemma \ref{lemma:firstmt}]
	Recall the $n\times t$ matrices $X_{t}=\left(x^{1}, \ldots,
x^{t}\right)$, $Y_{t}=\left(y^{1}, \ldots, y^{t}\right)$, and $S_{t}=\left(s^{1}, \ldots, s^{t}\right)$ which collect the AMP iterates. We fix $t$ and write $\mathcal{G}, X, Y, S, \Delta$ for $\mathcal{G}_{t}, X_{t}, Y_{t}, S_{t}, \Delta_{t}$.

From the definition of $\Z(\cU)$ in \eqref{eq:defpartitionfunc},
applying $y=A\st+\epsilon=Q^\top DO\st+\epsilon$ and $\xi=Q\epsilon$,
and writing as shorthand $\x:=\x(\s)=\s-\st$, we have
	\begin{subequations}\label{eq:Econdstart}
		\begin{align}
			&\mathbb{E}[\Z(\cU) \mid \cG]=\int
\mathbb{I}\Big(\frac{1}{n}\|\x\|^2 \in \cU\Big)
\cdot \exp \left(-\frac{\|\epsilon\|^{2}}{2}+\frac{n}{2} \cdot f_{n}(\x)\right)
\prod_{i=1}^{n} d\pi(\s_{i})\\
			&\text{with}\quad f_{n}(\x):=\frac{2}{n} \log
\mathbb{E}\left[\exp \left(-\frac{\x^{\top} O^{\top} \DDbar O \x}{2}+\x^{\top}
O^{\top} D^\top \xi\right) \;\bigg|\;\mathcal{G}\right].\label{eq:fndef}
		\end{align}
	\end{subequations}

Conditional on $\cG$, the only random quantity in the expectation in
(\ref{eq:fndef}) is the matrix $O$. By definition of $\cG$ in
(\ref{eq:defGt}) and the AMP iterations in (\ref{eq:xsy}),
its conditional law is that of a $\operatorname{Haar}(\SO(n))$
matrix $O$ conditioned on the event
	$(e_b, S, \Lambda S)=O(e, X, Y)$.
Then by \cite[Lemma B.2]{fan2022tap}, we may represent this conditional law of $O$ as
\begin{equation}\label{eq:theVV}
	\left.O\right|_{\cG} \stackrel{L}{=}\left(e_{b}, S, \Lambda
S\right)\begin{pmatrix} 
		e^{\top} e & e^{\top} X & e^{\top} Y \\
		X^{\top} e & X^{\top} X & X^{\top} Y \\
		Y^{\top} e & Y^{\top} X & Y^{\top} Y
	\end{pmatrix}^{-1}(e, X, Y)^{\top}+\Pi_{\left(e_{b}, S, \Lambda
S\right)^{\perp}} \tilde{O} \Pi_{(e, X, Y)^{\perp}}^{\top}
\end{equation}
where $\Pi_{(e, X, Y)^{\perp}},\Pi_{\left(e_{b}, S, \Lambda S\right)^{\perp}}\in
\R^{n\times (n-2t-1)}$ have orthonormal columns orthogonal to the column spans
of $(e, X,Y),(e_b,S,\Lambda S)\in \R^{n\times (2t+1)}$ respectively,
and $\tilde{O} \sim \Haar(\SO(n-2t-1))$ is independent of $\cG$.
We remark that the matrix inverse in (\ref{eq:theVV}) is well-defined almost
surely for all large $n$, by Corollary \ref{cor:ampSEcor1} and the statements
$\Delta \succ 0$ in Theorem \ref{thm:ampSE} and $b_*,\kappa_*>0$
strictly in Remark \ref{remark:non0param}.
Writing as shorthand
$\Pi=\Pi_{\left(e_{b}, S, \Lambda S\right)^{\perp}}$ and
\begin{equation}\label{eq:sigmaparper}
\x_{\perp}=\Pi_{(e, X, Y)^{\perp}}^{\top} \x \in \mathbb{R}^{n-2 t-1}, \qquad
\x_{\|}=\left(e_{b}, S, \Lambda S\right)\begin{pmatrix}
	e^{\top} e & e^{\top} X & e^{\top} Y \\
	X^{\top} e & X^{\top} X & X^{\top} Y \\
	Y^{\top} e & Y^{\top} X & Y^{\top} Y
\end{pmatrix}^{-1}(e, X, Y)^{\top} \x \in \mathbb{R}^{n},
\end{equation}
this yields the equality in conditional law $\left.O \x\right|_{\cG}
\stackrel{L}{=} \Pi \tilde{O} \x_{\perp}+\x_{\|}$. So $f_n(\x)$ in
\eqref{eq:fndef} is given by
    \begin{align}
f_{n}(\x)&=\frac{2}{n} \log \mathbb{E} \exp \left(-\frac{(\Pi \tilde{O}
\x_{\perp}+\x_{\|})^{\top} \DDbar(\Pi \tilde{O}
\x_{\perp}+\x_{\|})}{2}+\left(\Pi \tilde{O} \x_{\perp}+\x_{\|}\right)^{\top}
D^\top \xi\right) \label{eq:fnform}\\
&=-\frac{1}{n} \x_{\|}^{\top} \DDbar \x_{\|}+\frac{2}{n} \x_{\|}^{\top} D^\top
\xi+\frac{2}{n} \log \mathbb{E} \exp \left(-\frac{\x_{\perp}^{\top}
\tilde{O}^{\top} \Pi^{\top} \DDbar \Pi \tilde{O} \x_{\perp}}{2}+\left(\xi^{\top}
D \Pi-\x_{\|}^{\top} \DDbar \Pi\right) \tilde{O} \x_{\perp}\right)\nonumber
\end{align}
where this expectation is over only $\tilde{O} \sim \Haar(\SO(n-2t-1))$.\\

{\bf Uniform approximation of $f_n(\x)$.} We proceed to approximate $f_n(\x)$
by low-dimensional functions of $\x$ for large $n$. Define
$\Pfrak(\x)=(u(\x),r(\x),v(\x),w(\x))$ where
\begin{equation}\label{eq:ruvwdef}
    u(\x)=\frac{1}{n}\|\x\|^{2},\qquad \mqty(r(\x)\\ v(\x) \\
w(\x))=\qty[\frac{1}{n} \mqty(e^{\top} e & e^{\top} X & e^{\top} Y \\ X^{\top} e
& X^{\top} X & X^{\top} Y\\ Y^\top e & Y^\top X & Y^\top Y)]^{-1/2}\cdot
\frac{1}{n} (e,X,Y)^\top \x.
\end{equation}
Here, $u(\x), r(\x)\in \R$ and $v(\x),w(\x)\in \R^t$. Note that
$O\x_{\|}$ is the orthogonal projection of $\x$ onto the column span of
$(e,X,Y)$, so
\begin{equation}\label{eq:lengthre}
    r^{2}(\x)+\|v(\x)\|^{2}+\|w(\x)\|^{2}
=\frac{\|\x_{\|}\|^2}{n}=u(\x)-\frac{\left\|\x_{\perp}\right\|^{2}}{n}.
\end{equation}
Let $K>0$ be the bound given in the lemma for which $\cU \subseteq (0,K)$,
and define the open domain
\begin{equation}
\Kc=\left\{\x \in \R^n:u(\x) \in (0,K),\,\Ac(\Pfrak(\x))>0\right\}
\end{equation}
where $\Ac(\Pfrak)=u-r^2-\|v\|^2-\|w\|^2$ as defined in (\ref{eq:defCst}).
Restricting first to $\x \in \Kc$,
we apply \cite[Proposition 2.7]{fan2021replica} to approximate the
expectation over $\tilde{O}$. (This is stated
in \cite{fan2021replica} for $\Haar(\O(n))$, but the result and proof hold
verbatim for $\SO(n)$.) The needed conditions of
\cite[Proposition 2.7]{fan2021replica} are verified by the Cauchy interlacing
of eigenvalues of $\Pi^\top D^\top D \Pi$ with those of $D^\top D$,
the convergence assumptions (\ref{eq:Assump2D}), the bounds
$\|\x_{\perp}\|^2,\|\x_{\|}\|^2 \leq \|\x\|^2<nK$ for $\x \in \Kc$,
the bound $\|D\xi\|^2<2d_+ n$ almost surely for all large $n$,
and the observation that by Lemma \ref{lem:cauchy}(c),
for $\e<\e_0(\pib,K)$ sufficiently small and for some sufficiently small
constant $\df>0$,
    \begin{equation}\label{eq:Gddfbig}
        G(-d_-+\df)-\df>K.
    \end{equation}
(We allow $\df>0$ to depend on $\e$ and the law of $\D^2$,
and we will eventually take $\df \to 0$.)
For scalars or vectors $a_n(\x)$ and $b_n(\x)$ whose dimensions
are independent of $n$, let us write $a_n(\x) \doteq b_n(\x)$ to mean, almost
surely,
\[\lim_{n,m \to \infty} \sup_{\x \in \Kc} |a_n(\x)-b_n(\x)|=0, \qquad
\lim_{n,m \to \infty} \sup_{\x \in \Kc} \|a_n(\x)-b_n(\x)\|=0.\]
Then \cite[Proposition 2.7]{fan2021replica} applied to the expectation over
$\tilde{O}$ yields
\begin{equation}\label{eq:fnapp}
  f_n(\x) \doteq -\frac{1}{n} \x_\|^\top \DD \x_\|+\frac{2}{n} \x^\top_\|
D^\top \xi+E_n(\x)
\end{equation}
where 
\begin{equation}\label{eq:En}
\begin{aligned}
E_n(\x)=\inf_{\tmr \ge
-d_-+\df}\bigg\{\frac{\tmr\left\|\x_{\perp}\right\|^{2}}{n}&+\frac{(\xi^{\top} D
\Pi-\x_{\|}^{\top}\DD \Pi)(\tmr I+\Pi^{\top} \DD \Pi)^{-1}(\Pi^{\top} D^\top
\xi-\Pi^{\top} \DD \x_{\|})}{n}\\
& -\frac{1}{n} \log \det \left(\tmr I+\Pi^{\top} \DD \Pi\right)-\Big(1+\log
\frac{\left\|\x_{\perp}\right\|^{2}}{n}\Big)
\bigg\}.
\end{aligned}
\end{equation}

For the first term of (\ref{eq:fnapp}), applying \Cref{cor:ampSEcor1}
and the preceding definitions of $\x_{\|}$ and $r,v,w$,
\begin{equation}\label{eq:xparlim}
\x_\|=e_{b}\big(b_{*}^{-1/2} r(\x)\big)+S \Delta^{-1/2} v(\x)+\Lambda
S\left(\kappa_{*} \Delta\right)^{-1/2} w(\x)+r_n(\x)
\end{equation}
where $n^{-1}\|r_n(\x)\|^2 \to 0$ uniformly over $\x \in \Kc$.
It follows from \Cref{thm:ampSE}, Corollaries \ref{cor:ampSEcor1} and
\ref{cor:ampSEcor2}, $\E\D^2=d_*$,
and (\ref{eq:abcde}) that almost surely,
\begin{equation}\label{eq:SDSlim}
\lim_{n,m \to \infty} \frac{1}{n} (e_b, S, \Lambda S)^\top \DD (e_b, S, \Lambda
S)=\mqty(e_* & 0 & 0 \\ 0 & d_* \Delta & a_* \Delta \\ 0 & a_* \Delta & c_* \Delta
).
\end{equation}
Combining this with \eqref{eq:xparlim}, we obtain for the first term of \eqref{eq:fnapp} that
\begin{equation}\label{eq:sDDs}
    \frac{1}{n} \x^\top_\| \DD \x_\| \doteq \frac{e_{*}r^{2}(\x)}{b_{*}}
+\Tr \begin{pmatrix} d_* & a_* \\ a_* & c_* \end{pmatrix}
\left(v(\x),\frac{w(\x)}{\sqrt{\kappa_*}}\right)^\top
\left(v(\x),\frac{w(\x)}{\sqrt{\kappa_*}}\right).
\end{equation}
Similarly, for the second term of \eqref{eq:fnapp}, applying the form of $\Es_b$
in \Cref{prop:AMPparamconverge} and (\ref{eq:dB}), we have
\begin{equation}\label{eq:sDx}
    \frac{2}{n} \x^\top_\| D\xi
\doteq \frac{2}{n}e_b^\top D\xi\left(b_*^{-1/2}r(\x)\right)
\doteq \frac{2r(\x)}{\sqrt{b_*}}
\end{equation}
where the contributions from the other terms of $\x_{\|}$ vanish because
$\Xi$ in \Cref{prop:AMPparamconverge} has mean 0 and is
independent of $(\Ss_1,\ldots,\Ss_t)$ and $(\D,\Ls)$.

For the final term $E_n(\x)$ of \eqref{eq:fnapp}, note that (\ref{eq:Assump2D}) and Cauchy
eigenvalue interlacing imply that the empirical eigenvalue distribution of
$\Pi^\top D^\top D\Pi$ converges weakly to $\D^2$. Then, recalling
$n^{-1}\|\x_{\perp}\|^{2}=u(\x)-r^2(\x)-\|v(\x)\|^{2}-\|w(\x)\|^{2}=\Ac(\Pfrak(\x))$
from (\ref{eq:lengthre})
and the definition of $\mathcal{H}$ in (\ref{eq:defFH}), 
\begin{equation}\label{eq:scH}
    \begin{aligned}
    &\frac{\tmr\left\|\x_{\perp}\right\|^{2}}{n}-\frac{1}{n} \log
\operatorname{det}\left(\tmr I+\Pi^{\top} \DD \Pi\right)-\Big(1+\log
\frac{\left\|\x_{\perp}\right\|^{2}}{n}\Big)
    \doteq\mathcal{H}\left(\tmr,\Ac(\Pfrak(\x))\right).
    \end{aligned}
\end{equation}
The error of this approximation converges to 0 uniformly over $\zeta \geq
-d_-+\df$, by the same Arzel\`a-Ascoli argument as leading to
\cite[Eq.\ (3.25)]{fan2021replica}.

To analyze the remaining
second term of $E_n(\x)$ in \eqref{eq:En}, let us define
\begin{equation}\label{eq:defW}
\begin{aligned}
\bar{\Pi}&=\left(e_{b}, S, \Lambda S\right)\mqty(
e_{b}^{\top} e_{b} & e_{b}^{\top} S & e_{b}^{\top} \Lambda S \\
S^{\top} e_{b} & S^{\top} S & S^{\top} \Lambda S \\
S^{\top} \Lambda e_{b} & S^{\top} \Lambda S & S^{\top} \Lambda^{2} S
)^{-1 / 2}
=\left(e_{b}, S, \Lambda S\right)\mqty(
e^{\top} e & e^{\top} X & e^{\top} Y \\
X^{\top} e & X^{\top} X & X^{\top} Y \\
Y^{\top} e & Y^{\top} X & Y^{\top} Y)^{-1 / 2}
\end{aligned}
\end{equation}
whose columns are an orthogonalization of $\left(e_{b}, S, \Lambda S\right)$.
Then the columns of $(\Pi, \bar{\Pi})$ form a full orthonormal basis for
$\mathbb{R}^{n}$. Applying the Schur-complement formula for block
matrix inversion, we obtain analogously to \cite[Eq.\ (3.29)]{fan2021replica}
that the second term of \eqref{eq:En} is given by
\[n^{-1}(\xi^{\top} D \Pi-\x_{\|}^{\top} \DD \Pi)(\tmr I+\Pi^{\top} \DD
\Pi)^{-1}(\Pi^{\top} D^\top \xi-\Pi^{\top} \DD \x_{\|})
=\mathrm{I}-\mathrm{II}\]
where
\begin{align}
\mathrm{I}&=n^{-1}(\xi^{\top} D-\x_{\|}^{\top} \DD)\Pi \cdot \Pi^\top
\left(\tmr I+\DD\right)^{-1} \Pi \cdot \Pi^\top \left(D^\top \xi-\DD
\x_{\|}\right)\label{eq:sectermPI1}\\
\mathrm{II}&=n^{-1}(\xi^{\top} D-\x_{\|}^{\top} \DD )\Pi \cdot \Pi^\top
(\tmr I+\DD)^{-1} \bar{\Pi}\left(\bar{\Pi}^{\top}\left(\tmr I+\DD\right)^{-1}
\bar{\Pi}\right)^{-1}\nonumber\\
&\qquad\qquad\qquad \bar{\Pi}^{\top}(\tmr I+\DD)^{-1} \Pi \cdot
\Pi^\top (D^\top \xi-\DD \x_{\|})\label{eq:sectermPI2}
\end{align}

We derive almost-sure asymptotic limits for $\mathrm{I}$ and $\mathrm{II}$.
Recalling $\lambda(\cdot)$ and $\theta(\cdot)$ from (\ref{eq:defFH}),
let us define
\begin{equation}\label{eq:defFFF}
\begin{aligned}
\mathcal{F}^{e}(\tmr, r)&=\mathcal{F}_{22}^{e}(\tmr, r)-\mathcal{F}_{12}^{e}(\tmr, r) \mathcal{F}_{11}^{e}(\tmr)^{-1} \mathcal{F}_{12}^{e}(\tmr, r)\\
\mathcal{F}(\tmr)&=\mathcal{F}_{22}(\tmr)-\mathcal{F}_{12}(\tmr)^{\top} \mathcal{F}_{11}(\tmr)^{-1} \mathcal{F}_{12}(\tmr)
\end{aligned}
\end{equation}
where we set
\begin{align}
\Ae(x)&=b_*^{-1/2}\qty(-x+\frac{e_*}{b_*}), \quad
\Be(x)=\left(\frac{\gamma_{*}}{\eta_{*}}\right)
x+\frac{\gamma_{*}\left(\eta_{*}-\gamma_{*}\right)}{\eta_{*}}-\frac{1}{b_*},
\quad \fe(x,r)=\Ae(x) \cdot r+\Be(x),\nonumber\\
\mathcal{F}_{11}^{e}(\tmr)&=\E\left(\frac{1}{\tmr+\D^{2}}
\mathsf{E}_{b}^{2}\right), \quad
\mathcal{F}_{12}^{e}(\tmr,
r)=\mathbb{E}\left(\frac{\fe\left(\D^{2},r\right)}{\tmr+\D^{2}}
\mathsf{E}_{b}^{2}\right), \quad
\mathcal{F}_{22}^{e}(\tmr,
r)=\mathbb{E}\left(\frac{\fe\left(\D^{2},r\right)^2}{\tmr+\D^{2}}
\mathsf{E}_{b}^{2}\right),\label{eq:defFFeterms}\\
\mathcal{F}_{11}(\tmr)&= \E \frac{1}{\tmr+\D^2}\begin{pmatrix}
1 & \lambda(\D^2) \\
\lambda(\D^2) & \lambda(\D^2)^{2}
\end{pmatrix}, \quad
\mathcal{F}_{12}(\tmr)=\E \frac{1}{\tmr+\D^2}\begin{pmatrix}
\theta(\D^2) \\ \lambda(\D^2) \theta(\D^2) \end{pmatrix}, \quad
\mathcal{F}_{22}(\tmr)=\E \frac{1}{\tmr+\D^2} \theta(\D^2)^{2}.\nonumber
\end{align}
Since $\D^2$ has strictly positive variance and
$x \mapsto \lambda(x)$ is one-to-one on the support of $\D^2$, we have
$\Fc_{11}(\tmr) \succ 0$ strictly for $\tmr>-d_-$, so $\Fc(\tmr)$ is
well-defined. Also $\Fc_{11}^e(\tmr)>0$ strictly since $\E \mathsf{E}_b^2=b_*>0$
from \Cref{remark:non0param}, so $\Fc^e(\tmr,r)$ is well-defined.
We note that these functions are expressed equivalently as
\begin{equation}\label{eq:altFFe}
    \begin{aligned}
\mathcal{F}^{e}(\tmr, r)&=
\inf_{\Fpe\in \R}
\mathbb{E}\left[\frac{\mathsf{E}_{\mathrm{b}}^{2}}{\tmr+\D^2}\left(\fe(\D^2,r)-\Fpe\right)^{2}\right]
=\inf_{\Fpe\in \R} \mathbb{E}\left[\frac{\mathsf{E}_{\mathrm{b}}^{2}}{\tmr+\D^2}\left(\left[\frac{\gamma_{*}}{\eta_{*}}-b_{*}^{-1/2} r\right] \D^2-\Fpe\right)^{2}\right]\\
\mathcal{F}(\tmr)&=\inf_{\Fpnu, \Fpom\in \R}
\mathbb{E}\left[\frac{1}{\tmr+\D^2}\left(\theta(\D^2)-\lambda(\D^2)
\Fpnu -\Fpom\right)^{2}\right]
\end{aligned}
\end{equation}
where these coincide with the above definitions upon explicitly evaluating the
infima over $\Fpe,\Fpnu,\Fpom \in \R$, and the two expressions for $\Fc^e$ in
(\ref{eq:altFFe}) are identical upon absorbing all terms of $\fe(\D^2,r)$ not
depending on $\D^2$ into an additive shift of the variable $\Fpe$.

We first approximate the common term $\Pi\Pi^\top(D^\top \xi-D^\top D\x_{\|})$
in (\ref{eq:sectermPI1}--\ref{eq:sectermPI2}):
Applying \Cref{cor:ampSEcor1} together with
\eqref{eq:xparlim} and \eqref{eq:SDSlim}, we obtain
\[\begin{aligned}
\mqty(
e^{\top} e & e^{\top} X & e^{\top} Y \\
X^{\top} e & X^{\top} X & X^{\top} Y \\
Y^{\top} e & Y^{\top} X & Y^{\top} Y
)^{-1}\left(e_{b}, S, \Lambda S\right)^{\top} \DD \x_{\|}
&\doteq \mqty( b_{*} & 0 & 0 \\ 0 & \Delta & 0 \\ 0 & 0 & \kappa_{*} \Delta
)^{-1}\mqty( e_* & 0 & 0 \\ 0 & d_* \Delta & a_{*} \Delta \\
0 & a_{*} \Delta & c_{*} \Delta)\mqty(
b_{*}^{-1/2} r(\x) \\ \Delta^{-1/2} v(\x) \\
\left(\kappa_{*} \Delta\right)^{-1/2} w(\x)).
\end{aligned}\]
Then applying
\begin{equation}\label{eq:bigPI}
    \Pi\Pi^\top=I-\bar{\Pi}\bar{\Pi}^{\top}=I-\left(e_{b}, S, \Lambda S\right)\mqty(
e^{\top} e & e^{\top} X & e^{\top} Y \\
X^{\top} e & X^{\top} X & X^{\top} Y \\
Y^{\top} e & Y^{\top} X & Y^{\top} Y
)^{-1}\left(e_{b}, S, \Lambda S\right)^{\top}
\end{equation}
yields
\begin{equation}\label{eq:imdfa1}
    \begin{aligned}
    \Pi\Pi^\top \DD \x_\|&=b_*^{-1/2} r(\x)\left(\DD-e_*b_*^{-1}I\right)
e_{b}+\left(\DD-d_*I-a_* \kappa_*^{-1} \Lambda\right) S \Delta^{-1/2} v(\x)\\
    &\qquad\qquad+\kappa_*^{-1/2}\left(\DD \Lambda-a_{*} I-c_* \kappa_*^{-1}
\Lambda\right) S \Delta^{-1/2} w(\x)+r_n(\x)
    \end{aligned}
\end{equation}
where $n^{-1}\|r_n(\x)\|^2 \to 0$ uniformly over $\x \in \Kc$.
From the definitions of $a_*,c_*$ in (\ref{eq:auxparams}) and of $\Lambda$ in
(\ref{eq:Lambdadef}), a straightforward computation yields the identity
\begin{equation}\label{eq:imdfa2}
-\frac{\gamma_*}{\eta_*-\gamma_*}\left(\DD \Lambda-a_{*}
I-\frac{c_{*}}{\kappa_{*}} \Lambda\right)=\DD-d_*I-\frac{a_{*}}{\kappa_{*}}
\Lambda=:\widetilde{D}
\end{equation}
where we define $\widetilde{D} \in \R^{n \times n}$
as this common quantity. Then, recalling
$\alpha_*^A$ from \eqref{eq:auxparams},
we can rewrite \eqref{eq:imdfa1} as
\begin{equation}\label{eq:PIDDx}
\Pi\Pi^\top \DD \x_\|=
b_*^{-1/2}r(\x)\cdot\left(\DD-e_*b_*^{-1} I\right) e_{b}
+\widetilde{D} S\Delta^{-1/2} \left(v(\x)-\alpha_{*}^{A}w(\x)\right)+r_n(\x).
\end{equation}
Applying \Cref{cor:ampSEcor1}, \eqref{eq:SDSlim}, \eqref{eq:bigPI} and $D^\top
\xi=\frac{\gamma_{*}}{\eta_{*}}\left[\DD+\left(\eta_{*}-\gamma_{*}\right)
I\right] e_{b}$, we also have similarly to (\ref{eq:sDx})
\begin{equation}\label{eq:PIDxi}
\begin{aligned}
\Pi\Pi^\top D^\top \xi
&=D^\top \xi-\left(e_{b}, S, \Lambda S\right)\mqty(
b_{*} & 0 & 0 \\
0 & \Delta & 0 \\
0 & 0 & \kappa_{*} \Delta
)^{-1} \frac{1}{n}\left(e_{b}, S, \Lambda S\right)^{\top} D^\top \xi+r_n(\x)\\
&=\frac{\gamma_{*}}{\eta_{*}}\left[\DD+\left(\eta_{*}-\gamma_{*}\right) I\right]
e_{b}-b_{*}^{-1} e_{b}+r_n(\x).
\end{aligned}
\end{equation}
Then combining \eqref{eq:PIDDx} and \eqref{eq:PIDxi}, and applying $\fe$
from \eqref{eq:defFFeterms} to $\DD$ by functional calculus,
\begin{equation}\label{eq:PIDfinal}
\Pi\Pi^\top(D^\top \xi -\DD \x_\|)
=\fe(\DD,r(\x))e_b-\widetilde{D}
S\Delta^{-1/2}\left( v(\x)-\alpha_{*}^{A} w(\x)\right)+r_n(\x)
\end{equation}
for a remainder $r_n(\x)$ satisfying $n^{-1}\|r_n(\x)\|^2 \to 0$ uniformly over
$\x \in \Kc$.

We now apply \eqref{eq:PIDfinal} and \Cref{cor:ampSEcor2} to approximate the 
two terms (\ref{eq:sectermPI1}) and (\ref{eq:sectermPI2}): Recalling
$\alpha_*^B$ from (\ref{eq:auxparams}),
observe that the second definition for $\widetilde{D}$ in
(\ref{eq:imdfa2}) has the equivalent form
\[\widetilde{D}=\DD-\alpha_{*}^{B}
\eta_{*}\left(\DD+\left(\eta_{*}-\gamma_{*}\right)
I\right)^{-1}+(\alpha_{*}^{B}-d_*)I.\]
Then, recalling the definitions of $\Fc_{22},\Fc_{22}^e$ from
(\ref{eq:defFFeterms}), by \Cref{thm:ampSE} and \Cref{cor:ampSEcor2} we have
\begin{align*}
&n^{-1}S^{\top} \widetilde{D}\left(\tmr I+\DD\right)^{-1} \widetilde{D} S
\doteq \Fc_{22}(\tmr) \cdot \Delta,\\
&n^{-1}e_{b}^{\top} \fe\left(\DD,r(\x)\right)\left(\tmr I+\DD\right)^{-1}
\fe \left(\DD,r(\x)\right) e_{b} \doteq \mathcal{F}_{22}^{e}(\tmr,r(\x)),\\
&n^{-1}e_b^\top \fe(\left(\DD,r(\x)\right)\left(\tmr I+\DD\right)^{-1}
\widetilde{D} S \doteq 0.
\end{align*}
Combining with \eqref{eq:PIDfinal}, this shows for \eqref{eq:sectermPI1} that
\begin{equation}\label{eq:sectermPI1first}
\mathrm{I}\doteq \mathcal{F}_{22}^{e}(\tmr, r(\x))+\mathcal{F}_{22}(\tmr)
\cdot \left\|v(\x)-\alpha_{*}^{A} w(\x)\right\|_{2}^{2}
=\mathcal{F}_{22}^{e}(\tmr, r(\x))+\mathcal{F}_{22}(\tmr)\cdot \Bc(v(\x),w(\x)).
\end{equation}
For \eqref{eq:sectermPI2}, by \Cref{thm:ampSE} and \Cref{cor:ampSEcor2}, we have
\begin{align*}
n^{-1}(S,\Lambda S)^{\top}(\tmr I+\DD)^{-1}(S, \Lambda S) &\doteq
\mathcal{F}_{11}(\tmr) \otimes \Delta \in \mathbb{R}^{2 t \times 2 t},\\
n^{-1}e_b^\top (\tmr I+\DD)^{-1}e_b &\doteq \Fc^e_{11}(\tmr),\\
n^{-1}e_b^\top (\tmr I+\DD)^{-1}(S, \Lambda S) &\doteq 0.
\end{align*}
Then, recalling the form of $\bar{\Pi}$ from \eqref{eq:defW},
\[\bar{\Pi}\left(\bar{\Pi}^{\top}\left(\tmr I+\DD \right)^{-1}
\bar{\Pi}\right)^{-1} \bar{\Pi}^{\top}
=n^{-1}\left(e_b\Fc^e_{11}(\tmr)^{-1}e_b^\top+
(S,\Lambda S) (\mathcal{F}_{11}(\tmr) \otimes \Delta)^{-1}
(S,\Lambda S)^\top+r_n\right)\]
where $n^{-1}\|r_n\| \to 0$ in operator norm.
Combining this with \eqref{eq:PIDfinal}, and applying Theorem \ref{thm:ampSE},
Corollary \ref{cor:ampSEcor2}, and the definitions of $\Fc_{12},\Fc_{12}^e$ in
(\ref{eq:defFFeterms}), we obtain for \eqref{eq:sectermPI2}
\begin{equation}\label{eq:unifsecond}
\mathrm{II}\doteq \mathcal{F}_{12}^{e}(\tmr,r(\x))
\mathcal{F}_{11}^{e}(\tmr)^{-1}
\mathcal{F}_{12}^{e}(\tmr,r(\x))+\mathcal{F}_{12}(\tmr)^{\top}
\mathcal{F}_{11}(\tmr)^{-1} \mathcal{F}_{12}(\tmr) \cdot \Bc(v(\x),w(\x)).
\end{equation}
Combining (\ref{eq:sectermPI1first}) and (\ref{eq:unifsecond}),
the second term of $E_n(\x)$ in \eqref{eq:En} satisfies
\begin{equation}\label{eq:sectermPIlim}
    \begin{aligned}
&n^{-1}(\xi^{\top} D \Pi-\x_{\|}^{\top} \DD \Pi)(\tmr I+\Pi^{\top} \DD
\Pi)^{-1}(\Pi^{\top} D^\top \xi-\Pi^{\top} \DD \x_{\|})\\
    &\doteq \Fc^e(\tmr,r(\x))+\Fc(\tmr) \cdot \Bc(v(\x),w(\x)).
    \end{aligned}
\end{equation}
The error of this approximation again converges to 0 uniformly over
$\tmr\ge -d_-+\df$, by an argument that is the same as
that leading to \cite[Eq.\ (3.31)]{fan2021replica}.

Combining \eqref{eq:sDDs}, \eqref{eq:sDx}, \eqref{eq:scH}, and
\eqref{eq:sectermPIlim} and applying this to \eqref{eq:fnapp}, we obtain
\[\lim_{n,m \to \infty} \sup_{\x \in \Kc}
\Big|f_n(\x)-f(\Pfrak(\x))\Big|=0,\]
where we define on the domain $\Cst=\{\Pfrak:u-r^2-\|v\|^2-\|w\|^2>0\}$
the function
\[f(\Pfrak)=\inf_\tmr {-}\frac{e_*r^2}{b_*}+\frac{2r}{\sqrt{b_*}}
- \Tr\left[\begin{pmatrix} d_* & a_* \\ a_* & c_* \end{pmatrix}
\left(v,\frac{w}{\sqrt{\kappa_*}}\right)^\top
\left(v,\frac{w}{\sqrt{\kappa_*}}\right)\right]
+\mathcal{H}\qty(\tmr,\Ac(\Pfrak))+\Fc^e \qty(\tmr,r)
+\Fc(\tmr) \cdot \Bc(v,w),\]
and the infimum is over $\tmr \ge -d_-+\df$.
The functions $\Fc^e,\Fc$ are decreasing over
$\tmr>-d_-$ by definition in (\ref{eq:altFFe}).
For any fixed $\Pfrak$ with $u<K$,
applying $\pdv{\tmr} \mathcal{H}(\tmr, \Ac(\Pfrak))
=\Ac(\Pfrak)-G(\tmr)<K-G(\tmr)<0$ by \eqref{eq:Gddfbig},
the function $\Hc(\tmr,\Ac(\Pfrak))$ is also decreasing over
$\tmr \in (-d_-,-d_-+\df]$.
Then the above infimum defining $f(\Pfrak)$ may be extended to $\zeta>-d_-$.
Finally, this uniform approximation for $f_n$ may be extended from $\Kc$ to its
closure $\overline{\Kc}$: Here $f_n$ as defined in (\ref{eq:fndef})
is continuous on $\overline{\Kc}$, and
the map $\Pfrak:\overline{\Kc} \to \{\Pfrak \in \bar{\Cst}:u \in [0,K]\}$
is continuous, relatively open, and maps the dense subset $\Kc \subset
\overline{\Kc}$ to the interior $\{\Pfrak \in \Cst:u \in (0,K)\}$ for each fixed
$n$. Then \cite[Proposition C.1]{fan2021replica} shows
that $f$ has a continuous extension to
$\{\Pfrak \in \bar{\Cst}: u \in [0,K]\}$, and
(denoting also by $f$ this extension)
\[\lim_{n,m \rightarrow \infty} \sup_{\x \in \overline{\Kc}}
\left|f_{n}(\x)-f(\Pfrak(\x))\right|=0.\]
Applying this to (\ref{eq:Econdstart}), and denoting by $\langle \cdot
\rangle_\pi$ the expectation over $(\s_i)_{i=1}^n \overset{iid}{\sim} \pi$,
we obtain almost surely
\begin{equation}\label{eq:finalfirstorderapprox}
\lim_{n,m \to \infty} \frac{1}{n}\log \E[\Z(\cU) \mid \cG]
=-\frac{1}{2}+\lim_{n,m \to \infty} \frac{1}{n}\log \left \langle
\mathbb{I}\big\{u(\x) \in \cU\big\}
\exp\left(\frac{n}{2}f(\Pfrak(\x))\right)\right \rangle_\pi.
\end{equation}
The same statement also holds with the closure $\overline{\cU}$ in place of $\cU$
on both sides.\\

{\bf Large deviations analysis.} We conclude the proof by establishing large
deviations upper and lower bounds for $\Pfrak(\x)$ and applying Varadhan's
lemma. Recall that $\x=\s-\st$, and introduce dual variables $\Rfrak=(U,R,V,W)$
where $U,R \in \R$ and $V,W\in \R^t$. For the large deviations upper bound,
define the cumulant generating function
\begin{equation}\label{eq:CGF}
\begin{aligned}
\lambda_n(\Rfrak)&=\frac{1}{n} \log \Big\langle
\exp\qty(n \cdot \Pfrak(\x)^\top \Rfrak) \Big \rangle_\pi \\
&=\frac{1}{n} \sum_{i=1}^n \log \int \exp\Big(U \sigma_i^2
+A_i\sigma_i+B_i\Big)d\pi(\sigma_i)
=\frac{1}{n}\sum_{i=1}^n \Big[B_i+\log c_\pi(U,A_i)\Big]
\end{aligned}
\end{equation}
where, denoting by $(e_i,x_i,y_i) \in \R^{2t+1}$ the $i^\text{th}$ row of
$(e,X,Y)$, we have set
\begin{align*}
A_i&=A_i(\Rfrak)=-2U\st_i+(R,V,W)^\top
\qty[\frac{1}{n} \mqty(e^{\top} e & e^{\top} X & e^{\top} Y \\ X^{\top} e &
X^{\top} X & X^{\top} Y\\ Y^\top e & Y^\top X & Y^\top Y)]^{-1/2}
(e_i,x_i,y_i),\\
B_i&=B_i(\Rfrak)=U(\st_i)^2-(R,V,W)^\top
\qty[\frac{1}{n} \mqty(e^{\top} e & e^{\top} X & e^{\top} Y \\ X^{\top} e &
X^{\top} X & X^{\top} Y\\ Y^\top e & Y^\top X & Y^\top Y)]^{-1/2}
(e_i,x_i,y_i) \cdot \st_i.
\end{align*}
By Theorem \ref{thm:ampSE}, \Cref{cor:ampSEcor1}, and Propositions
\ref{prop:contW} and \ref{prop:combW},
almost surely as $n,m \to \infty$, the empirical distributions
of $(A_i)_{i=1}^n$ and $(B_i)_{i=1}^n$ converge (in Wasserstein-$\pfrak$ for
every fixed order $\pfrak \geq 1$) respectively to
\begin{align*}
\As(\Rfrak)&=-2U\Xstar+b_*^{-1/2}R\Es+V^\top \Delta^{-1/2}(\Xs_1,\ldots,\Xs_t)
+\kappa_*^{-1/2} W^\top\Delta^{-1/2}(\Ys_1,\ldots,\Ys_t),\\
\Bs(\Rfrak)&=U{\Xstar}^2-b_*^{-1/2}R \Es\Xstar-V^\top \Delta^{-1/2}(\Xs_1,\ldots,\Xs_t)
\Xstar-\kappa_*^{-1/2} W^\top\Delta^{-1/2}(\Ys_1,\ldots,\Ys_t)\Xstar.
\end{align*}
In the remainder of the proof, we restrict to the event of probability 1 where
this empirical convergence holds. These limiting random variables satisfy
$\E{\Xstar}^2=\rho_*$, $\E \mathsf{E}\Xstar=0$, 
$\E\Ys_r\Xstar=0$, and $\E\Xs_r\Xstar=\E[F(\mathsf{P},\Xstar)\Xstar]=\pi_*$,
where $\mathsf{P}\sim N(0,\gamma_*^{-1})$ is independent of $\Xstar$, and
$\pi_*$ is defined in \eqref{eq:auxparams}. Define the limit
cumulant generating function
\begin{align*}
\lambda(\Rfrak)&=\E \log c_\pi (U,\As(\Rfrak))
+\Vpi\,U-\pi_*\,V^\top \Delta^{-1/2}1_{t\times 1} \in
(-\infty,\infty].
\end{align*}
Then Lemma \ref{lem:abar} and the above empirical Wasserstein
convergence ensure
\begin{equation}\label{eq:MGFconvergence}
\lim_{n,m \to \infty} \lambda_n(\Rfrak)=\lambda(\Rfrak)<\infty
\text{ for all } U<\bar{a}, \qquad
\lambda_n(\Rfrak)=\lambda(\Rfrak)=\infty \text{ for all } U>\bar{a}.
\end{equation}

Denote the Fenchel-Legendre transform of $\lambda$ by
\begin{equation}\label{eq:FenchelLegendre}
\lambda^{*}(\Pfrak)=\sup_{\Rfrak \in \R^{2t+2}}
\Pfrak^\top \Rfrak-\lambda(\Rfrak) \in [0,\infty].
\end{equation}
Observe that the concentration bound (\ref{eq:subGaussian}) implies
that $\int e^{ax^2}d\pi(x)<\infty$ for some sufficiently small $a>0$.
Then $\bar{a}>0$ strictly, so $\lambda$ is finite in an open neighborhood of
$\Rfrak=0$, and hence $\lambda^*$ is a good convex rate function (i.e.\ lower
semi-continuous and having compact level sets)
\cite[Lemma 2.3.9(a)]{Dembo1998large}. 
Let us show the large-deviations upper bound
\begin{equation}
\limsup_{n,m \to \infty}
\frac{1}{n}\log \Big\langle \mathbb{I}\{\Pfrak(\x) \in F\}
\Big\rangle_\pi \leq -\inf_{\Pfrak \in F} \lambda^*(\Pfrak)
\text{ for all closed } F \subseteq \bar{\Cst}.\label{eq:LDPupper}
\end{equation}
For this, set
$\bar{\lambda}(\Rfrak)=\limsup_{n,m \to \infty} \lambda_n(\Rfrak)$.
Here $\bar{\lambda}(\Rfrak)=\lambda(\Rfrak)$ whenever $U \neq \bar{a}$,
by (\ref{eq:MGFconvergence}). The upper bound in the G\"artner-Ellis Theorem
shows that (\ref{eq:LDPupper}) holds with $\lambda^*$ replaced by the
Fenchel-Legendre transform $\bar{\lambda}^*$ of $\bar{\lambda}$ (see
e.g.\ \cite[Exercise 2.3.25]{Dembo1998large}). Note that both $\lambda$ and
$\bar{\lambda}$ are convex, so the restrictions of $\lambda$ and
$\bar{\lambda}$ to any line segment are upper-semicontinuous by the
Gale-Klee-Rockafellar Theorem \cite[Theorem 10.2]{rockafellar2015convex}.
Then the supremum in (\ref{eq:FenchelLegendre}) and the analogous supremum
defining $\bar{\lambda}^*$ may be restricted to $\{\Rfrak:U \neq \bar{a}\}$,
implying that $\lambda^*=\bar{\lambda}^*$. This proves (\ref{eq:LDPupper}).

For the large deviations lower bound, set
\[\lambda^M(\Rfrak)=\E \log c_\pi^M(U,\As(\Rfrak))
+\Vpi\,U-\pi_*\,V^\top \Delta^{-1/2}1_{t \times 1},\]
and let $(\lambda^M)^*$ be its Fenchel-Legendre transform defined analogously to
(\ref{eq:FenchelLegendre}). We aim to show
\begin{equation}\label{eq:LDPlower}
\liminf_{n,m \to \infty}
\frac{1}{n}\log \Big\langle \mathbb{I}\{\Pfrak(\x) \in G\}
\Big\rangle_\pi \geq \sup_{M>0}\left(
{-}\inf_{\Pfrak \in G} (\lambda^M)^*(\Pfrak)\right)
\text{ for all open } G \subseteq \Cst.
\end{equation}
For this, consider any $M>0$ where $(-M,M)$ intersects the support of $\pi$,
and denote by $\pi_M$ the conditional distribution of $\pi$ over $(-M,M)$.
Let $\langle \cdot \rangle_{\pi_M}$ be the expectation over
$(\sigma_i)_{i=1}^n \overset{iid}{\sim} \pi_M$.
Then analogously to (\ref{eq:CGF}), we have
\begin{align*}
\lambda_n^M(\Rfrak):=\frac{1}{n} \log \Big\langle
\exp\qty(n \cdot \Pfrak(\x)^\top \Rfrak) \Big \rangle_{\pi_M}
=\frac{1}{n}\sum_{i=1}^n \Big[B_i+\log
c_\pi^M(U,A_i)\Big]-\log \pi((-M,M)).
\end{align*}
By the above empirical Wasserstein convergence,
$\lambda_n^M(\Rfrak)$ converges pointwise over all $\Rfrak \in \R^{2t+2}$ to
$\lambda^M(\Rfrak)-\log \pi((-M,M))$, which is now finite.
Then the G\"artner-Ellis lower bound \cite[Theorem 2.3.6]{Dembo1998large}
may be applied for the law of $\Pfrak(\x)$
under $\pi_M$, giving for all open $G \subseteq \Cst$
\begin{align*}
\frac{1}{n}\log \Big\langle \mathbb{I}\{\Pfrak(\x) \in G\}
\Big\rangle_\pi
&\geq \log \pi((-M,M))+\frac{1}{n}\log \Big\langle \mathbb{I}\{\Pfrak(\x)
\in G \Big\rangle_{\pi_M}\\
&\geq \log \pi((-M,M))-\inf_{\Pfrak \in G}
\big(\lambda^M-\log \pi((-M,M))\big)^*(\Pfrak)
=-\inf_{\Pfrak \in G} (\lambda^M)^*(\Pfrak).
\end{align*}
This lower bound is increasing in $M$, and taking the supremum over $M>0$
yields (\ref{eq:LDPlower}).

The function $\Pfrak \mapsto f(\Pfrak)$ is continuous and thus bounded over the
compact set $\{\Pfrak \in \bar{\Cst}:u \in [0,K]\}$. Then for any $M>0$,
applying (\ref{eq:LDPlower}) and Varadhan's lemma in the form of Lemma
\ref{lemma:varadhan}(a),
\[\liminf_{n,m \to \infty}
\frac{1}{n} \log \left\langle \mathbb{I}\Big\{
\Pfrak(\x) \in \Cst,\,u(\x) \in \cU\Big\}
\exp\Big(\frac{n}{2}f(\Pfrak(\x))\Big)\right\rangle_\pi
\geq \sup_{\Pfrak \in \Cst:\,u \in \cU}
\frac{1}{2}f(\Pfrak)-(\lambda^M)^*(\Pfrak).\]
Recalling the definitions of $\Phi_{1,t}^M$ and $\Psi_{1,t}^M$ in
(\ref{eq:defPhi1t}--\ref{eq:defPsi1t}), note that
$\Psi_{1,t}^M(\Pfrak)=-\frac{1}{2}+\frac{1}{2}f(\Pfrak)-(\lambda^M)^*(\Pfrak)$.
Then taking the supremum of the above over $M>0$ and applying this to
(\ref{eq:finalfirstorderapprox}) yields the desired lower bound
for $\E[\Z(\cU) \mid \cG_t]$. Similarly, recalling that $\lambda^*$ is a good
convex rate function and applying (\ref{eq:LDPupper}) and Varadhan's lemma in
the form of Lemma \ref{lemma:varadhan}(b),
\[\limsup_{n,m \to \infty}
\frac{1}{n} \log \left\langle \mathbb{I}\Big\{
\Pfrak(\x) \in \bar{\Cst},\,u(\x) \in \bar{\cU}\Big\}
\exp\Big(\frac{n}{2}f(\Pfrak(\x))\Big)\right\rangle_\pi
\leq \sup_{\Pfrak \in \bar{\Cst}:\,u \in \bar{\cU}}
\frac{1}{2}f(\Pfrak)-\lambda^*(\Pfrak).\]
Note that the condition $\Pfrak(\x) \in \bar{\Cst}$ on the left side holds
always, by definition of $\Pfrak$.
Since $\lambda^*$ is lower-semicontinuous, by the Gale-Klee-Rockafellar
Theorem, its restriction to any line segment is continuous. Since $f$ is also
continuous, the supremum on the right side may be restricted to the interior
$\{\Pfrak \in \Cst:u \in \cU\}$, and this yields the desired upper bound for
$\E[\Z(\overline{\cU}) \mid \cG_t]$.
\end{proof}

We now prove Lemma \ref{lemma:analysisf}. Let $e_t=(0,\ldots,0,1)$ be the
$t^\text{th}$ standard basis vector in $\R^t$.
Denote $\Pfrak_*=(u_*,r_*,v_*,w_*)$ and
$\Qfrak_*=(\zeta_*,U_*,R_*,V_*,W_*,\Fpe_*,\Fpnu_*,\Fpom_*)$, where
\begin{align}
 & u_{*}=\frac{2}{\eta_*}, \quad
r_*=\frac{\gamma_*}{\eta_*} b_*^{1/2},\quad
v_{*}=\frac{\eta_{*}-\gamma_{*}}{\eta_{*}} \Delta^{1/2}_t e_t, \quad
w_{*}=\frac{\gamma_{*}}{\eta_{*}} \kappa_{*}^{1/2} \Delta^{1/2}_t e_t
\label{eq:firstmomentspec}\\
&\tmr_{*}=\eta_*-\gamma_*, \quad \Fpe_*=\Fpnu_*=\Fpom_*=0, \quad
U_{*}=-\frac{1}{2}\gamma_*, \quad R_*=\gamma_* b_*^{1/2}, \quad V_{*}=0,
\quad W_{*}=\gamma_* \kappa_{*}^{1/2} \Delta_{t}^{1/2} e_t.\nonumber
\end{align}
We will show that $(\Pfrak_*,\Qfrak_*)$ is an approximate stationary point of
$\Phi_{1,t}$, which is an approximate global optimizer of
$\sup_\Pfrak \inf_\Qfrak \Phi_{1,t}$ for sufficiently small $\e>0$.

Denote by $\partial_u \Phi_{1,t} \in \R$, $\partial_v \Phi_{1,t} \in \R^t$
etc.\ the partial derivative or gradient of $\Phi_{1,t}$ in the variables $u$,
$v$, etc.

\begin{Lemma}\label{lemma:plugstat}
In the setting of Lemma \ref{lemma:analysisf},
    for all $t\ge 1$ and each $\iota\in \qty{u,r,v,w,\tmr,\Fpe,\Fpnu,\Fpom,U,R,W}$,
    \begin{equation}\label{eq:plugstat}
        \Phi_{1,t}(\Pfrak_*,\Qfrak_*)=\Psi_{\mathrm{RS}}, \quad \partial_\iota
\Phi_{1,t}(\Pfrak_*,\Qfrak_*)=0, \quad
        \lim_{t\to \infty }\norm{\partial_V \Phi_{1,t}(\Pfrak_*,\Qfrak_*)}=0.
    \end{equation}
\end{Lemma}
\begin{proof}
For the first term of $\Phi_{1,t}$, denote
\begin{align}
\As_*&=-2U_*\Xstar+b_{*}^{-1/2} R_* \mathsf{E}+V^{\top}_*
\Delta^{-1/2}_t\left(\Xs_{1}, \ldots, \Xs_{t}\right)+\kappa_{*}^{-1/2}
W^{\top}_* \Delta^{-1/2}_t\left(\Ys_{1}, \ldots, \Ys_{t}\right)\nonumber\\
&=\gamma_*(\Xstar+\Es+\Ys_t).
\label{eq:Astar}
\end{align}
As $\Es \sim N(0,b_*)$ and $\Ys_t \sim N(0,\sigma_*^2)$ are independent of each
other and of $\Xstar$, and $b_*+\sigma_*^2=\gamma_*^{-1}$, we have
\begin{equation}\label{eq:plug1}
\E\log c_{\pi}\left(U_*, \As_*\right)=\E\log
c_{\pi}\left(-\frac{1}{2}\gamma_*, \gamma_*\Xstar+\sqrt{\gamma_*}\Zs\right)
\end{equation}
    for $\Zs \sim N(0,1)$ independent of $\Xstar$. For the next terms of
$\Phi_{1,t}$, we have
    \begin{equation}\label{eq:plug2}
        \begin{aligned}
        &-(u_*-\rho_*)\cdot U_*-r_*\cdot R_*-(v_*+\pi_*\cdot \Delta^{-1/2}_t
1_{t \times 1})^\top V_*-w^\top_* W_*\\
        &\quad
=\frac{\gamma_*}{2}\left(\frac{2}{\eta_*}-\rho_*\right)-\frac{\gamma_{*}^2}{\eta_{*}}
b_{*}-0-\frac{\gamma_{*}^{2}}{\eta_{*}} \kappa_{*}\delta_*=-\frac{1}{2}\gamma_*\rho_*\\
        \end{aligned}
    \end{equation}
    where we used $e_t^\top \Delta_t e_t=\delta_{tt}=\delta_*$
from \Cref{thm:ampSE} and the identity
$b_*+\kappa_*\delta_*=b_*+\sigma_*^2=\gamma_*^{-1}$. For the
next terms of $\Phi_{1,t}$, applying $e_t^\top \Delta_t
e_t=\delta_*=(\eta_*-\gamma_*)^{-1}$ and the definitions of $a_*,c_*,e_*$ in
(\ref{eq:auxparams}), direct calculation gives, after some simplification,
    \begin{equation}\label{eq:plug3}
    \begin{aligned}
    &\frac{e_*}{b_{*}} r_*^2-2 b_{*}^{-1/2} r_*+
\Tr\left[\begin{pmatrix} d_* & a_* \\ a_* & c_* \end{pmatrix}
(v_*,\kappa_*^{-1/2}w_*)^\top (v_*,\kappa_*^{-1/2}w_*)\right] \\
    &\quad=e_*\qty(\frac{\gamma_*}{\eta_*})^2-2 \frac{\gamma_*}{\eta_*}+d_*\delta_*\qty(\frac{\eta_*-\gamma_*}{\eta_*})^2+c_*\delta_*\qty(\frac{\gamma_*}{\eta_*})^2+2 a_{*} \frac{\gamma_*}{\eta_*^2}
    =-\frac{\gamma_*}{\eta_*}.
    \end{aligned}
    \end{equation}
    Finally, applying again
$e_t^\top \Delta_t e_t=\delta_*=(\eta_*-\gamma_*)^{-1}$ and
$b_*+\kappa_*\delta_*=\gamma_*^{-1}$, we have
\begin{equation}\label{eq:APstar}
\Ac(\Pfrak_*)=u_*-r^2_*-\norm{v_*}^2-\norm{w_*}^2=\eta_*^{-1}.
\end{equation}
We have also $\tmr_*=\eta_*+R(\eta_*^{-1})=G^{-1}(\eta_*^{-1})$ from
(\ref{eq:fix}). Then by \cite[Proposition 2.9]{fan2021replica},
\begin{equation}\label{eq:plug4}
\Hc(\tmr_*,\Ac(\Pfrak_*))=\int_0^{\eta_*^{-1}} R(z)dz.
\end{equation}
The last two terms of $\Phi_{1,t}(\Pfrak_*,\Qfrak_*)$ are 0 because
$\gamma_{*}\eta_*^{-1}-b_{*}^{-1/2} r_*=0$ and
$\Bc(v_*,w_*)=\|v_{*}-\alpha_{*}^{A}w_{*}\|^2=0$. Then
combining (\ref{eq:plug1}), (\ref{eq:plug2}), (\ref{eq:plug3}), (\ref{eq:plug4})
shows $\Phi_{1,t}(\Pfrak_*,\Qfrak_*)=\Psi_{\mathrm{RS}}$ in (\ref{eq:plugstat}).
    
    To check the stationarity conditions, first by the form of $\mathcal{H}$ in
\eqref{eq:defFH}, we have $\partial_{\tmr} \mathcal{H}(\tmr, \Ac)=\Ac-{G}(\tmr)$
and $\partial_{\Ac} \mathcal{H}(\tmr, \Ac)=\tmr-1 / \Ac$. 
Recalling $\tmr_*=\eta_*-\gamma_*=G^{-1}(\eta_*^{-1})$ and
$\Ac(\Pfrak_*)=\eta_*^{-1}$, we have
    \begin{equation}\label{eq:Hcderivs}
        \partial_{\tmr} \mathcal{H}\left(\tmr_{*}, \Ac(\Pfrak_*)\right)=0, \quad
\partial_{\alpha} \mathcal{H}\left(\tmr_{*}, \Ac(\Pfrak_*)\right)=-\gamma_*.
    \end{equation}
    In addition, note that since $[\gamma_{*}\eta_{*}^{-1}-b_{*}^{-1/2} r_*] \D^2-\Fpe_*=0,$
    \begin{equation}\label{eq:Federivs}
    \partial_r,\partial_\tmr,\partial_{\Fpe} \mathbb{E}\left[\frac{\mathsf{E}_{b}^{2}}{\tmr+\D^2}\left(\left[\frac{\gamma_{*}}{\eta_{*}}-b_{*}^{-1/2} r\right] \D^2-\Fpe\right)^{2}\right] \Bigg |_{r=r_*,\Fpe=\Fpe_*,\tmr=\tmr_*}=0
    \end{equation}
Using this and $\Bc(v_*,w_*),\partial_v \Bc(v_*,w_*),\partial_w \Bc(v_*,w_*)=0$,
we obtain
        \begin{align*}
            &\partial_u \Phi_{1,t}(\Pfrak_*,\Qfrak_*)=-U_*-\gamma_*/2=0\\
            &\partial_r \Phi_{1,t}(\Pfrak_*,\Qfrak_*)=-R_*-(e_*b_*^{-1}r_*-b_*^{-1/2})+\gamma_*r_*
           \stackrel{(a)}{=}-r_*(\eta_*-\gamma_*)+b^{1/2}_*\frac{\gamma_*(\eta_*-\gamma_*)}{\eta_*}=0\\
            & \partial_v
\Phi_{1,t}(\Pfrak_*,\Qfrak_*)=-V_*-(d_*v_*+a_*\kappa_*^{-1/2}w_*)+\gamma_*v_*=\qty(-d_*-(1-d_*\gamma_*^{-1})\gamma_*+\gamma_*)v_*=0\\
            & \partial_w
\Phi_{1,t}(\Pfrak_*,\Qfrak_*)=-W_*-(c_*\kappa_*^{-1}w_*+a_*\kappa_*^{-1/2}v_*)+\gamma_*w_*=\qty(-\eta_*-\kappa_*^{-1}\qty(c_*+\frac{\eta_*-\gamma_*}{\gamma_*}a_*)+\gamma_*)w_*\stackrel{(b)}{=}0\\
            & \partial_\tmr,\partial_{\Fpe},\partial_{\Fpnu},\partial_{\Fpom} \Phi_{1,t}(\Pfrak_*,\Qfrak_*)=0
        \end{align*}
    where we used in (a) $R_*=\eta_* r_*$ and (by the definition of $b_*,e_*$ in
(\ref{eq:scalarparams}) and (\ref{eq:auxparams}))
\begin{equation}\label{eq:reidentity}
-e_*b_*^{-1}r_*+b_*^{-1/2}=b_*^{-1/2}
\left(1-\frac{\gamma_*}{\eta_*}(1+\kappa_*)\right)
=b_*^{1/2}\frac{\gamma_*(\eta_*-\gamma_*)}{\eta_*},
\end{equation}
and in (b) (by the definitions of $a_*,c_*$ in (\ref{eq:auxparams}))
$c_*+a_*(\eta_*-\gamma_*)/\gamma_*=-(\eta_*-\gamma_*)\kappa_*$.

Now note that
    \begin{equation}\label{eq:cpideriv}
        \partial_{a} \log c_{\pi}(a, b)=\frac{\int x^{2} \exp \left(a x^{2}+b
x\right) d\pi(x)}{\int \exp \left(a x^{2}+b x\right) d\pi(x)}, \quad
\partial_{b} \log c_{\pi}(a, b)=\frac{\int x \exp \left(a x^{2}+b x\right)
d\pi(x)}{\int \exp \left(a x^{2}+b x\right) d\pi(x)}.
    \end{equation}
    Recall (\ref{eq:Astar}), where $\Es+\Ys_t \sim N(0,\gamma_*^{-1})$ is
independent of $\Xstar$. Then from the expressions for $f,f'$ in \eqref{eq:deff} and \eqref{eq:denoiserpp1}, we have that 
    \begin{equation}\label{eq:logcpiabst}
        \partial_a \log c_{\pi}\left(U_*, \As_*\right)
        =\gamma^{-1}_*f'(\Xstar+\mathsf{E}+\Ys_t)+f(\Xstar+\mathsf{E}+\Ys_t)^2,
\quad \partial_b \log c_{\pi}\left(U_*, \As_*\right)=f(\Xstar+\mathsf{E}+\Ys_t).
    \end{equation}
    Also recall from \Cref{thm:ampSE} that
$\Xs_{t+1}=F(\Ys_{t}+\mathsf{E},\Xstar)$, so by the definition (\ref{eq:Fdef}),
    \begin{equation}\label{eq:fFid}
        f(\Xstar+\mathsf{E}+\Ys_t)=\frac{\eta_*-\gamma_*}{\eta_*} \Xs_{t+1}+\frac{\gamma_*}{\eta_*} (\Ys_t+\mathsf{E})+\Xstar.
    \end{equation}
Then, denoting
\[\Delta_{t+1}=\begin{pmatrix} \Delta_t & \delta_t \\ \delta_t^\top & \delta_*
\end{pmatrix}\]
and applying Corollary \ref{cor:ampSEcor1} and
$\E \Xs_s\Xstar=\pi_*$ from (\ref{eq:auxparams}), we have
\[\E(\Xs_1,\ldots,\Xs_t)f(\Xstar+\mathsf{E}+\Ys_t)
=\frac{\eta_*-\gamma_*}{\eta_*}\delta_t+\pi_* 1_{t \times 1}.\]
It follows from this, Stein's lemma, and $\E f'(\Xstar+\Es+\Ys_t)=
\gamma_* \E[\V[\Xstar \mid \Xstar+\Es+\Ys_t]]=\gamma_*/\eta_*$ that
\begin{align*}
\partial_R \Phi_{1,t}(\Pfrak_*,\Qfrak_*)&=b_*^{-1/2} \E
\mathsf{E}f(\Xstar+\mathsf{E}+\Ys_t)-r_*=b_*^{1/2} \E
f'(\Xstar+\mathsf{E}+\Ys_t)-r_*=b_*^{1/2}(\gamma_*/\eta_*)-r_*=0\\
\partial_V \Phi_{1,t}(\Pfrak_*,\Qfrak_*)&=\Delta_t^{-1/2}\E(\Xs_1,\ldots,
\Xs_t)f(\Xstar+\mathsf{E}+\Ys_t)-(v_*+\pi_*\cdot
\Delta_t^{-1/2}1_{t \times 1})=\frac{\eta_{*}-\gamma_{*}}{\eta_{*}}
\Delta_t^{-1/2}\delta_{t}-v_*\\
\partial_W \Phi_{1,t}(\Pfrak_*,\Qfrak_*)&=\kappa_*^{-1/2}\Delta_t^{-1/2} \E (\Ys_1,\ldots, \Ys_t)f(\Xstar+\mathsf{E}+\Ys_t)-w_*
=\kappa_*^{1/2}\Delta_t^{1/2}e_t \cdot \E f'(\Xstar+\Es+\Ys_t)-w_*=0.
\end{align*}
From (\ref{eq:fFid}) and the identities $\E \Xs_{t+1}\Xstar =\pi_*$ and
$\E\Xs_{t+1}^2=\delta_*=(\eta_*-\gamma_*)^{-1}$, we have also
\begin{align}
\E\Xstar f(\Xstar+\mathsf{E}+\Ys_t)
&=\frac{\eta_*-\gamma_*}{\eta_*}\pi_*+\Vpi\label{eq:fsqid}\\
\E f(\Xstar+\mathsf{E}+\Ys_t)^2&=
\qty(\frac{\eta_*-\gamma_*}{\eta_*})^2\frac{1}{\eta_*-\gamma_*}+\qty(\frac{\gamma_*}{\eta_*})^2\frac{1}{\gamma_*}+\rho_*+\frac{2(\eta_*-\gamma_*)}{\eta_*}\pi_*\nonumber\\
&=\frac{1}{\eta_*}+\rho_*+\frac{2(\eta_*-\gamma_*)}{\eta_*}\pi_*.
\label{eq:fxid}
\end{align}
Then
\begin{align*}
\partial_U \Phi_{1,t}(\Pfrak_*,\Qfrak_*)
&=\E\Big[\gamma^{-1}_*f'(\Xstar+\mathsf{E}+\Ys_t)+f(\Xstar+\mathsf{E}+\Ys_t)^2\Big]-2\E\Xstar f(\Xstar+\mathsf{E}+\Ys_t) -(u_*-\rho_*)=0.
\end{align*}
This shows $\partial_\iota \Phi_{1,t}(\Pfrak_*,\Qfrak_*)=0$
in \eqref{eq:plugstat} for all $\iota \neq V$, and it remains to verify
the bound for $\partial_V \Phi_{1,t}$. For this, note that from the above and
$\delta_{tt}=\delta_{t+1,t+1}=\delta_*$,
    \begin{equation*}
    \begin{aligned}
    &\norm{\partial_V \Phi_{1,t}(\Pfrak_*,\Qfrak_*)}^2=\norm{\frac{\eta_{*}-\gamma_{*}}{\eta_{*}}
\Delta_t^{-1/2} \delta_{t}-v_*}^2=\qty(\frac{\eta_{*}-\gamma_{*}}{\eta_{*}})^2
\norm{\Delta_t^{-1/2} \delta_t-\Delta_t^{1/2}e_t}^2\\
    & \qquad =\qty(\frac{\eta_{*}-\gamma_{*}}{\eta_{*}})^2\qty(\delta_t^\top \Delta_t^{-1}\delta_t +\delta_*-2\delta_t^\top e_t)=\qty(\frac{\eta_{*}-\gamma_{*}}{\eta_{*}})^2\qty(\delta_t^\top \Delta_t^{-1}\delta_t-\delta_{t+1,t+1} +2\delta_*-2\delta_{t,t+1}).
    \end{aligned}
    \end{equation*}
    By \Cref{prop:convsmallbeta}, $\lim _{t \rightarrow \infty} \delta_{t,
t+1}=\delta_{*}$. By the condition $\Delta_{t+1} \succ 0$ and the
Schur-complement formula,
    \begin{equation*}
        0<\delta_{t+1, t+1}-\delta_{t}^{\top} \Delta_{t}^{-1} \delta_{t}=\inf _{\alpha \in \mathbb{R}^{t}} \mathbb{E}\left[\left(\Xs_{t+1}-\alpha^{\top}\left(\Xs_{1}, \ldots, \Xs_{t}\right)\right)^{2}\right]\le \mathbb{E}\left[\left(\Xs_{t+1}-\Xs_{t}\right)^{2}\right]=2 \delta_{*}-2 \delta_{t, t+1}.
    \end{equation*}
    So, we also have that $\lim _{t \rightarrow \infty} \delta_{t+1,
t+1}-\delta_{t}^{\top} \Delta_{t}^{-1} \delta_{t}=0$. Thus $\lim _{t \rightarrow
\infty}\left\|\partial_{V} \Phi_{1, t}(\Pfrak_*,\Qfrak_*)\right\|=0$.
\end{proof}

\begin{proof}[Proof of Lemma \ref{lemma:analysisf}]

We begin with the lower bound in (\ref{eq:analysisf1}). Denote by $o_t(1)$
any scalar quantity that converges to 0 as $t \to \infty$.
%
Let us specialize $\Psi_{1,t}^M$ to
$\tilde{\Pfrak}_*=\left(u_*, r_{*}, \tilde{v}_{*}, w_{*}\right)$ where
\begin{equation}\label{eq:deffvtid}
    \begin{aligned}
\tilde{v}_{*}=\frac{\eta_{*}-\gamma_{*}}{\eta_{*}} \Delta_t^{-1 / 2} \delta_{t}
=v_*+\partial_V \Phi_{1,t}(\Pfrak_*,\Qfrak_*).
\end{aligned}
\end{equation}
Then Lemma \ref{lemma:plugstat} shows $\|\tilde{v}_*-v_*\|=o_t(1)$.
Observing that $\|w_*\|^2=(\gamma_*/\eta_*)^2 \kappa_*\delta_*$ is constant in
$t$, and applying (\ref{eq:plug3}), we then have
\begin{equation}\label{eq:Psi1tlowerspec}
\Psi_{1,t}^M(\tilde{\Pfrak}_*)=-\frac{1}{2}+\frac{\gamma_*}{2\eta_*}
+o_t(1)+\inf_{U,R,V,W}
X^M(U,R,V,W)+\frac{1}{2}\inf_{\tmr,\Fpe,\Fpnu,\Fpom} Y(\tmr,\Fpe,\Fpnu,\Fpom)
\end{equation}
where
\begin{equation*}
    \begin{aligned}
    X^M(U,R,V,W)&=\E\log c_{\pi}^M\left(U, -2U\Xstar+b_{*}^{-1/2} R
\mathsf{E}+V^{\top} \Delta_t^{-1/2}\left(\Xs_{1}, \ldots,
\Xs_{t}\right)+\kappa_{*}^{-1/2} W^{\top} \Delta_t^{-1/2}\left(\Ys_{1}, \ldots,
\Ys_{t}\right)\right)\\
    &\qquad -(u_*-\rho_*)U-r_*R-(\tilde{v}_*+\pi_* \Delta_t^{-1/2}1_{t \times
1})^\top V-w^\top_* W,\\
Y(\tmr,\Fpe,\Fpnu,\Fpom)&=\Hc(\tmr,\Ac(\tilde{\Pfrak}_*))
+\E\left[\frac{\Es_b^2}{\tmr+\D^2}(\chi^A)^2\right]
+\E\left[\frac{1}{\tmr+\D^2}\left(\theta\left(\D^2\right)-\lambda\left(\D^2\right)
\Fpnu -\Fpom\right)^{2}\right] \Bc(\tilde{v}_{*},w_*).
    \end{aligned}
\end{equation*}

Let $X(U,R,V,W)$ have the same definition as $X^M(U,R,V,W)$ with $c_\pi^M$
replaced by $c_\pi$. 
We note that by Lemma \ref{lemma:plugstat}, (\ref{eq:deffvtid}) is defined
exactly so that $(U_*,R_*,V_*,W_*)$ is a critical point of $X(\cdot)$. Since
$X(\cdot)$ is convex, $(U_*,R_*,V_*,W_*)$ is then a global
minimizer of $X$. We claim that
\begin{equation}\label{eq:XMlim}
\sup_{M>0} \inf_{U,R,V,W} X^M(U,R,V,W)=X(U_*,R_*,V_*,W_*).
\end{equation}
This is trivial if $\pi$ has bounded support, because
$X^M(U,R,V,W)$ is increasing in $M$ and equals $X(U,R,V,W)$ whenever $[-M,M]$
contains the support of $\pi$. To show (\ref{eq:XMlim}) when $\pi$ has unbounded
support, let us check the strict convexity of $X(\cdot)$:
Fix any unit vector $(U',R',V',W') \in \R^{2t+2}$ and $s>0$, and denote
\[(U(s),R(s),V(s),W(s))=(U_*,R_*,V_*,W_*)+s \cdot (U',R',V',W').\]
Set
\begin{equation}\label{eq:Fsdef}
\Fs=\begin{pmatrix}
b_*^{-1/2}\Es \\
\Delta_t^{-1/2} (\Xs_1,...,\Xs_t) \\
\kappa_*^{-1/2} \Delta_t^{-1/2} (\Ys_1,...,\Ys_t)\end{pmatrix} \in \R^{2t+1}.
\end{equation}
Recalling $U_*=-\gamma_*/2$ and $\As_*$ from (\ref{eq:Astar}), denote
\[\langle f(x) \rangle_*=\frac{\int f(x)\exp(-\frac{\gamma_*}{2}x^2
+\As_* x)d\pi(x)}{\int \exp(-\frac{\gamma_*}{2}x^2+\As_* x)d\pi(x)}\]
with the corresponding variance
$\V_*[f(x)]=\langle f(x)^2 \rangle_*-\langle f(x) \rangle_*^2$
(conditional on $\Xstar,\Es,\Xs_1,\ldots,\Xs_t,\Ys_1,\ldots,\Ys_t$). Then,
applying (\ref{eq:cpideriv}) and the chain rule,
\begin{equation}\label{eq:HessX}
\partial_s^2 X(U(s),R(s),V(s),W(s))\big|_{s=0}=
\E\left[\V_*\left[U'(x-\Xstar)^2+(R',V',W')^\top \Fs \cdot x\right]\right].
\end{equation}
Since $\pi$ has unbounded support, there are at least three distinct
points in its support, and hence also three distinct points
in the support of the posterior measure
defining $\langle \cdot \rangle_*$. Then the conditional variance
$\V_*\left[U'(x-\Xstar)^2+(R',V',W')^\top \Fs \cdot x\right]$ is 0 only if
the quadratic function $x \mapsto U'(x-\Xstar)^2+(R',V',W')^\top \Fs \cdot x$
takes constant value at these three points, which occurs only when both
$U'=0$ and $(R',V',W')^\top \Fs=0$.
When $U'=0$, we have $\|(R',V',W')\|=1$.
By Theorem \ref{thm:ampSE}, Corollary \ref{cor:ampSEcor1}, and
Proposition \ref{prop:denoiserprop},
$\Fs$ has zero mean and identity covariance, so
$(R',V',W')^\top \Fs$ has variance 1. Then
in particular, $(R',V',W')^\top \Fs \neq 0$ with positive probability.
Then $\V_*\left[U'(x-\Xstar)^2+(R',V',W')^\top \Fs \cdot x\right]>0$
with positive probability, and hence (\ref{eq:HessX}) is strictly positive.
This shows the strict convexity $\nabla^2 X(U_*,R_*,V_*,W_*) \succ 0$ as desired. Then by continuity, also $\nabla^2 X(U,R,V,W) \succ 0$ and $X(U,R,V,W)<\infty$
in a bounded neighborhood $\mathcal{O}$ of $(U_*,R_*,V_*,W_*)$. By the monotone 
convergence theorem, $\lim_{M \to \infty} X^M(U,R,V,W)=X(U,R,V,W)$, and this
convergence is uniform over $\mathcal{O}$ because $X^M$ and $X$ are convex
\cite[Theorem 10.8]{rockafellar2015convex}. Then the infimum of $X^M(U,R,V,W)$
is attained also in $\mathcal{O}$ for all large $M$, and
$\lim_{M \to \infty} \inf_{U,R,V,W} X^M(U,R,V,W)=X(U_*,R_*,V_*,W_*)$.
Hence (\ref{eq:XMlim}) holds, as claimed.

For the second term $Y(\tmr,\Fpe,\Fpnu,\Fpom)$, recall
$\Ac(\Pfrak_*)=\eta_*^{-1}$ from (\ref{eq:APstar}), and
$\partial_\tmr \Hc(\tmr_*,\eta_*^{-1})=0$. Then, since $\|v_*\|^2-\|\tilde
v_*\|^2=o_t(1)$, we have
$\Ac(\tilde{\Pfrak}_*)=\eta_*^{-1}+o_t(1)$ and
$\partial_\tmr \Hc(\tmr_*,\Ac(\tilde{\Pfrak}_*))=o_t(1)$.
Furthermore $\partial_\tmr^2 \mathcal{H}(\tmr,\Ac(\tilde\Pfrak_*))=-G'(\tmr)>0$
in a neighborhood of $\tmr_*$, so $\inf_{\tmr>d_-} \Hc(\tmr,\Ac(\tilde\Pfrak_*))
=\Hc(\tmr_*,\eta_*^{-1})+o_t(1)$
(see e.g.\ \cite[Proposition C.2]{fan2021replica}).
Applying that the infimum of $Y(\tmr,\Fpe,\Fpnu,\Fpom)$ occurs at $\Fpe=0$,
and that $\Bc(\tilde v_*,w_*)=\|v_*-\tilde v_*\|^2=o_t(1)$, we obtain
\begin{equation}\label{eq:Ylim}
\inf_{\tmr,\Fpe,\Fpnu,\Fpom} Y(\tmr,\Fpe,\Fpnu,\Fpom)
=\Hc(\tmr_*,\eta_*^{-1})+o_t(1).
\end{equation}
Applying (\ref{eq:XMlim}) and (\ref{eq:Ylim}) to (\ref{eq:Psi1tlowerspec}), we
obtain
\[\lim_{t \to \infty} \sup_{M>0} \Psi_{1,t}^M(\tilde \Pfrak_*)
=-\frac{1}{2}+\frac{\gamma_*}{2\eta_*}+X(U_*,R_*,V_*,W_*)
+\frac{1}{2}\Hc(\tmr_*,\eta_*^{-1})=\Psi_{\mathrm{RS}},\]
which implies the lower bound in (\ref{eq:analysisf1}).

For the upper bounds in (\ref{eq:analysisf1}) and (\ref{eq:analysisf2}),
we now fix $\hf \in \R_+$ and
specialize the dual variables $\Qfrak=\Qfrak(\Pfrak)$
as functions of $\Pfrak=(u,r,v,w)$, given by
\[U(u)=U_*+\hf (u-u_*), \quad R(r)=R_*+\hf (r-r_*),
\quad V(v)=V_*+\hf (v-v_*),\quad W(w)=W_*+\hf (w-w_*),\]
\[\tmr(\Pfrak)=G^{-1}(\Ac(\Pfrak)),\quad
    \Fpe(r)=-\left(\frac{\gamma_*(\eta_*-\gamma_*)}{\eta_*}-b_*^{-1}+rb_*^{-3/2}e_*\right), \quad
\Fpnu=\Fpnu_*=0, \quad \Fpom=\Fpom_*=0.\]
Here, for any $\Pfrak \in \Cst$ with $u\in (0,K)$, we have $\Ac(\Pfrak)<K$ so
$\tmr(\Pfrak)$ is well-defined for all $\e$ sufficiently small. 
Note that under the above choices of $\Fpe,\Fpnu,\Fpom$, we have
\[\Fc^e_{22}(\tmr,r)=\mathbb{E}\left[\frac{\mathsf{E}_{b}^{2}}{\tmr+\D^2}\left(\left[\frac{\gamma_{*}}{\eta_{*}}-b_{*}^{-1/2}
r\right] \D^2-\Fpe(r) \right)^{2}\right], \quad
    \Fc_{22}(\tmr)=\mathbb{E}\left[\frac{1}{\tmr+\D^2}\left(\theta\left(\D^2\right)-\lambda\left(\D^2\right)
\Fpnu_*-\Fpom_*\right)^{2}\right]\]
where $\Fc^e_{22},\Fc_{22}$ are the functions previously
defined in \eqref{eq:defFFeterms}. Then
\begin{equation}\label{eq:FSTTs1}
    \Psi_{1,t}(\Pfrak)
=\inf_{\Qfrak} \Phi_{1,t}(\Pfrak,\Qfrak) \leq
\Phi_{1,t}(\Pfrak,\Qfrak(\Pfrak))=:
\bar{\Psi}_{1,t}(\Pfrak)=-\frac{1}{2}+\Tone+\Ttwo+
\frac{1}{2}\Big(\Tthree+\Tfour+\Tfive+\Tsix\Big)
\end{equation}
where, recalling $\Fs$ from (\ref{eq:Fsdef}), we define
\begin{equation}\label{eq:FSTTs2}
    \begin{aligned}
    &\Tone=\E\log c_{\pi}\big(U(u), -2U(u)\Xstar
+(R(r),V(v),W(w))^\top \Fs\big)\\
    &\Ttwo=-(u-\rho_*)\cdot U(u)-r\cdot R(r)-(v+\pi_* \Delta_t^{-1/2}1_{t
\times 1})^\top V(v)-w^\top W(w)\\
    &\Tthree=-\frac{e_* r^2}{b_{*}}+\frac{2r}{\sqrt{b_*}}
-\Tr \begin{pmatrix} d_* & a_* \\ a_* & c_* \end{pmatrix}
\left(v,\frac{w}{\sqrt{\kappa_*}}\right)^\top
 \left(v,\frac{w}{\sqrt{\kappa_*}}\right)\\
    &\Tfour=\mathcal{H}\left(\tmr(\Pfrak),\Ac(\Pfrak)\right),
\quad \Tfive=\mathcal{F}^e_{22}(\tmr(\Pfrak),r), \quad
    \Tsix=\mathcal{F}_{22}(\tmr(\Pfrak))\cdot \Bc(v,w)
    \end{aligned}
\end{equation}

At $\Pfrak=\Pfrak_*$, using $\Ac(\Pfrak_*)=\eta_*^{-1}$,
$\tmr_*=G^{-1}(\eta_*^{-1})$, and (\ref{eq:reidentity}) to verify $\Fpe(r_*)=0$,
we observe that these specializations give $\Qfrak(\Pfrak_*)=\Qfrak_*$.
Then $\bar{\Psi}_{1,t}(\Pfrak_*)=\Phi_{1,t}(\Pfrak_*,\Qfrak_*)
=\Psi_{\mathrm{RS}}$ by \Cref{lemma:plugstat}. Furthermore,
noting that the only coordinates of $\Qfrak(\Pfrak)$ depending on $v$ are $V(v)$
and $\zeta(\Pfrak)$, the derivative of $\bar{\Psi}_{1,t}$ in $v$ is
\[\partial_{v} \bar{\Psi}_{1,t}(\Pfrak_*)=\partial_{v}
\Phi_{1,t}(\Pfrak_*,\Qfrak_*)+\partial_{\tmr} \Phi_{1,t}(\Pfrak_*,\Qfrak_*)
\cdot \partial_{v} \tmr(\Pfrak_*)+\partial_{V} \Phi_{1,t}(\Pfrak_*,\Qfrak_*)
\cdot \partial_{v} V(v_{*}).\]
The first term has norm $o_t(1)$, and the remaining two terms are 0, by
\Cref{lemma:plugstat}. Similarly
$\partial_u \bar{\Psi}_{1,t}(\Pfrak_*)=0$,
$\partial_r \bar{\Psi}_{1,t}(\Pfrak_*)=0$, and
$\partial_w \bar{\Psi}_{1,t}(\Pfrak_*)=0$, so $\|\nabla
\bar{\Psi}_{1,t}(\Pfrak_*)\|=o_t(1)$. 

We now show, using the small-$\e$ approximations of \Cref{prop:betalimit},
that the upper bound $\bar{\Psi}_{1,t}(\Pfrak)$
in (\ref{eq:FSTTs1}) is concave in $\Pfrak$ over the domain
$\{\Pfrak \in \Cst:u \in (0,K)\}$. Let us write $O(1)$, $O(\e)$ etc.\ for
scalar quantities bounded in magnitude
by $C$, $C\e$, etc.\ where the constant $C>0$ depends only on
$\pib,K$ (and not on $d_*,\e,\hf$ or the dimension $t$).
Fix any $\Pfrak=(u,r,v,w) \in \Cst$ with $u \in (0,K)$,
fix any unit vector $(u',r',v',w') \in \R^{2t+2}$, and define for $s>0$
\begin{equation}\label{eq:Pfraks}
\Pfrak(s)=(u(s),r(s),v(s),w(s))=(u,r,v,w)+s\cdot (u',r',v',w').
\end{equation}
We compute the second derivative of $\bar{\Psi}_{1,t}(\Pfrak(s))$ at $s=0$.
For the first term $\Tone$, denote
\[\expval{f(x)}_\Pfrak=\frac{\int f(x)\exp(U(u) \cdot x^2-2U(u)\Xstar \cdot x
+(R(r),V(v),W(w))^\top \Fs\cdot x)d\pi(x)}{\int
\exp(U(u) \cdot x^2-2U(u)\Xstar \cdot x
+(R(r),V(v),W(w))^\top \Fs\cdot x)d\pi(x)}\]
and let $\V_\Pfrak[f(x)]=\expval{f(x)^2}_\Pfrak-\expval{f(x)}_\Pfrak^2$
be the corresponding variance. Then
\begin{align}
\partial_s^2 \Tone\big|_{s=0}&=\hf^2\,\E\left[\V_\Pfrak\left[
u'x^2-2u'\Xstar \cdot x+(r',v',w')^\top \Fs \cdot x\right]\right]\nonumber\\
&\leq 2\hf^2\,\E\left[(u')^2 \cdot \V_\Pfrak[x^2]
+\left({-}2u'\Xstar+(r',v',w')^\top \Fs\right)^2 \V_\Pfrak[x] \right].
\label{eq:Isecondderiv}
\end{align}
Let us apply Assumption \ref{AssumpPrior} and
Proposition \ref{prop:bayesbnd} in dimension $k=1$, with
$\Gamma=\gamma_{\max}=\gamma_{\min}=U(u)$
and $z=-2U(u)\Xstar+(R(r),V(v),W(w))^\top \Fs$. We observe
that, since $u,u_* \in (0,K)$, we have
\begin{equation}\label{eq:Uuapprox}
U(u)=U_*+\hf(u-u_*)=-(\gamma_*/2)+\hf(u-u_*)=-(d_*/2)(1+O(\e))+O(\hf),
\end{equation}
the last equality applying
Proposition \ref{prop:betalimit}. In particular, since $d_*>0$,
for all $\hf \in (0,\hf_0)$ where $\hf_0$ is a small constant depending only
on $(K,\pib)$, we have $\Gamma<(4\pib)^{-1}$.
Then the condition (\ref{eq:poincare}) from Assumption
\ref{AssumpPrior} implies $\V_\Pfrak[x]=O(1)$.
Since also $r^2+\|v\|^2+\|w\|^2<u<K$ and $r_*^2+\|v_*\|^2+\|w_*\|^2<u_*<K$ and
$b_*+\kappa_*\delta_*=\gamma_*^{-1}$ and $\gamma_*=d_*(1+O(\e))$, we have
\[\|(R(r),V(v),W(w))\| \leq \|(R_*,V_*,W_*)\|+O(\hf)
=(\gamma_*^2b_*+\gamma_*^2\kappa_*\delta_*)^{1/2}+O(\hf)=d_*^{1/2}(1+O(\e))+O(\hf).\]
Then, for all $\hf \in (0,\hf_0)$, we have
$\|z\|^2 \leq ({\Xstar}^2+(q^\top \Fs)^2) \cdot O(1+d_*)$
for some unit vector $q \in \R^{2t+1}$. From (\ref{eq:Uuapprox}), we have
$\pib^{-1}-\gamma_{\max} \geq c(1+d_*)$ and $-\gamma_{\min}<C(1+d_*)$ for
constants $C,c>0$ depending only on $(K,\pib)$. So
Proposition \ref{prop:bayesbnd} shows
$\V_\Pfrak[x^2] \leq C(1+{\Xstar}^2+(q^\top \Fs)^2)$. Applying these to
(\ref{eq:Isecondderiv}), and applying also
$\E[{\Xstar}^2] \leq \pib$ and $\E[(q^\top \Fs)^2]=1$ for any unit vector $q$,
we obtain
\begin{equation}\label{Hess1}
\partial_s^2 \Tone\big|_{s=0}=O(\hf^2).
\end{equation}

For $\Ttwo$, we have
\begin{equation}\label{Hess2}
\partial_s^2 \Ttwo \mid_{s=0}=-2\hf({u'}^2+{r'}^2+\|v'\|^2+\|w'\|^2)
=-2\hf.
\end{equation}
For $\Tthree$, applying \Cref{prop:betalimit},
\begin{equation}\label{Hess3}
\partial_s^2 \Tthree \big|_{s=0} =-2\Tr \begin{pmatrix}
d_* & a_* \kappa_*^{-1/2} \\ a_* \kappa_*^{-1/2} & c_* \kappa_*^{-1} \end{pmatrix}
(v',w')^\top(v',w')-\frac{2e_*}{b_*} {r'}^2
=-2d_*\|(r',v',w')\|^2+O(\e).
\end{equation}
For $\Tfour$, we have $\Tfour=\int_{0}^{\Ac(\Pfrak(s))} R(z) dz$ by
\cite[Proposition 2.9(a)]{fan2021replica}. It is easily checked that
$|\Ac(\Pfrak(s))|=O(1)$, $|\partial_s \Ac(\Pfrak(s))|=O(1)$, and
$\partial_s^2 \Ac(\Pfrak(s))=-2\|(r',v',w')\|^2$ at $s=0$.
Then by \Cref{prop:betalimit},
\begin{equation}\label{Hess4}
\partial^2_s \Tfour \big|_{s=0}
=R'(\Ac(\Pfrak(s))) \cdot \partial_s \Ac(\Pfrak(s))^2
+R(\Ac(\Pfrak(s))) \cdot \partial_s^2 \Ac(\Pfrak(s))^2\Big|_{s=0}
=2d_*\|(r',v',w')\|^2+O(\e^2).
\end{equation}

For $\Tfive$, we may apply the series expansion for $R(z)$ from
\eqref{eq:Rseries} with $\kappa_1=-\E[\D^2]=-d_*$, to write for any
$z \in (0,K)$, $x \in [d_*-\e,d_*+\e]$, and sufficiently small $\e$,
    \begin{align}
    (G^{-1}(z)+x)^{-1}&=\qty(R(z)+z^{-1}+x)^{-1}
    =z\qty(1+(x-d_*)z+\sum_{k\ge 2} \kappa_k z^k)^{-1} \nonumber\\
    &=z \cdot \sum_{j \geq 0}\left(-\left(x-d_*\right) z-\sum_{k\ge 2}
\kappa_{k} z^{k}\right)^{j}=:z+\sum_{k \geq 1} c_{k}(x) z^{k+1}.
\label{eq:Vseries}
    \end{align}
Here, $|c_k(x)| \leq (O(\e))^k$,
and these series are absolutely convergent for sufficiently small $\e$.
Then the derivatives in $z$ may be computed term-by-term. Recalling
$\tmr(\Pfrak(s))=G^{-1}(\Ac(\Pfrak(s)))$ where $\Ac(\Pfrak(0)) \in (0,K)$, we
obtain by the chain rule
\begin{equation}\label{eq:dtmrbound}
\sup_{x \in
\operatorname{supp}(\D^2)}\left|\partial_{s}^{k}(\tmr(\Pfrak(s))+x)^{-1}\right|
\Big|_{s=0}=O(1) \text{ for } k=0,1,2.
\end{equation}
Recalling $\Ae(x)$ from \eqref{eq:defFFeterms} and applying
\Cref{prop:betalimit}, we have for any $x \in [d_*-\e,d_*+\e]$,
\[\Ae(x)=b_{*}^{-1/2}\left(-x+\frac{e_{*}}{b_{*}}\right)=O(b_*^{-1/2} \e).\]
By \Cref{prop:betalimit}, we have $\gamma_*=d_*+O(\kappa_2/\eta_*)$
and also $b_*^{-1}=d_*(1+\eta_*^{-1} \cdot
O(\kappa_2/d_*))=d_*+O(\kappa_2/\eta_*)$. Then
recalling $\Be(x)$ from \eqref{eq:defFFeterms}, we have
\[\Be(x)=\left(\frac{\gamma_{*}}{\eta_{*}}\right)x+\left(\frac{\gamma_{*}\left(\eta_{*}-\gamma_{*}\right)}{\eta_{*}}-\frac{1}{b_{*}}\right)
=\frac{\gamma_*}{\eta_*}(x-\gamma_*)+O\left(\frac{\kappa_2}{\eta_*}\right)
=O\left(\frac{\gamma_*\e}{\eta_*}+\frac{\kappa_2}{\eta_*}\right).\]
Here, $\kappa_2=O(d_*\e)=O(\gamma_*\e)$, the second equality holding because
$\gamma_*=d_*(1+O(\e))$. Then, applying also $\eta_*^{-1}=O(1)$ and
$0<\gamma_*<\eta_*$ by Proposition \ref{prop:uniquefix}, this gives
\[\Be(x)=O\left(\frac{\gamma_* \e}{\eta_*}\right)
=O\left(\sqrt{\frac{\gamma_*}{\eta_*}} \cdot
\sqrt{\frac{1}{\eta_*}} \cdot \sqrt{\gamma_*} \cdot e\right)
=O(\sqrt{d_*}\cdot \e)=O(b_*^{-1/2} \cdot \e).\]
Now applying these bounds for $\Ae,\Be$ to $\fe(x,r)$ from
\eqref{eq:defFFeterms} and differentiating by the chain rule,
\begin{equation}\label{eq:federivbound}
\sup_{x \in \operatorname{supp}(\D^2)}
\left|\partial_{s}^k\fe(x,r(s))\right|\Big|_{s=0}=O(b_*^{-1/2} \cdot \e)
\text{ for } k=0,1,2.
\end{equation}
Combining (\ref{eq:dtmrbound}) and (\ref{eq:federivbound})
and differentiating $\Fc_{22}^e$ from
\eqref{eq:defFFeterms} by the chain rule,
\begin{equation}\label{Hess5}
\partial_s^2 \Tfive\big|_{s=0}
=\partial_s^2 \Fc^e_{22} (\tmr(\Pfrak(s)),r(s)) \Big|_{s=0}
\leq O(b_*^{-1} \cdot \e^2) \cdot \E[\Es_b^2]=O(b_*^{-1} \cdot \e^2) \cdot b_*=O(\e^2).
\end{equation}

For $\Tsix$, for any $x \in [d_*-\e,d_*+\e]$, we may write $\theta(x)$ from
(\ref{eq:defFH}) as
\begin{align*}
\theta(x)&=x-d_*-\alpha_{*}^{B}\left(\frac{1}{1-\eta_*^{-1}(\gamma_{*}-x)}-1\right)=x-d_*-\alpha_{*}^{B}\sum_{k \geq 1}\left(\frac{\gamma_{*}-x}{\eta_{*}}\right)^{k} \\
&=\left(x-\gamma_{*}\right)\left(1+\eta_*^{-1}\alpha_{*}^{B}\right)+\gamma_{*}-d_*-\alpha_{*}^{B}\sum_{k
\geq 2}\left(\frac{\gamma_{*}-x}{\eta_{*}}\right)^{k}.
\end{align*}
Then, using $|x-\gamma_*|/\eta_*=O(\e)$, we have
\[\E\left[\theta(\D^2)^2\right]
\leq 3\Bigg((1+\eta_*^{-1}\alpha_*^B)^2\cdot
\E\left[(\D^2-\gamma_{*})^{2}\right]+(\gamma_*-d_*)^2+\left(\alpha_{*}^{B}\right)^{2}\E\left[\left(\frac{\gamma_{*}-\D^2}{\eta_{*}}\right)^4\right](1+O(\e))\Bigg)\]
By \Cref{prop:betalimit} and \Cref{lem:cauchy}(c), we have
$(\gamma_*-d_*)^2=O(\kappa_2^2\eta_*^{-2})=O(\kappa_2\e^2/\eta_*^2)$,
$\E[(\D^2-\gamma_*)^2]=\kappa_2+(d_*-\gamma_*)^2=O(\kappa_2)$,
$\E[(\D^2-\gamma_*)^4]=O(\mu_4+(d_*-\gamma_*)^4)=O(\kappa_2 \e^2)$,
$1+\eta_*^{-1}\alpha_*^B=O(\e/\eta_*)$, and
$\eta_*^{-2}(\alpha_*^B)^2=O(1)$. This gives
\begin{equation}\label{eq:thetasqbound}
\E[\theta(\D^2)^2]=O(\kappa_2\e^2/\eta_*^2)
\end{equation}
Then, applying again (\ref{eq:dtmrbound}) and
differentiating $\Fc_{22}$ from \eqref{eq:defFFeterms}, we
have at $s=0$ that
$\partial_s^k \Fc_{22}(\tmr(\Pfrak(s)))=O(\kappa_2\e^2/\eta_*^2)$ for $k=0,1,2$.
Using $\alpha_*^A=\frac{\eta_{*}}{\sqrt{\kappa_{2}}}(1+O(\eta_*^{-1}\e))$ and
$\norm{v},\norm{w},\norm{v'},\norm{w'}=O(1)$, we have also
$\partial_{s}^{k} \Bc(v(s),w(s))=O(\max(1,\eta_*^2/\kappa_2))$ for $k=0,1,2$.
Combining these bounds, we conclude that
\begin{equation}\label{Hess6}
\partial_s^2 \Tsix \Big|_{s=0}=O(\e^2).
\end{equation}
Now, combining \eqref{Hess1}, \eqref{Hess2}, \eqref{Hess3}, \eqref{Hess4}, \eqref{Hess5}, \eqref{Hess6}, and setting $\hf=\e^{1/2}$, we conclude that
\[\partial_s^2 \bar{\Psi}_{1,t}(\Pfrak(s))|_{s=0}=-2\e^{1/2}+O(\e)<-\e^{1/2}\]
for all $\e<\e_0(\pib,K)$. This holds for
$\Pfrak(s)$ as defined in (\ref{eq:Pfraks}) for any $(u,r,v,w) \in \Cst$
with $u \in (0,K)$ and for any unit vector $(u',r',v',w')$, implying the
concavity
\begin{equation}\label{eq:hessbound}
\nabla^2 \bar{\Psi}_{1,t}(\Pfrak) \prec -\e^{1/2}I \text{ over }
\{\Pfrak \in \Cst:u \in (0,K)\}.
\end{equation}

Finally, since $u_*=2\eta_*^{-1} \in \cU \subseteq (0,K)$
by assumption, we have that
$\Pfrak_*=(u_*,r_*,v_*,w_*)$ is an interior point of the open domain $\{\Pfrak
\in \Cst:u \in \cU\}$. Recalling $\bar{\Psi}_{1,t}(\Pfrak_*)=\Psi_{\mathrm{RS}}$
and $\|\nabla \bar{\Psi}_{1,t}(\Pfrak_*)\|=o_t(1)$, we then have
(see e.g.\ \cite[Proposition C.2]{fan2021replica})
\[\sup_{\Pfrak \in \Cst:u \in \cU} \Psi_{1,t}(\Pfrak)
\leq \sup_{\Pfrak \in \Cst:u \in \cU} \bar{\Psi}_{1,t}(\Pfrak)
=\bar{\Psi}_{1,t}(\Pfrak_*)+o_t(1)=\Psi_{\mathrm{RS}}+o_t(1).\]
This shows the upper bound of (\ref{eq:analysisf1}). Furthermore, by (\ref{eq:hessbound}) and a Taylor expansion, for any
$\Pfrak \in \Cst$ with $u \in (0,K)$ and $|u-u_*|>\varsigma$, we have
\[\bar{\Psi}_{1,t}(\Pfrak)
\leq \bar{\Psi}_{1,t}(\Pfrak_*)+
\nabla \bar{\Psi}_{1,t}(\Pfrak_*)^\top(\Pfrak-\Pfrak_*)
-\frac{1}{2}\e^{1/2}\|\Pfrak-\Pfrak_*\|^2
\leq \Psi_{\mathrm{RS}}+o_t(1) \cdot \|\Pfrak-\Pfrak_*\|
-\frac{1}{2}\e^{1/2}\varsigma^2.\]
Applying the bound $\|\Pfrak-\Pfrak_*\|<C$ for a constant $C=C(K)>0$
independent of $t$, and taking the limit $t \to \infty$, we obtain
(\ref{eq:analysisf2}).
\end{proof}

\section{Analysis of the conditional second moment}\label{appendix:analysis}

In this appendix we prove Lemmas \ref{lemma:secondmt} and
\ref{lemma:analysisff}. We will abbreviate parts of the arguments that are
similar to the preceding analysis of the conditional first moment, and also
refer to \cite[Lemmas 4.1 and 4.2]{fan2021replica} for some of the technical
details.

\begin{Lemma}\label{lem:abar2}
Let $\pi$ be any probability distribution over $\R$. For $a,b \in \R^2$ and $c
\in \R$, let $c_\pi(a,b,c)$ be as defined in (\ref{eq:defcpit}), and let
\[\Oc=\left\{(a,c) \in \R^3:\int e^{a_1x_1^2+a_2x_2^2+cx_1x_2}
d\pi(x_1)d\pi(x_2)<\infty\right\}.\]
Then $\Oc$ is a non-empty convex subset of $\R^3$. For any $(a,c)$ in the
interior of the complement of $\Oc$ and any $b \in \R^2$,
we have $c_\pi(a,b,c)=\infty$. For any $(a,c)$ in the interior of $\Oc$,
the function $b \mapsto \log c_\pi(a,b,c)$ is continuous and satisfies, for some
$(a,c,\pi)$-dependent constant $C>0$ and for all $b \in \R^2$,
\[\log c_\pi(a,b,c) \leq C(\|b\|^2+1).\]
\end{Lemma}
\begin{proof}
The set $\Oc$ is convex by convexity of the function $(a,c) \mapsto
c_\pi(a,0,c)$, and non-empty because $(a,c)=0$ belongs to $\Oc$. The proofs of
the remaining statements are similar to the proof of Lemma \ref{lem:abar} and
omitted for brevity.
\end{proof}

\begin{proof}[Proof of Lemma \ref{lemma:secondmt}]
Fix $t$ and write $\cG,X,Y,S,\Delta$ for $\cG_t, X_t,Y_t, S_t, \Delta_t$. Then
    \begin{equation*}
        \E[\Z(\overline{\cU})^2\mid\cG]=\int \mathbb{I}\Big(\frac{1}{n}\|\x\|^2
\in \overline{\cU},\,\frac{1}{n}\|\z\|^2 \in \overline{\cU}\Big) \cdot
\exp(-\norm{\epsilon}^2+\frac{n}{2}\cdot f_n(\x,\z))\prod_{i=1}^{n}
d\pi(\s_{i})d\pi(\ta_{i})
    \end{equation*}
    with $\x:=\x(\s)=\s-\st$, $\z:=\z(\ta)=\ta-\st$, and
\[f_{n}(\x,\z):=\frac{2}{n} \log \mathbb{E}\left[\exp
\left(-\frac{\x^{\top} O^{\top} \DDbar O \x}{2}-\frac{\z^{\top} O^{\top} \DDbar O \z}{2}+(\x+\z)^{\top} O^{\top} D^\top \xi \right)
\;\bigg|\; \mathcal{G}\right].\]\\

{\bf Uniform approximation of $f_n(\x,\z)$.} 
Define $\Pfrak(\x,\z)=(u(\x,\z),r(\x,\z),v(\x,\z),w(\x,\z),p(\x,\z))$ by
\[u(\x,\z)=\frac{1}{n}\Big(\|\z\|^{2},\|\x\|^2\Big) \in \R^2, \qquad
p(\x,\z)=\frac{1}{n} \x^\top \z \in \R,\]
\[\mqty(r(\x,\z)^\top \\ v(\x,\z) \\ w(\x,\z))=\qty[\frac{1}{n} \mqty(e^{\top} e &
e^{\top} X & e^{\top} Y \\ X^{\top} e & X^{\top} X & X^{\top} Y\\ Y^\top e &
Y^\top X & Y^\top Y)]^{-1/2}\cdot \frac{1}{n} (e,X,Y)^\top (\x,\z)
\in \R^{(2t+1) \times 2}\]
where $r(\x,\z) \in \R^2$ and $v(\x,\z),w(\x,\z) \in \R^{t \times 2}$.
Define the open domain
\begin{equation}\label{eq:K2domain}
\Kc=\big\{\x,\z \in \R^n:u(\x,\z) \in (0,K)^2,\,
\Ac(\Pfrak(\x,\z)) \succ 0\big\}
\end{equation}
and write again $a_n(\x,\z) \doteq b_n(\x,\z)$ to mean
$a_n(\x,\z)-b_n(\x,\z) \to 0$ uniformly over $(\x,\z) \in \Kc$, almost surely
as $n,m \to \infty$.

Recall $\x_\bot,\x_\|$ from \eqref{eq:sigmaparper},
$\Pi=\Pi_{\left(e_{b}, S, \Lambda S\right)^{\perp}}\in \R^{n\times
(n-2t-1)}$, and denote similarly
\[\z_{\perp}=\Pi_{(e, X, Y)^{\perp}}^{\top} \z \in \mathbb{R}^{n-2 t-1}, \quad
\z_{\|}=\left(e_{b}, S, \Lambda S\right)\begin{pmatrix}
	e^{\top} e & e^{\top} X & e^{\top} Y \\
	X^{\top} e & X^{\top} X & X^{\top} Y \\
	Y^{\top} e & Y^{\top} X & Y^{\top} Y
\end{pmatrix}^{-1}(e, X, Y)^{\top} \z \in \mathbb{R}^{n}.\]
Observe that, similarly to (\ref{eq:lengthre}), we have
\[\Ac(\Pfrak(\x,\z))=\frac{1}{n}\begin{pmatrix} \|\x_\perp\|^2 & \x_\perp^\top
\z_\perp \\ \x_\perp^\top \z_\perp & \|\z_\perp\|^2 \end{pmatrix}.\]
The condition $\Ac(\Pfrak(\x,\z)) \succ 0$ defining $\Kc$ then requires that
$\x_\perp$ and $\z_\perp$ are non-zero and linearly independent.
Then for a sufficiently small constant $\df>0$,
an application of \cite[Lemma B.2]{fan2022tap} and
\cite[Proposition 2.8]{fan2021replica} yields
\begin{equation}\label{eq:fnappt}
    f_n(\x) \doteq -\frac{1}{n}\left(\x_\|^\top \DD \x_\|+\z_\|^\top \DD
\z_\|\right)+\frac{2}{n} (\x_\|+\z_\|)^\top D^\top \xi +E_n(\x,\z)
\end{equation}
where 
\begin{equation}\label{eq:Enappt}
\begin{aligned}
E_n(\x,\z)&=\inf_{\Tmr \succeq (-d+\df)I}\Bigg\{\frac{1}{n}\Tr\left[\Tmr \cdot
\mqty(\norm{\x_\bot}^2 & \x_\bot^\top \z_\bot \\ \x_\bot^\top \z_\bot
&\norm{\z_\bot}^2)\right]\\
&\qquad+\frac{1}{n}\mqty(\Pi^{\top} D^\top \xi-\Pi^{\top} \DD \x_{\|}\\
\Pi^{\top} D^\top \xi-\Pi^{\top} \DD \z_{\|})^\top \qty(\Tmr\oplus \Pi^\top \DD
\Pi)^{-1} \mqty(\Pi^{\top} D^\top \xi-\Pi^{\top} \DD \x_{\|}\\ \Pi^{\top} D^\top
\xi-\Pi^{\top} \DD \z_{\|})\\
&\qquad-\frac{1}{n} \log \det \left(\Tmr\oplus \Pi^\top \DD \Pi\right)-\left(2+\log \det \frac{1}{n}\mqty(\norm{\x_\bot}^2 & \x_\bot^\top \z_\bot \\ \x_\bot^\top \z_\bot &\norm{\z_\bot}^2) \right)
\Bigg\}
\end{aligned}
\end{equation}
and we define
\[\Tmr\oplus \Pi^\top \DD \Pi:= \Tmr \otimes I_{(n-2t-1)\times
(n-2t-1)}+I_{2\times 2} \otimes \Pi^\top \DD \Pi.\]
Similarly to (\ref{eq:sDDs}), (\ref{eq:sDx}), and (\ref{eq:scH}) from
Lemma \ref{lemma:firstmt}, we have
\begin{align}
\frac{1}{n}(\x^\top_\| \DD \x_\|+\z^\top_\| \DD \z_\|)
&\doteq \sum_{i=1}^2 \frac{e_{*}r_i(\x,\z)^2}{b_{*}}
+\Tr \begin{pmatrix} d_* & a_* \\ a_* & c_* \end{pmatrix}
\left(v_i(\x,\z),\frac{w_i(\x,\z)}{\sqrt{\kappa_*}}\right)^\top
\left(v_i(\x,\z),\frac{w_i(\x,\z)}{\sqrt{\kappa_*}}\right),\label{eq:sDDs2}\\
\frac{2}{n}(\x_{\|}+\z_{\|})^\top D\xi &\doteq
\sum_{i=1}^2 \frac{2r_i(\x,\z)}{\sqrt{b_*}},\label{eq:sDx2}
\end{align}
and
\begin{align}
&\frac{1}{n}\Tr\left[\Tmr \cdot \mqty(\norm{\x_\bot}^2 & \x_\bot^\top \z_\bot \\
\x_\bot^\top \z_\bot &\norm{\z_\bot}^2)\right]
-\frac{1}{n} \log \det \left(\Tmr\oplus \Pi^\top \DD \Pi\right)-\left(2+\log \det
\frac{1}{n}\mqty(\norm{\x_\bot}^2 & \x_\bot^\top \z_\bot \\ \x_\bot^\top \z_\bot
&\norm{\z_\bot}^2) \right)\nonumber\\
&\qquad\doteq \Ht(\Tmr, \Aft(\Pfrak(\x,\z))).\label{eq:scH2}
\end{align}
This last approximation (\ref{eq:scH2})
holds uniformly over $\Tmr \succeq (-d_-+\df)I$, 
by the same argument as in the proof of \cite[Lemma 4.2]{fan2021replica}.

For the remaining second term of $E_n(\x,\z)$, write the eigen-decompositions
\[\Tmr=\mqty(y_{1} & y_{2}
)\mqty(\tmr_{1} & 0 \\
0 & \tmr_{2}
)\mqty(
y_{1}^{\top} \\
y_{2}^{\top}
), \qquad
\Pi^{\top}\DD \Pi=\Pi^{\prime \top} D^{\prime} \Pi^{\prime}=\Pi^{\prime
\top} \operatorname{diag}\left(d_{1}^{\prime}, \ldots, d_{n-2
t-1}^{\prime}\right) \Pi^{\prime}\]
where $\tmr_1,\tmr_2$ and
$d_1^{\prime},\ldots,d_{n-2t-1}^{\prime}$ are the eigenvalues
of $\Tmr$ and $\Pi^\top \DD \Pi$ respectively, and
$y_1,y_2 \in \mathbb{R}^{2}$ and the rows of
$\Pi'\in \R^{(n-2t-1)\times(n-2t-1)}$ are the eigenvectors. Then
\[(\Tmr\oplus \Pi^\top \DD \Pi)^{-1}=\mqty(\Pi' & 0\\ 0 & \Pi')^\top
\mqty(\Tmr_{11}\cdot I+D' & \Tmr_{12}\cdot I \\ \Tmr_{12}\cdot I &
\Tmr_{22}\cdot I+D')^{-1} \mqty(\Pi' & 0\\ 0 & \Pi')\]
and we may compute the inverse on the right side
by inverting separately the $2\times 2$ blocks,
$$
\left(\begin{array}{cc}
\Tmr_{11}+d_{i}^{\prime} & \Tmr_{12} \\
\Tmr_{12} & \Tmr_{22}+d_{i}^{\prime}
\end{array}\right)^{-1}=\frac{1}{\tmr_{1}+d_{i}^{\prime}} y_{1}
y_{1}^{\top}+\frac{1}{\tmr_{2}+d_{i}^{\prime}} y_{2} y_{2}^{\top}.
$$
Then for each $j,k\in\qty{1,2}$, the $(j,k)$ block of
$(\Tmr\oplus \Pi^\top \DD \Pi)^{-1}$ is
\begin{equation}\label{eq:blockinverse}
(\Tmr\oplus \Pi^\top \DD \Pi)^{-1}_{jk}=y_{1 j} y_{1 k}\left(\tmr_{1}
I+\Pi^{\top} \DD \Pi\right)^{-1}+y_{2 j} y_{2 k}\left(\tmr_{2} I+\Pi^{\top}
\DD \Pi\right)^{-1}.
\end{equation}
Recall $\Ae,\Be$ from \eqref{eq:defFFeterms}, and define
$\Xe(\tmr)=\Xe_{22}(\tmr)-\Xe_{12}(\tmr)\Xe_{11}(\tmr)^{-1}{\Xe_{12}(\tmr)}^\top
\in \R^{2 \times 2}$ where
\begin{equation*}
    \Xe_{22}(\tmr)=\E \frac{\mathsf{E}_b^2}{\tmr+\D^2} \mqty(\Ae(\D^2) \\ \Be(\D^2)) \mqty(\Ae(\D^2) \\ \Be(\D^2))^\top, \quad \Xe_{12}(\tmr)=\E \frac{\mathsf{E}_b^2}{\tmr+\D^2} \mqty(\Ae(\D^2) \\ \Be(\D^2)),\quad \Xe_{11}(\tmr)=\E \frac{\mathsf{E}_b^2}{\tmr+\D^2}.
\end{equation*}
Note in particular that $\Fc^e$ defined in (\ref{eq:defFFF})
is given by $\Fc^e(\tmr,r)=(r\; 1)\Xe(\tmr)(r\; 1)^\top$.
Then, applying the argument leading to (\ref{eq:sectermPIlim}) separately
for each of the four blocks $j,k \in \{1,2\}$ of (\ref{eq:blockinverse}),
we have
\begin{align}
    &\frac{1}{n}\mqty(\Pi^{\top} D^\top \xi-\Pi^{\top} \DD \x_{\|}\\ \Pi^{\top} D^\top
\xi-\Pi^{\top} \DD \z_{\|})^\top \qty(\Tmr\oplus \Pi^\top \DD \Pi)^{-1} \mqty(\Pi^{\top}
D^\top \xi-\Pi^{\top} \DD \x_{\|}\\ \Pi^{\top} D^\top \xi-\Pi^{\top} \DD
\z_{\|})\nonumber\\
    &\doteq \sum_{j,k=1}^2 y_{1j}y_{1k}
\left[\begin{pmatrix}r_j(\x,\z) & 1 \end{pmatrix}
\Xe(\tmr_1)\begin{pmatrix} r_k(\x,\z) \\ 1 \end{pmatrix}
+\Fc(\tmr_1) \cdot \Bc(v(\x,\z),w(\x,\z))_{jk} \right]\nonumber\\
&\qquad\qquad+y_{2j}y_{2k}
\left[\begin{pmatrix}r_j(\x,\z) & 1 \end{pmatrix}
\Xe(\tmr_2)\begin{pmatrix} r_k(\x,\z) \\ 1 \end{pmatrix}
+\Fc(\tmr_2) \cdot \Bc(v(\x,\z),w(\x,\z))_{jk} \right]\nonumber\\
&=\Tr\left[\sum_{i=1}^2 \mqty(r_1(\x,\z) & r_2(\x,\z)\\ 1 & 1)^\top \Xe(\tmr_i)
\mqty(r_1(\x,\z) & r_2(\x,\z)\\ 1 & 1) y_i y_i^\top\right]
+\Tr\Big[\Ft(\Tmr)\cdot \Bc(v(\x,\z),w(\x,\z))\Big]\label{eq:sectermPIlim2}
\end{align}
uniformly over $\Tmr \succeq (-d_-+\df)I$,
where $\Fc$ is applied in the second term spectrally to
$\Tmr$ via functional calculus. 

From the form of $\Fc$ in (\ref{eq:altFFe}),
the second term of (\ref{eq:sectermPIlim2}) may be expressed as
\begin{equation}\label{eq:FB2alt}
\Tr[\Ft(\Tmr)\cdot \Bc(v,w)]=\inf_{\Fpnut,\Fpomt \in \R^2}
\sum_{i=1}^2
\mathbb{E}\left[\frac{1}{\tmr_i+\D^2}\left(\theta\left(\D^2\right)-\lambda\left(\D^2\right)
\Fpnu_i -\Fpom_i\right)^{2}\right]y_i^\top \Bc(v,w)y_i.
\end{equation}
For the first term, for $r=(r_1,r_2) \in \R^2$, set
\[\fe(x,r)=\begin{pmatrix} \fe(x,r_1) \\ \fe(x,r_2) \end{pmatrix}
=\begin{pmatrix} r_1 & r_2 \\ 1 & 1 \end{pmatrix}^\top \begin{pmatrix}
\Ae(x) \\ \Be(x) \end{pmatrix} \in \R^2,\]
where $\fe(x,r_i)$ is the function defined in (\ref{eq:defFFeterms}). Then
define
\begin{align}
    \Ft^e(\Tmr,r)&=\inf_{\Fpet \in \R^{2}} \E \qty
[\qty(\mqty(\frac{\gamma_*}{\eta_*}-b_*^{-1/2}r_1 \\
\frac{\gamma_*}{\eta_*}-b_*^{-1/2}r_2) \D^2-\Fpet)^\top
\qty(\mathsf{E}_b^2(\Tmr+\D^2\cdot
I)^{-1})\qty(\mqty(\frac{\gamma_*}{\eta_*}-b_*^{-1/2}r_1 \\
\frac{\gamma_*}{\eta_*}-b_*^{-1/2}r_2) \D^2-\Fpet)]\nonumber\\
&=\inf_{\Fpet \in \R^{2}} \E \qty
[\Big(\fe(\D^2,r)-\Fpet\Big)^\top \qty(\mathsf{E}_b^2(\Tmr+\D^2\cdot
I)^{-1})\Big(\fe(\D^2,r)-\Fpet\Big)]\label{eq:Fe2alt}
\end{align}
where these expressions are equivalent by an additive shift of $\Fpet$.
Evaluating explicitly the infimum over $\Fpet$, we get
$\Ft^e(\Tmr,r)=\Ft^e_{22}(\Tmr,r)-\Ft^e_{12}(\Tmr,r)^\top
\Ft^e_{11}(\Tmr)^{-1} \Ft^e_{12}(\Tmr,r)$ where
\begin{equation}\label{eq:Feterms2}
\begin{aligned}
\Ft^e_{22}(\Tmr,r)&=\E\,\fet(\D^2,r)^\top
\qty(\mathsf{E}_b^2(\Tmr+\D^2\cdot I)^{-1})\fet(\D^2,r)\\
&=\E\sum_{i=1}^2 \begin{pmatrix} \Ae(\D^2) \\ \Be(\D^2)
\end{pmatrix}^\top \mqty(r_1 & r_2 \\ 1 & 1) \cdot
\frac{\Es_b^2 y_i y_i^\top}{\tmr_i+\D^2} \cdot
\mqty(r_1 & r_2 \\ 1 & 1)^\top \begin{pmatrix} \Ae(\D^2) \\ \Be(\D^2)
\end{pmatrix}\\
&=\Tr\left[\sum_{i=1}^2 \mqty(r_1 & r_2 \\ 1 & 1)^\top \Xe_{22}(\tmr_i)
\mqty(r_1 & r_2 \\ 1 & 1) y_iy_i^\top\right],\\
\Ft^e_{12}(\Tmr,r)&=\E\,\mathsf{E}_b^2(\Tmr+\D^2\cdot I)^{-1}\fet(\D^2,r)\\
&=\E\sum_{i=1}^2 \frac{\mathsf{E}_b^2 y_iy_i^\top}{\tmr_i+\D^2} \cdot
\mqty(r_1 & r_2 \\ 1 & 1)^\top \mqty(\Ae(\D^2) \\ \Be(\D^2))
=\sum_{i=1}^2 y_iy_i^\top \mqty(r_1 & r_2 \\ 1 & 1)^\top \Xe_{12}(\tmr_i),\\
\Ft^e_{11}(\Tmr)&=\E\,\mathsf{E}_{b}^{2}(\mathfrak{Z}+\D^{2}\cdot I)^{-1}
=\E\sum_{i=1}^2 \frac{\mathsf{E}_b^2 y_iy_i^\top}{\tmr_i+\D^2}
=\sum_{i=1}^2 y_iy_i^\top \Xe_{11}(\tmr_i).
\end{aligned}
\end{equation}
Then it follows that the first term of (\ref{eq:sectermPIlim2}) has the form
\begin{equation}\label{eq:sectermPIlim2alt}
    \Tr\left[\sum_{i=1}^2 \mqty(r_1 & r_2 \\ 1 & 1)^\top \Xe(\tmr_i)
\mqty(r_1 & r_2 \\ 1 & 1) y_i y_i^\top\right]=\Ft^e(\Tmr,r).
\end{equation}
Combining (\ref{eq:fnappt}), (\ref{eq:Enappt}), (\ref{eq:sDDs2}),
(\ref{eq:sDx2}), (\ref{eq:scH2}), (\ref{eq:sectermPIlim2}), and
(\ref{eq:sectermPIlim2alt}),
we obtain the uniform approximation
\[\lim_{n,m \to \infty} \sup_{(\x,\z) \in \Kc}
\Big|f_n(\x,\z)-f(\Pfrak(\x,\z))\Big|=0\]
where we define on the domain $\Cstt=\{\Pfrak:\Ac(\Pfrak) \succ 0\}$ the
function
\begin{align*}
f(\Pfrak)&=\inf_\Tmr \bigg(-\frac{e_*\|r\|^2}{b_*}
+\frac{2r^\top 1_{2 \times 1}}{\sqrt{b_*}}-
\Tr \begin{pmatrix} d_* & a_* \\ a_* & c_* \end{pmatrix}
\left[\left(v_1,\frac{w_1}{\sqrt{\kappa_*}}\right)^\top
\left(v_1,\frac{w_1}{\sqrt{\kappa_*}}\right)
+\left(v_2,\frac{w_2}{\sqrt{\kappa_*}}\right)^\top
\left(v_2,\frac{w_2}{\sqrt{\kappa_*}}\right)\right] \\
&\qquad+\Ht(\Tmr,\Ac(\Pfrak))+\Ft^e(\Tmr,r)+\Tr[\Ft(\Tmr)\cdot \Bc(v,w)]\bigg)
\end{align*}
and the infimum is over $\Tmr \succeq (-d_-+\df)I$. It is immediate from the
forms (\ref{eq:FB2alt}) and (\ref{eq:Fe2alt}) that $\Ft^e(\Tmr,r)$
and $\Tr[\Ft(\Tmr) \cdot \Bc(v,w)]$ are decreasing in the eigenvalues
$\tmr_1,\tmr_2$ of $\Tmr$.
The same argument as in the proof of \cite[Lemma 4.2]{fan2021replica} shows that
$\Ht(\Tmr,\Ac(\Pfrak))$ is also decreasing in each eigenvalue $\tmr_1,\tmr_2$
over the range $(-d_-,-d_-+\df]$, and
hence this infimum may be extended to the domain $\Tmr \succ -d_- \cdot I$.
Finally, since $f_n$ is continuous on $\overline{\Kc}$ and the map
$\Pfrak:\overline{\Kc} \to \{\Pfrak \in \overline{\Cstt}:u_1,u_2 \in [0,K]\}$
is continuous, relatively open, and maps $\Kc$ to the interior
$\{\Pfrak \in \Cstt:u_1,u_2 \in (0,K)\}$ for each fixed $n$,
\cite[Proposition C.1]{fan2021replica} shows that $f$ extends continuously to
$\{\Pfrak \in \overline{\Cstt}:u_1,u_2 \in [0,K]\}$, and
\[\lim_{n,m \to \infty} \sup_{(\x,\z) \in \overline{\Kc}}
\Big|f_n(\x,\z)-f(\Pfrak(\x,\z))\Big|=0.\]
Then, writing $\langle \cdot \rangle_\pi$ for the expectation over
$(\s_i)_{i=1}^n,(\ta_i)_{i=1}^n \overset{iid}{\sim} \pi$, we obtain
\[\lim_{n,m \to \infty} \frac{1}{n}\log 
\E[\Z(\overline{\cU})^2\mid\cG]
=-1+\lim_{n,m \to \infty} \frac{1}{n}\log
\left\langle \mathbb{I}\left\{u(\x,\z) \in \overline{\cU} \times \overline{\cU}
\right\} \exp\Big(\frac{n}{2}f(\Pfrak(\x,\z))\Big)\right\rangle_\pi.\]\\

{\bf Large deviations analysis.} Introduce dual variables
$\Rfrak=(U,R,V,W,P)$ where $U,R \in \R^2$, $V,W \in \R^{t \times 2}$, and $P \in
\R$. Define
\begin{align*}
\lambda_n(\Rfrak)&=\frac{1}{n} \log \left\langle \exp \left(n\cdot
\Pfrak(\x,\z)^\top \Rfrak \right) \right\rangle_\pi\\
\lambda(\Rfrak)&=\E \log \cpit (U,\As(\Rfrak),P)
+\rho_*\,U^\top 1_{2 \times 1}-\pi_*\,V_1^\top \Delta_t^{-1/2}1_{t\times 1}
-\pi_*\,V_2^\top \Delta_t^{-1/2}1_{t\times 1}+\rho_*\,P
\end{align*}
where $\cpit$ is the function from (\ref{eq:defcpit}), and where
\[\As(\Rfrak)=-2U\,\Xstar+\frac{R\,\mathsf{E}}{\sqrt{b_*}}
+V^{\top} \Delta_t^{-1/2}(\Xs_{1}, \ldots, \Xs_{t})+\frac{W^{\top}
\Delta_t^{-1/2}(\Ys_{1}, \ldots, \Ys_{t})}{\sqrt{\kappa_*}}-P\,\Xstar 1_{2
\times 1} \in \R^2.\]
Then, applying Lemma \ref{lem:abar2} and the same argument as leading to
(\ref{eq:MGFconvergence}), we have almost surely
\begin{align*}
\lim_{n,m \to \infty} \lambda_n(\Rfrak)&=\lambda(\Rfrak)<\infty
\text{ if } (U,P) \text{ is in the interior of } \Oc,\\
\lim_{n,m \to \infty} \lambda_n(\Rfrak)&=\lambda(\Rfrak)=\infty
\text{ if } (U,P) \text{ is in the interior of the complement of } \Oc
\end{align*}
where $\Oc$ is the convex set defined in Lemma \ref{lem:abar2}.
Under the sub-Gaussian condition (\ref{eq:subGaussian}), $(U,P)=0$
belongs to the interior of $\Oc$, so the Fenchel-Legendre dual $\lambda^*$ of
$\lambda$ is a good convex rate function \cite[Lemma 2.3.9(a)]{Dembo1998large}.
Defining $\bar{\lambda}(\Rfrak)=\limsup_{n,m \to \infty} \lambda_n(\Rfrak)$,
this coincides with $\lambda(\Rfrak)$ whenever $(U,P) \notin \partial \Oc$,
the boundary of $\Oc$.
For $(U,P)\in \partial \Oc$, since $\Oc$ is convex
and 0 belongs to the interior of $\Oc$, the open line segment
$\{s \cdot (U,P): s \in [0,1)\}$ also belongs to the interior of $\Oc$
\cite[Theorem 6.1]{rockafellar2015convex}. The Gale-Klee-Rockafellar Theorem
shows that $\lambda,\bar{\lambda}$ are upper-semicontinuous on
$\{s \cdot (U,P): s \in [0,1]\}$, so the supremum defining the
Fenchel-Legendre dual $\lambda^*(\Pfrak)=\sup_{\Rfrak} \Pfrak^\top
\Rfrak-\lambda(\Rfrak)$ may then be restricted to $\Rfrak$ where $(U,P) \notin
\partial \Oc$, and similarly for $\bar{\lambda}$. Then
$\lambda^*$ coincides with the Fenchel-Legendre dual of $\bar{\lambda}$, and
the proof is concluded as in \Cref{lemma:firstmt} by an application of the
upper bound in the G\"artner-Ellis Theorem and Varadhan's lemma.
\end{proof}

We now prove Lemma \ref{lemma:analysisff}. Let $e_t=(0,\ldots,0,1) \in \R^t$,
and consider $\Pfrak_*,\Qfrak_*$ with the components
$$
\begin{aligned}
&u_*=\frac{2}{\eta_*}1_{2 \times 1},
\quad r_*=\frac{\gamma_*}{\eta_*}b_*^{1/2}1_{2 \times 1},
\quad v_*=\frac{\eta_{*}-\gamma_{*}}{\eta_{*}} \Delta_t^{1/2}(e_t,e_t),
\quad w_*=\frac{\gamma_{*}}{\eta_{*}} \kappa_{*}^{1/2} \Delta_t^{1/2} (e_t,e_t),
\quad p_*=\frac{1}{\eta_*}\\
&\Tmr_*=(\eta_*-\gamma_*)\cdot I_{2\times 2}, \quad
\Fpet_*=\Fpnut_*=\Fpomt_*=0,\\
&U_*=-\frac{\gamma_*}{2}1_{2 \times 1},
\quad R_*=\gamma_* b_*^{1/2}1_{2 \times 1},
\quad V_*=0,
\quad W_*=\gamma_* \kappa_{*}^{1 / 2} \Delta_{t}^{1 / 2} (e_t,e_t),
\quad P_*=0.
\end{aligned}
$$
Here, each of the two coordinates/columns of $u_*,r_*,v_*,w_*,U_*,R_*,V_*,W_*$
coincides with our previous specialization (\ref{eq:firstmomentspec}) in the
analysis of the conditional first moment.

\begin{Lemma}\label{lemma:plugstatt}
In the setting of Lemma \ref{lemma:analysisff},
for all $t\ge 1$ and each $\iota\in \qty{u,r,v,w,p,\Tmr,\Fpet,\Fpnut,\Fpomt,U,R,W,P}$,
    \begin{equation}\label{eq:plugstatt}
        \Phi_{2,t}(\Pfrak_*,\Qfrak_*)=2\Psi_{\mathrm{RS}}, \quad \partial_\iota
\Phi_{2,t}(\Pfrak_*,\Qfrak_*)=0, \qquad
\lim_{t\to \infty} \norm{\partial_V \Phi_{2,t}(\Pfrak_*,\Qfrak_*)}=0.
    \end{equation}
\end{Lemma}
\begin{proof}
At $P_*=0$, we have $\log \cpit (a_1,a_1,b_1,b_2,0)=\log c_\pi (a_1,b_1)+\log
c_\pi (a_2,b_2)$ where $c_\pi$ on the right side is defined by \eqref{eq:defFH}.
At $\Pfrak_*$, we have also
$\Ac(\Pfrak_*)=\eta_*^{-1}\cdot I_{2\times 2}$ by the same calculation as
(\ref{eq:APstar}), and $\Bc(v_*,w_*)=0$ because $v_*=\alpha_*^A w_*$.
Since both $\Ac(\Pfrak_*)$ and $\Tmr_*$ are diagonal,
it is then easily checked that
\[\Phi_{2,t}(\Pfrak_*,\Qfrak_*)=2\Phi_{1,t}(\Pfrak_*,\Qfrak_*)\]
where $\Pfrak_*,\Qfrak_*$ on the right side are the specializations of
(\ref{eq:firstmomentspec}) from our previous analysis of the conditional first
moment. Then $\Phi_{2,t}(\Pfrak_*,\Qfrak_*)=2\Psi_{\mathrm{RS}}$ follows from
Lemma \ref{lemma:plugstat}.

To check the stationarity conditions, define
\begin{align*}
\As_*&=-2U_*\Xstar+b_{*}^{-1/2} R_* \mathsf{E}+V^{\top}_*
\Delta^{-1/2}_t\left(\Xs_{1}, \ldots, \Xs_{t}\right)+\kappa_{*}^{-1/2}
W^{\top}_* \Delta^{-1/2}_t\left(\Ys_{1}, \ldots, \Ys_{t}\right)
-P_*\Xstar 1_{2 \times 1}\\
&=\gamma_*(\Xstar+\Es+\Ys_t)1_{2 \times 1}.
\end{align*}
Then we obtain, similarly to \eqref{eq:Hcderivs}, \eqref{eq:Federivs},
and \eqref{eq:logcpiabst},
\[\partial_{\Tmr} \Ht(\Tmr_*,\Ac(\Pfrak_*))=0,
\quad \partial_{\Ac} \Ht(\Tmr_*,\Ac(\Pfrak_*))=-\gamma_* \cdot I_{2\times 2},\]
\[\partial_r,\partial_\Tmr,\partial_{\Fpet}
\E\qty [\qty(\left[\frac{\gamma_*}{\eta_*}1_{2 \times
1}-\frac{r_*}{\sqrt{b_*}}\right]\D^2-\Fpet_*)^\top
\Big(\mathsf{E}_b^2(\Tmr_*+\D^2\cdot I_{2 \times
2})^{-1}\Big)\qty(\left[\frac{\gamma_*}{\eta_*}1_{2 \times
1}-\frac{r_*}{\sqrt{b_*}}\right]\D^2-\Fpet_*)]=0,\]
and
\begin{align*}
\partial_{a_1} \log \cpit(U_*,\As_*,P_*)&=\partial_{a_2} \log
\cpit(U_*,\As_*,P_*)=\gamma^{-1}_*f'(\Xstar+\mathsf{E}+\Ys_t)+f(\Xstar+\mathsf{E}+\Ys_t)^2,\\
\partial_{b_1} \log \cpit(U_*,\As_*,P_*)&=\partial_{b_2} \log
\cpit(U_*,\As_*,P_*)=f(\Xstar+\mathsf{E}+\Ys_t),\\
\partial_c \log \cpit(U_*,\As_*,P_*)&=f(\Xstar+\mathsf{E}+\Ys_t)^2.
\end{align*}
Using these, $\Bc(v_*,w_*),\partial_v \Bc(v_*,w_*),\partial_w \Bc(v_*,w_*)=0$,
and the identities (\ref{eq:fsqid}) and (\ref{eq:fxid}), we have
\[\partial_\Tmr \Phi_{2,t}(\Pfrak_*,\Qfrak_*)=0,
\quad \partial_p \Phi_{2,t}(\Pfrak_*,\Qfrak_*)=-P_*=0,\]
\[\partial_P \Phi_{2,t}(\Pfrak_*,\Qfrak_*)
=\E f(\Xstar+\mathsf{E}+\Ys_t)^2-2\E \Xstar
f(\Xstar+\mathsf{E}+\Ys_t)-(p_*-\rho_*)=0,\]
and the analyses of the remaining derivatives are the same as in
Lemma \ref{lemma:plugstat}.
\end{proof}

\begin{proof}[Proof of Lemma \ref{lemma:analysisff}]
The proof is analogous to the upper bound of \Cref{lemma:analysisf}.
We fix $\hf \in \R_+$ and
specialize the dual variables $\Qfrak=\Qfrak(\Pfrak)$ as functions of
$\Pfrak=(u,r,v,w,p)$, given by
\begin{equation}\label{eq:ukar}
    \begin{aligned}
&U(u)=U_*+\hf (u-u_*), \quad
R(r)=R_*+\hf (r-r_*),\quad V(v)=\hf (v-v_*),\quad W(w)=W_*+\hf (w-w_*),\\
&P(p)=\hf (p-p_*), \quad \Tmr=G^{-1}(\Ac(\Pfrak)),
\quad \Fpet(r)=\mqty(\frac{\gamma_*(\eta_*-\gamma_*)}{\eta_*}-b_*^{-1})
1_{2 \times 1}+b_*^{-3/2}e_*r,\quad \Fpnut=\Fpomt=0
    \end{aligned}
\end{equation}
where $G^{-1}(\Ac)$ is defined spectrally by functional calculus. Then
\begin{equation}\label{eq:inviewt}
\Psi_{2,t}(\Pfrak)=\inf_{\Qfrak} \Phi_{2,t}(\Pfrak,\Qfrak)
\leq \Phi_{2,t}(\Pfrak,\Qfrak(\Pfrak))=:
\bar{\Psi}_{2,t}(\Pfrak)=-1+\Tone+\Ttwo+\frac{1}{2}\Big(\Tthree+\Tfour+\Tfive+\Tsix\Big)
\end{equation}
where
\begin{equation}\label{eq:FSTTs2t}
    \begin{aligned}
&\Tone=\E \log \cpit\left(U(u),\;-2\Xstar\,U(u)-\Xstar P(p)\,1_{2 \times 1}
+\begin{pmatrix} R(r) & V(v)^\top & W(w)^\top \end{pmatrix}^\top \Fs,\;P(p)\right)\\
&\Ttwo= -(u-\rho_* 1_{2 \times 1})^\top U(u)-r^\top R(r)-(v_1+\pi_*
\Delta_t^{-1/2}1_{t \times 1})^\top V_1(v)
-(v_2+\pi_* \Delta_t^{-1/2}1_{t \times 1})^\top V_2(v)\\
&\qquad\qquad-w_1^\top W_1(w)-w_2^\top W_2(w)-(p-\rho_*)P(p)\\
&\Tthree=-\frac{e_*\|r\|^2}{b_*}+\frac{2r^\top 1_{2 \times 1}}{\sqrt{b_*}}
-\Tr \begin{pmatrix} d_* & a_* \\ a_* & c_* \end{pmatrix}
\left[\left(v_1,\frac{w_1}{\sqrt{\kappa_*}}\right)^\top
\left(v_1,\frac{w_1}{\sqrt{\kappa_*}}\right)
+\left(v_2,\frac{w_2}{\sqrt{\kappa_*}}\right)^\top
\left(v_2,\frac{w_2}{\sqrt{\kappa_*}}\right)\right]\\
&\Tfour=\Ht\big(\Tmr(\Pfrak),\Ac(\Pfrak)\big),
\quad \Tfive=\Ft^e_{22}(\Tmr(\Pfrak),r), \quad
\Tsix=\Tr[\Ft_{22}(\Tmr(\Pfrak)) \cdot \Bc(v,w)].
    \end{aligned}
\end{equation}
Here $\Fs$ in $\Tone$ is the tuple of random variables defined in
(\ref{eq:Fsdef}), the function $\Ft_{22}^e(\Tmr,r)=\E[\fet(\D^2,r)^\top (\Es_b^2(\Tmr+\D^2 \cdot
I)^{-1}) \fet(\D^2,r)]$ in $\Tfive$ is as previously defined in
(\ref{eq:Feterms2}), and the function
$\Ft_{22}(\zeta)=\E[\theta(\D^2)^2/(\zeta+\D^2)]$ in $\Tsix$
is as previously defined
in (\ref{eq:defFFeterms}) and applied to $\Tmr$ by functional calculus. We
observe that $\Qfrak(\Pfrak_*)=\Qfrak_*$, so
$\bar{\Psi}_{2,t}(\Pfrak_*)=2\Psi_{\mathrm{RS}}$. An analysis similar to that of
Lemma \ref{lemma:analysisf} shows also that $\|\nabla
\bar{\Psi}_{2,t}(\Pfrak_*)\|=o_t(1)$.

We now show that for sufficiently small $\e_0=\e_0(\pib,K)>0$ and $\e<\e_0$,
this function $\bar{\Psi}_{2,t}(\Pfrak)$ is concave over
$\{\Pfrak \in \Cstt:u_1,u_2 \in (0,K)\}$, by analyzing the Hessian of each term
$\Tone$-$\Tsix$. Fix $\Pfrak=(u,r,v,w,p) \in \Cstt$ with $u_1,u_2 \in (0,K)$,
fix a unit vector $\Pfrak'=(u',r',v',w',p')$, and define for $s>0$
\[\Pfrak(s)=(u(s),r(s),v(s),w(s),p(s))
=(u,r,v,w,p)+s\cdot (u',r',v',w',p').\]
For $\Tone$, recalling the random vector $\Fs \in \R^{2t+1}$
from (\ref{eq:Fsdef}), define
\[\Qs(x_1,x_2)=\begin{pmatrix}
x_1^2-2x_1\Xstar \\ x_2^2-2x_2\Xstar \\ 
x_1x_2-x_1\Xstar-x_2\Xstar
\\ x_1 \Fs \\ x_2 \Fs \\ \end{pmatrix} \in \R^{4t+5}.\]
Vectorize $\Qfrak(\Pfrak)$ in the corresponding
order $(U_1,U_2,P,R_1,V_1,W_1,R_2,V_2,W_2) \in \R^{4t+5}$ and
vectorize similarly $\Pfrak'$, and denote
\[\expval{f(x_1,x_2)}_\Pfrak=\frac{\int f(x_1,x_2)\exp\big(\Qfrak(\Pfrak)^\top
\Qs(x_1,x_2)\big)d\pi(x_1)d\pi(x_2)}{\int \exp\big(\Qfrak(\Pfrak)^\top
\Qs(x_1,x_2)\big)d\pi(x_1)d\pi(x_2)}.\]
Write $\V_\Pfrak[\cdot]$ for the corresponding variance. Then we have, analogous
to (\ref{eq:Isecondderiv}),
\[\partial_s^2 \Tone=\hf^2 \E[\V_\Pfrak[{\Pfrak'}^\top \Qs(x_1,x_2)]].\]
Note that the distribution defining $\langle \cdot \rangle_\Pfrak$
corresponds to $\mu$ in Assumption \ref{AssumpPrior} with $k=2$ and
\[\Gamma=\begin{pmatrix} U_1(u) & \frac{1}{2}P(p) \\ \frac{1}{2}P(p) &
U_2(u) \end{pmatrix}, \qquad z=-2\Xstar\,U(u)-\Xstar P(p)\,1_{2 \times 1}
+\begin{pmatrix} R(r) & V(v)^\top & W(w)^\top \end{pmatrix}^\top \Fs.\]
By Proposition \ref{prop:betalimit}, we have $U_{1,*}=U_{2,*}=-\gamma_*/2
=-(d_*/2)(1+O(\e))$. We have also $P_*=0$ and $u_1,u_2,p,u_{1,*},u_{2,*},p_*
\in (0,K)$ under the conditions $2\eta_*^{-1} \in (0,K)$ and $\Pfrak \in
\Cst$. Thus, choosing $\hf<\hf_0$ for a sufficiently small constant $\hf_0$
depending only on $(K,\pib)$, we have $\Gamma \prec (4\pib)^{-1}I$, and also
$\pib^{-1}-\gamma_{\max} \geq c(1+d_*)$ and $-\gamma_{\min}<C(1+d_\star)$ for
its largest and smallest eigenvalues. The same arguments as leading to
(\ref{Hess1}) show $\|z\|^2 \leq ({\Xstar}^2+(q^\top \Fs)^2) \cdot O(1+d_*)$,
and hence by Assumption \ref{AssumpPrior} and Proposition \ref{prop:bayesbnd},
\[\partial_s^2 \Tone \big|_{s=0}=O(\hf^2).\]

The same arguments as in (\ref{Hess2}--\ref{Hess3}) show
    \[\partial_s^2 \Ttwo \big|_{s=0}=-2\hf, \qquad
    \partial_s^2 \Tthree \big|_{s=0}=-2d_*\norm{(r',v',w')}^2+O(\e)\]
where $(r',v',w') \in \R^{4t+2}$ is its vectorization and $\|\cdot\|$ is its
Euclidean norm. For $\Tfour$, we have by
\cite[Proposition 2.9(b)]{fan2021replica} that
$\Tfour=\Tr f(\Ac(\Pfrak(s)))$ where $f(\alpha)=\int_0^\alpha R(z)dz$.
For all sufficiently small $\e$, we may integrate the series representation
\eqref{eq:Rseries} for $R(z)$
term-by-term to write $f(\Ac(\Pfrak(s)))$ as the convergent matrix series
    \[f(\Ac(\Pfrak(s)))=-\Ac(\Pfrak(s))+\sum_{k\ge 2}\frac{\kappa_k}{k}
\Ac(\Pfrak(s))^k\]
where $|\kappa_k|\le \kappa_2(16 \e)^{k-2}$ and $\kappa_2=O(\e^2)$.
It is easily checked that at $s=0$, we have
$\|\partial_s^k \Ac(\Pfrak(s))\|=O(1)$ for $k=0,1,2$, and in particular
$\partial_s^2 \Ac(\Pfrak(s))=-2(r'{r'}^\top+{v'}^\top v'+{w'}^\top w')
\in \R^{2 \times 2}$, with trace $-2\|(r',v',w')\|^2$. Then, differentiating
$f(\Ac(\Pfrak(s)))$ term-by-term and taking the trace, it follows that
    \[\partial^2_s \Tfour \big|_{s=0} = 2d_*\norm{(r',v',w')}^2+O(\e^2).\]
For $\Tfive$, applying the series expansion \eqref{eq:Vseries} now to the matrix
argument $z=\Ac(\Pfrak(s))$ and differentiating term-by-term, we have
\begin{equation}\label{eq:tmrxO1t}
\sup_{x \in \operatorname{supp}(\D^2)}
\left\|\partial_s^k(\Tmr(\Pfrak(s))+xI)^{-1}\right\| \bigg|_{s=0}=O(1) \text{
for } k=0,1,2.
\end{equation}
Combining with the bound (\ref{eq:federivbound}),
we have as in the proof of Lemma \ref{lemma:analysisf} that
$\partial_s^2 \Tfive \big|_{s=0}=O(\e^2)$.
For $\Tsix$, recalling the bound (\ref{eq:thetasqbound}) and combining this
with \eqref{eq:tmrxO1t}, we obtain at $s=0$ that
$\|\partial_s^k\Ft_{22}(\Tmr(s))\|=O(\kappa_2\e^2/\eta_*^2)$ for $k=0,1,2$.
As in the proof of Lemma \ref{lemma:analysisf}, we have
$\|\partial_s^k \Bc(v(s),w(s))\|=O(\max(1,\eta_*^2/\kappa_2))$
for $k=0,1,2$. Then $\partial_s^2 \Tsix \big|_{s=0}=O(\e^2)$.

Combining the above and setting $\hf=\e^{1/2}$, we conclude that for
$\e<\e_0(\pib,K)$,
    \[\nabla^2 \bar{\Psi}_{2,t}(\Pfrak) \prec -\e^{1/2}\cdot I
\text{ for all } \Pfrak \in \Cst \text{ with } u_1,u_2 \in (0,K).\]
Since $u_* \in \cU \subseteq (0,K)$ by assumption, 
the same argument as in Lemma \ref{lemma:analysisf} shows
$\sup_{\Pfrak \in\Cstt:u_1,u_2 \in \cU} \Psi_{2,t}(\Pfrak)
\leq \bar{\Psi}_{2,t}(\Pfrak_*)+o_t(1)=2\Psi_{\mathrm{RS}}+o_t(1)$,
and taking the limit $t \to \infty$ concludes the proof.
\end{proof}

\section{Concentration of the log-partition
function}\label{appendix:concentration}

\begin{proof}[Proof of Lemma \ref{lemma:concentration}]
Recall the representation \eqref{eq:Econdstart} and \eqref{eq:fnform}
of the conditional law of $\Z(U)$ given $\cG_t$,
\begin{equation}\label{eq:warlo}
    \begin{aligned}
        \frac{1}{n} \log \Z(U) \bigg|_{\cG_t}&\overset{L}{=}
\frac{1}{n} \log \int\prod_{i=1}^{n} d\pi(\s_{i})
 \mathbb{I}\qty(\frac{1}{n}\left\|\x\right\|^{2}
\in U) \times \\
    &\quad\exp \bigg(-\frac{\|\xi\|^{2}}{2}-\frac{(\Pi \tilde{O}
\tilde{\s}_{\perp}+\tilde{\s}_{\|})^{\top} \DD (\Pi \tilde{O} \tilde{\s}_{\perp}+\tilde{\s}_{\|})}{2}
    +\left(\Pi \tilde{O} \tilde{\s}_{\perp}+\tilde{\s}_{\|}\right)^{\top}
D^{\top} \xi\bigg)
    \end{aligned}
\end{equation}
where $\x=\s-\st$ and $\|\x_{\|}\|^2,\|\x_\perp\|^2 \leq \|\x\|^2$.
We denote the right side of (\ref{eq:warlo}) as $F(\tilde{O})$, where
$\tilde{O} \sim \Haar(\SO(n-(2t+1)))$ is independent of $\cG_t$, and all
other quantities defining $F(\tilde{O})$ are $\cG_t$-measurable.
On the event $\cE$,
this integral in (\ref{eq:warlo}) is non-zero, so for any $\tilde{O}
\in \SO(n-(2t+1))$ the Euclidean gradient of $F$ is bounded as
\begin{align*}
\norm{\frac{\partial}{\partial \tilde{O}} F(\tilde{O})}_F &\le
\frac{1}{n} \sup_{\x:\frac{\|\x\|^2}{n} \in U}
\left\|-\tilde{\sigma}_{\perp} \tilde{\sigma}_{\perp}^{\top} \tilde{O}^{\top}
\Pi^{\top} D^{\top} D \Pi-\tilde{\sigma}_{\perp} \tilde{\sigma}_{\|}^{\top} D^{\top}
D \Pi+\tilde{\sigma}_{\perp} \xi^{\top} D \Pi\right\|_{F}\\
        &\leq \frac{1}{n} \sup_{\x:\frac{\|\x\|^2}{n} \in U}
\qty (\norm{\DD}_{\mathrm{op}}\norm{\x_{\perp}}^2+
\norm{\DD}_{\mathrm{op}}\norm{\x_\perp}\norm{\x_\|}+
\norm{\x_\perp}\norm{D^\top \xi})\\
&\leq 2K\|D^\top D\|_{\mathrm{op}}+\sqrt{KL}
\end{align*}
where we applied
$\|\tilde{O}\|_{\mathrm{op}},\|\Pi\|_{\mathrm{op}} \leq 1$,
$\|D^\top \xi\|^2 \leq Ln$, $\norm{uv^\top A}\le
\normop{A}\norm{uv^\top}_F=\normop{A}\|u\|\|v\|$, and the bounds
$\norm{\x_\perp},\norm{\x_\|}\le \norm{\x}\le \sqrt{nK}$ when $U
\subseteq [0,K]$. Then,
applying the assumption $\normop{D^\top D} \to d_+$,
this bound is less than $2K(d_++1)+\sqrt{KL}$ for all large $n$. Then
\[\mathbb{P}\left(\left|\frac{1}{n}\log \Z(\cU)-
\mathbb{E}\left[\frac{1}{n} \log \Z(\cU) \;\bigg|\; \cG_t\right]
\right| \geq \delta \;\bigg|\; \cG_t\right)\mathbb{I}\{\cE\}
\leq 2\exp\left(-\frac{(n-(2t+1)-2)\delta^2}{8(2K(d_++1)+\sqrt{KL})^2}\right)\]
by Gromov's inequality, see
e.g.\ \cite[Theorem 4.4.27]{anderson2010introduction}. Choosing
$C(K,L,d_+)>8(2K(d_++1)+\sqrt{KL})^2$ strictly, the statement
(\ref{eq:conditionalconcentration}) thus holds for all sufficiently large $n$.
\end{proof}

\begin{proof}[Proof of Corollary \ref{cor:PsiRS}]
Set $\overline{\cU}=[0,K]$.
Fix any $\delta'>0$, and fix $L>0$ such that the condition
$\|D^\top \xi\|^2 \leq Ln$ in Lemma \ref{lemma:concentration} holds almost
surely for all large $n$. By Lemma \ref{lemma:concentration}, for a constant
$c_0=c_0(\delta',K,L,\e)>0$, any $t \geq 1$, and all large $n$,
\begin{equation}\label{eq:gromov}
\mathbb{P}\left(\left|\frac{1}{n} \log
\Z(\overline{\cU})-\mathbb{E}\left[\frac{1}{n} \log \Z(\overline{\cU})
\;\bigg|\; \cG_t\right]\right| \geq \delta' \;\bigg|\; \cG_t\right)
\mathbb{I}(\cE) \leq e^{-c_0n}.
\end{equation} 
Now fix any $\delta \in (0,c_0/6)$.
By Lemmas \ref{lemma:firstmt}, \ref{lemma:analysisf},
\ref{lemma:secondmt}, and \ref{lemma:analysisff},
for a large enough iteration $t=t(\delta) \geq 1$ and large enough
$M=M(\delta)>0$, almost surely
\begin{align}
	\limsup_{n,m \rightarrow \infty} \frac{1}{n} \log
\mathbb{E}\left[\Z(\overline{\cU}) \mid \cG_{t}\right]
&\leq \sup_{u \in \cU} \Psi_{1,t}(u)<\Psi_{\mathrm{RS}}+\delta
\label{eq:firstmomentupperbound}\\
	\liminf_{n,m \rightarrow \infty} \frac{1}{n} \log
\mathbb{E}\left[\Z(\overline{\cU}) \mid \cG_{t}\right]
&\geq \sup_{u \in \cU} \Psi_{1,t}^M(u)>\Psi_{\mathrm{RS}}-\delta
\label{eq:firstmomentlowerbound}\\
	\lim_{n,m \rightarrow \infty} \frac{1}{n} \log
\mathbb{E}\left[\Z(\overline{\cU})^2 \mid \cG_{t}\right]
&\leq \sup_{u \in \cU} \Psi_{2,t}(u)<2\Psi_{\mathrm{RS}}+\delta.
\label{eq:secondmomentbound}
\end{align}

Letting $\cE$ be the ($\cG_t$-measurable)
event in Lemma \ref{lemma:concentration}, if the first
condition of $\cE$ does not hold, then we have
$\mathbb{P}[\Z(\overline{\cU})=0 \mid \cG_t]=1$ and
hence $\log \E[\Z(\overline{\cU}) \mid \cG_t]=-\infty$. Thus the finite
lower bound in (\ref{eq:firstmomentlowerbound}) and the above choice of $L>0$
imply that
$\cE$ holds almost surely for all large $n$. Then taking the expectation of
(\ref{eq:gromov}) and applying the Borel-Cantelli lemma, this implies
almost surely for all large $n$,
\[\frac{1}{n} \log
\Z(\overline{\cU})<\mathbb{E}\left[\frac{1}{n} \log \Z(\overline{\cU})
\;\bigg|\; \cG_t\right]+\delta'.\]
Then applying Jensen's inequality and
(\ref{eq:firstmomentupperbound}), almost surely for all large $n$,
\begin{equation}\label{eq:limsuplogZ}
\frac{1}{n} \log \Z(\overline{U})<
\frac{1}{n} \log \E[\Z(\overline{\cU}) \mid \cG_t]+\delta'
<\Psi_{\mathrm{RS}}+\delta+\delta'.
\end{equation}

For the complementary lower bound, let $\cE'$ be the ($\cG_t$-measurable)
event where
\[\frac{1}{n}\log \frac{\mathbb{E}\left[\Z(\overline{\cU}) \mid
\mathcal{G}_{t}\right]}{2}>\Psi_{\mathrm{RS}}-\delta,
\qquad \frac{1}{n}\log \mathbb{E}\left[\Z(\overline{\cU})^2 \mid
\mathcal{G}_{t}\right]<2\Psi_{\mathrm{RS}}+\delta.\]
Then (\ref{eq:firstmomentlowerbound}) and
(\ref{eq:secondmomentbound}) show that $\cE'$ holds almost surely 
for all large $n$. On $\cE \cap \cE'$,
$$
\begin{aligned}
&\mathbb{P}\left[\frac{1}{n} \log \Z(\overline{\cU})>\Psi_{\mathrm{RS}}-\delta
\;\bigg|\; \mathcal{G}_{t}\right]\mathbb{I}(\cE \cap \cE') \stackrel{(a)}{\geq} \mathbb{P}\left[\frac{1}{n} \log
\Z(\overline{\cU}) \geq \frac{1}{n} \log
\frac{\mathbb{E}\left[\Z(\overline{\cU}) \mid \mathcal{G}_{t}\right]}{2}
\;\Bigg|\; \mathcal{G}_{t}\right]\mathbb{I}(\cE \cap \cE')\\
&\qquad\qquad=\mathbb{P}\left[\Z(\overline{\cU}) \geq
\frac{\mathbb{E}\left[\Z(\overline{\cU}) \mid \mathcal{G}_{t}\right]}{2}
\;\bigg|\; \mathcal{G}_{t}\right]\mathbb{I}(\cE \cap \cE')
\stackrel{(b)} \geq \frac{\mathbb{E}\left[\Z(\overline{\cU}) \mid
\mathcal{G}_{t}\right]^{2}}{4 \mathbb{E}\left[\Z(\overline{\cU})^{2} \mid
\mathcal{G}_{t}\right]} \cdot \mathbb{I}(\cE \cap \cE')
\stackrel{(c)}>e^{-3n\delta} \cdot \mathbb{I}(\cE \cap \cE')
\end{aligned}
$$
where (a) and (c) apply the definition of the event $\cE'$, and
(b) applies the Paley-Zygmund inequality. By our choice
$\delta<c_0/6$ where $c_0$ is the constant in (\ref{eq:gromov}), on the event $\cE \cap \cE'$ this last quantity is bounded below by $e^{-c_0n/2}$.
This and (\ref{eq:gromov}) together imply that on the event $\cE \cap \cE'$,
$$
\mathbb{E}\left[\frac{1}{n} \log \Z(\overline{\cU}) \;\bigg|\;\cG_t\right]
>\Psi_{\mathrm{RS}}-\delta-\delta'.
$$
Then multiplying (\ref{eq:gromov}) by $\mathbb{I}(\cE')$, taking the
expectation on both sides, and applying the Borel-Cantelli lemma and the 
statement that $\cE \cap \cE'$ holds almost surely for all large $n$,
we get
\begin{equation}\label{eq:liminflogZ}
\frac{1}{n} \log \Z(\overline{\cU})
>\mathbb{E}\left[\frac{1}{n} \log \Z(\overline{\cU}) \;\bigg|\;\cG_t\right]
-\delta'>\Psi_{\mathrm{RS}}-\delta-2\delta'
\end{equation}
almost surely for all large $n$.
The result follows upon taking $\delta,\delta' \to 0$ in
(\ref{eq:limsuplogZ}) and (\ref{eq:liminflogZ}).
\end{proof}

\section{Proofs for unbounded support}\label{appendix:logconcavemain}

In this appendix, we complete the proofs of Theorems \ref{thm:maintheorem},
\ref{thm:MMSE}, and \ref{thm:TAP} in the more general setting of Assumption
\ref{AssumpPrior} where $\pi$ may have unbounded support.

\begin{proof}[Proof of \Cref{thm:maintheorem}, unbounded support]
We apply a truncation argument. First note from (\ref{eq:RSdef}) that
\begin{equation*}
    \begin{aligned}
        \Psi_{\mathrm{RS}}&\stackrel{(a)}{\ge}-\frac{1}{2}-\frac{\gamma_*\rho_*}{2}+\frac{\gamma_*}{2\eta_*}-\frac{d_*}{2\eta_*}+\mathbb{E} \log c_{\pi}\left(-\frac{1}{2} \gamma_{*},
\gamma_{*} \Xstar+\sqrt{\gamma_{*}} \Zs\right)\\
        & \stackrel{(b)}{\ge }-\frac{1}{2}- \gamma_{*}
\rho_{*}+\frac{\gamma_*}{2\eta_*}-\frac{d_*}{2\eta_*}\\
    \end{aligned}
\end{equation*}
where in $(a)$ we used that $R(z)$ is increasing by Lemma \ref{lem:cauchy},
so $R(z) \geq \lim_{z \to 0} R(z)=\E[-\D^2]=-d_*$,
and in (b) we used Jensen's inequality and the condition that $\pi$ has mean 0
to bound $\mathbb{E} \log c_{\pi}\left(-\frac{1}{2} \gamma_{*}, \gamma_{*}
\Xstar+\sqrt{\gamma_{*}} \Zs\right)\ge \int (-\frac{1}{2}\gamma_*)x^2\,d\pi(x)
=-\frac{1}{2}\gamma_*\rho_*$. Then applying
$\eta_*^{-1}\le \rho_*\le \pil$ by Proposition \ref{prop:uniquefix},
and $\gamma_* \in [d_*/2,\,2d_*]$ for all $\e<\e_0(\pil)$ sufficiently small
by \Cref{prop:betalimit},
\begin{equation}\label{eq:PsiLB}
    \Psi_{\mathrm{RS}}> -\frac{1}{2}-\frac{5d_*}{2} \pil=:\Psi_{\operatorname{LB}}.
\end{equation}

Fix a constant $K>6\pil$, let $\cU=(0,K)$, and consider
\[\Z-\Z(\overline{\cU})=\Z((K,\infty))
=\int \mathbb{I}\left(\frac{1}{n}\left\|\s-\st\right\|^{2}>K\right) \cdot \exp
\left(-\frac{\left\|A \st+\epsilon-A \s\right\|^{2}}{2}\right) \prod_{i=1}^{n}
d\pi(\s_{i}).\]
Define the event $\mathcal{A}$ where
\[n^{-1}\|\st\|^2<2\pib, \quad n^{-1}\|Q\epsilon\|^2<2,
\quad \min(\diag(D^\top D))>d_*-2\e.\]
Observe that under the conditions $\min(\diag(D^\top D)) \to d_- \geq d_*-\e$ by
(\ref{eq:Assump2D}) and Assumption \ref{AssumpHighTemp}, 
$\rho_* \leq \pib$ by (\ref{eq:subGaussian}), and the concentration
bound of Proposition \ref{prop:concentration}, this event $\mathcal{A}$
holds almost surely for all large $n$. On the event $\mathcal{A}$
and for $\s$ satisfying $n^{-1}\|\s-\st\|^2>K$, let us first bound
\[e^{-\frac{\left\|A \st+\epsilon-A \s\right\|^{2}}{2}}
=e^{-\frac{\left\|DO(\st-\s)+Q\epsilon\right\|^2}{2}}
\leq e^{-\frac{\|DO(\st-\s)\|^2}{4}+\frac{\|Q\epsilon\|^2}{2}}
\leq e^{\left(-\frac{(d_*-2\e)}{4}K+1\right)n}\]
Note that this quantity is also bounded above trivially by $e^0=1$. Let us then
apply $\|\sigma-\st\|^2 \leq 2\|\sigma\|^2+2\|\st\|^2$ to bound
\begin{align}
\mathbb{I}(\mathcal{A}) \cdot \Z((K,\infty)) & \leq 
\exp\left(\min\bigg(0,-\frac{d_*-2\e}{4}K+1\bigg)n\right)
\cdot \int
\mathbb{I}\left(\frac{1}{n}\left\|\s\right\|^{2}>\frac{K-4\pil}{2}\right) \cdot
\prod_{i=1}^{n} d\pi(\s_{i})\nonumber\\
&\le 2\exp\left(\min\bigg(0,-\frac{d_*-2\e}{4}K+1\bigg)n-c_0 \min
\qty(\qty(\frac{K-6\pil}{2\pil})^2,\qty(\frac{K-6\pil}{2\pil}))
n\right)\label{eq:logconctailbound}
\end{align}
where the second inequality applies Proposition \ref{prop:concentration}
and that the mean of $n^{-1}\|\sigma\|^2$ is $\rho_* \leq \pib$ under $\pi$.
If $d_*<1$, then $\Psi_{\operatorname{LB}}>-(1/2)(1+5\pib)$.
Choosing $K>6\pib+2\pib\cdot \max(1,(2c_0)^{-1} \cdot
(1+5\pib))$ and bounding the first term in the exponent
of (\ref{eq:logconctailbound}) by 0, we obtain almost surely
\begin{equation}\label{eq:truncationupperbound}
    \limsup_{n\to \infty } \frac{1}{n} \log \Z((K,\infty))
<\Psi_{\operatorname{LB}}<\Psi_{\operatorname{RS}}.
\end{equation}
If $d_* \geq 1$, then assuming $\e<\e_0<1/4$, the first term in the exponent of
(\ref{eq:logconctailbound}) is at most $[-(d_*/8)K+1]n$. Then choosing
$K>20\pib+12$ and bounding the second term in the exponent
of (\ref{eq:logconctailbound}) by 0 again ensures
(\ref{eq:truncationupperbound}).
In either case, the choice of $K$ depends only on $\pib$.
Combining with $n^{-1}\log \Z(\overline{\cU}) \to \Psi_{\mathrm{RS}}$ from
Corollary \ref{cor:PsiRS}, which holds for $\e<\e_0(K,\pib)$ sufficiently
small, this shows $n^{-1}\log \Z \to \Psi_{\mathrm{RS}}$ almost surely.

To apply the dominated convergence theorem, observe that the 
right side of (\ref{eq:jensen}) may be bounded using
$\int n^{-1}\|\st-\sigma\|^2 \prod_i d\pi(\sigma_i) \leq (2/n)\|\st\|^2+2\Vpi$,
and that $\{\|\st\|^2/n\}_{n \geq 1}$ is uniformly integrable by the tail bound
of Proposition \ref{prop:concentration}. Then the dominated convergence
theorem yields
$n^{-1}\E[\log \Z \mid A] \to \Psi_{\mathrm{RS}}$ almost surely, and
the remainder of the proof is the same as in the setting where $\pi$ has
bounded support.
\end{proof}

\begin{proof}[Proof of \Cref{thm:MMSE}, unbounded support]
Recall $\Psi_{\mathrm{LB}}$ from (\ref{eq:PsiLB}), and consider again
$\Z((K,\infty))$ for $K>6\pil$. Recall the bound
(\ref{eq:truncationupperbound}), and note that this bound holds simultaneously
for every $K>6\pil$ on the event $\mathcal{A}$ which holds almost
surely for all large $n$. Applying this bound with
$K(t)=6\pil+2\pil \cdot \min(1,-(2c_0)^{-1} \cdot (1+5\pib))+t$ when $d_*<1$
and with $K(t)=20\pib+12+4c_0t/\pib$ when $d_* \geq 1$ and $\e<\e_0<1/4$ gives
\[\mathbb{I}(\mathcal{A}) \cdot \Z((K(t),\infty))
\leq 2\exp\left(\Psi_{\mathrm{RS}} \cdot n-\frac{c_0t}{2\pil} \cdot n\right).\]
Write as shorthand $X(\sigma,\st)=n^{-1}\|\sigma-\st\|^2$.
Applying $\E[X \cdot \mathbb{I}(X>t)]=\int_0^\infty \mathbb{P}[X>\max(s,t)]ds$
for any nonnegative random variable $X$, we then have
\begin{align*}
&\mathbb{I}(\mathcal{A}) \Big\langle X(\sigma,\st)
\cdot \mathbb{I}\Big(X(\sigma,\st)>K(1)\Big) \Big\rangle
=\int_0^\infty \mathbb{I}(\mathcal{A}) \cdot 
\Big\langle \mathbb{I}\Big(X(\sigma,\st)>\max(s,K(1))\Big)\Big\rangle\,ds\\
&\qquad \qquad =\frac{1}{\Z}\int_0^\infty
\mathbb{I}(\mathcal{A}) \cdot \Z\Big((\max(s,K(1)),\infty)\Big)\,ds
\leq \frac{2C}{\Z} \exp\left(\Psi_{\text{RS}} \cdot n-\frac{c_0}{2\pil} \cdot
n\right)
\end{align*}
for a constant $C>0$ depending only on $(K(1),\pil,c_0)$. The event
$\mathcal{A}$ holds almost surely for all large $n$, and
the preceding proof of Theorem \ref{thm:maintheorem} verifies $n^{-1}\log \Z \to
\Psi_{\mathrm{RS}}$ almost surely. Writing as shorthand $K=K(1)$,
this shows that almost surely
\begin{equation}\label{eq:logconcexpbound}
\lim_{n,m \to \infty} \left\langle X(\sigma,\st)
\cdot \mathbb{I}\Big(X(\sigma,\st)>K\Big) \right\rangle=0.
\end{equation}

Fixing any small constant $\varsigma>0$ and defining $\cU=(0,K)
\setminus (2\eta_*^{-1}-\varsigma,2\eta_*^{-1}+\varsigma)$, the proof in
the setting of bounded support shows
\begin{equation}\label{eq:logconcmidbound}
\lim_{n,m \to \infty} \left\langle
\mathbb{I}\Big(X(\sigma,\st) \in \overline{\cU}\Big) \right\rangle=0
\end{equation}
where $\overline{\cU}$ is the closure of $\cU$.
Then, applying $2\eta_*^{-1} \leq 2\Vpi \leq 2\pil$,
\begin{align*}
\left\langle\left|X(\sigma,\st)-2 \eta_{*}^{-1}\right|\right\rangle
&\leq \varsigma+\Big\langle\Big|X(\sigma,\st)-2 \eta_{*}^{-1}\Big|
\cdot \mathbb{I}\Big(X(\sigma,\st) \in \overline{\cU}\Big)\Big\rangle
+\Big\langle\Big|X(\sigma,\st)-2 \eta_{*}^{-1}\Big|
\cdot \mathbb{I}\Big(X(\sigma,\st)>K\Big)\Big\rangle\\
&\leq \varsigma+(K+2\pib) \cdot \Big\langle \mathbb{I}\Big(X(\sigma,\st) \in
\overline{\cU}\Big) \Big \rangle
+\Big\langle \Big(X(\sigma,\st)+2\pil\Big)
\cdot \mathbb{I}\Big(X(\sigma,\st)>K\Big) \Big\rangle.
\end{align*}
This last bound is at most $2\varsigma$ almost surely for all large $n$ by
(\ref{eq:logconcexpbound}) and (\ref{eq:logconcmidbound}). Thus, almost surely
\[\lim_{n,m \to \infty} \frac{1}{2n}\langle \|\sigma-\st\|^2 \rangle
=\lim_{n,m \to \infty} \frac{X(\sigma,\st)}{2}=\eta_*^{-1}.\]

To apply the dominated convergence theorem, note that $(2n)^{-1}\langle
\|\sigma-\st\|^2 \rangle \leq \langle \|\sigma\|^2/n \rangle+\|\st\|^2/n$.
Here $\{\|\st\|^2/n\}_{n \geq 1}$ is uniformly integrable by the tail bound of
Proposition \ref{prop:concentration}, and $\{\langle \|\sigma\|^2/n \rangle\}_{n
\geq 1}$ is uniformly integrable as it is uniformly bounded in $L^2$:
\[\E\bigg[\left\langle\frac{\|\sigma\|^2}{n} \right\rangle^2 \bigg]
\leq \E\left[\left\langle\frac{\|\sigma\|^4}{n^2} \right\rangle \right]
=\E\left[\frac{\|\sigma\|^4}{n^2}\right]
\leq \frac{1}{n}\sum_{i=1}^n \E[\sigma_i^4]
=\E_{\Xs \sim \pi}[\Xs^4]<\infty.\]
Thus the dominated convergence theorem yields
$\E[(2n)^{-1}\langle \|\sigma-\st\|^2 \rangle] \to \eta_*^{-1}$ almost surely,
and \Cref{lemma:harli} concludes the proof.
\end{proof}

\begin{proof}[Proof of Theorem \ref{thm:TAP}, unbounded support]
Having established Theorem \ref{thm:MMSE} under
Assumption \ref{AssumpPrior}, the proofs of Proposition \ref{prop:compmse}
and Theorem \ref{thm:TAP} are the same as in the setting where $\pi$ has
bounded support. We note that the expectation of the right side of
(\ref{eq:AMPbound}) remains finite and independent of $n$, by the sub-Gaussian
condition (\ref{eq:subGaussian}) for $\pi$,
justifying the application of the
dominated convergence theorem. The same bound (\ref{eq:AMPbound}) holds for
$\|r_1^t\|^4/n^2$ and each $t \geq 1$, so that the application of the
dominated convergence theorem for (\ref{eq:TAPDCT}) is justified by the
Lipschitz condition for $f$ and the same argument.
\end{proof}

\section{Auxiliary lemmas}\label{appendix:technical} 

\subsection{Empirical Wasserstein convergence}\label{appendix:Wasserstein}

\begin{Definition}\label{def:Wass}
For a matrix $\left(v_{1}, \ldots, v_{k}\right)=$ $\left(v_{i, 1}, \ldots, v_{i, k}\right)_{i=1}^{n} \in \mathbb{R}^{n \times k}$ and a random vector $\left(\mathsf{V}_{1}, \ldots, \mathsf{V}_{k}\right)$, we write
	$$
\left(v_{1}, \ldots, v_{k}\right) \stackrel{W}{\rightarrow}\left(\mathsf{V}_{1}, \ldots, \mathsf{V}_{k}\right)
	$$
	for the convergence of the empirical distribution of rows of
$\left(v_{1}, \ldots, v_{k}\right)$ to $(\Vs_1,\ldots,\Vs_k)$ in
Wasserstein-$\pfrak$ for every order $\pfrak \geq 1$. This means that
$(\Vs_1,\ldots,\Vs_k)$ has finite mixed moments of all orders, and
for any continuous function
$f: \mathbb{R}^{k} \rightarrow \mathbb{R}$ satisfying
	\begin{equation}\label{eq:wasslip}
		\left|f\left(v_{1}, \ldots, v_{k}\right)\right| \leq C\left(1+\left\|\left(v_{1}, \ldots, v_{k}\right)\right\|^{\pfrak}\right)
	\end{equation}
for some $C>0$ and $\pfrak \geq 1$, we have $\lim_{n \to \infty}
\frac{1}{n} \sum_{i=1}^{n} f\left(v_{i, 1}, \ldots, v_{i,
k}\right)=\mathbb{E}\left[f\left(\mathsf{V}_{1}, \ldots,
\mathsf{V}_{k}\right)\right]$.
\end{Definition}

The following results are direct consequences of \cite[Propositions E.1, E.2,
E.4, F.2]{fan2020approximate}.
\begin{Proposition}\label{prop:iidW}
Suppose $V \in \R^{n \times t}$ has i.i.d.\ rows equal in law to $\Vs \in \R^t$,
which has finite mixed moments of all orders. Then $V \toW \Vs$ almost surely as $n
\to \infty$. Furthermore, if $E \in \R^{n \times k}$ is deterministic with $E
\toW \Es$, then $(V,E) \toW (\Vs,\Es)$ almost surely
where $\Vs$ is independent of $\Es$.
\end{Proposition}
\begin{Proposition}\label{prop:contW}
Suppose $V \in \R^{n \times k}$ satisfies $V \toW \Vs$ as $n \to \infty$,
and $g:\R^k \to \R^l$ is continuous with $\|g(v)\| \leq C(1+\|v\|)^\pfrak$ for some $C>0$ and
$\pfrak \geq 1$. Then $g(V) \toW g(\Vs)$ where $g(\cdot)$ is applied
row-wise to $V$.
\end{Proposition}
\begin{Proposition}\label{prop:combW}
Suppose $V \in \R^{n \times k}$, $W \in \R^{n \times l}$, and $M_n,M \in \R^{k
\times l}$ satisfy $V \toW \Vs$, $W \toW 0$, and $M_n \to M$ entrywise as $n \to
\infty$. Then $VM_n+W \toW \Vs^\top \cdot M$.
\end{Proposition}
\begin{Proposition}\label{prop:orthoW}
Fix $l \geq 0$, let $O \sim \Haar(\SO(n-l))$, and let $v \in \R^{n-l}$ 
and $\Pi \in \R^{n \times (n-l)}$ be deterministic, where $\Pi$ has orthonormal
columns and $n^{-1}\|v\|^2 \to \sigma^2$ as $n \to \infty$. Then
$\Pi O v \toW \Zs \sim N(0,\sigma^2)$ almost surely.
Furthermore, if $E \in \R^{n \times k}$ is deterministic with $E \toW \Es$,
then $(\Pi O v,E) \toW (\Zs,\Es)$ almost surely
where $\Zs$ is independent of $\Es$.
\end{Proposition}
We note that Proposition \ref{prop:orthoW} is stated in \cite[Proposition
F.2(a)]{fan2020approximate} for $O \sim \Haar(\O(n-l))$, but the proof is
identical also for $O \sim \Haar(\SO(n-l))$.

\subsection{Properties of Cauchy- and R-transform}

Let $\{\mu_k\}_{k\ge 2}$ and $\left\{\kappa_{k}\right\}_{k
\geq 1}$ be the central moments and free cumulants of $-\D^2$ respectively (see e.g. \cite[Lecture 11]{nica2006lectures}).
In particular, $\kappa_{1}=-\mathbb{E} \mathsf{D}^{2}=-d_*$
and $\kappa_{2}=\mu_2=\mathbb{V}\left(\mathsf{D}^{2}\right)$. 
The following shows that the Cauchy- and R-transforms of $-\D^2$ are
well-defined by (\ref{eq:CauchyR}), and reviews their properties.

\begin{Lemma}\label{lem:cauchy}
Let $G(\cdot)$ and $R(\cdot)$ be the Cauchy- and R-transforms of $-\D^2$
under Assumption \ref{AssumpD}.
    \begin{itemize}
        \item[(a)] The function $G: (-d_{-}, \infty) \to \R$ is positive and
strictly decreasing. Setting $G\left(-d_-\right) := \lim _{z \to -d_-} G(z)
\in (0,\infty]$, $G$ admits a functional inverse $G^{-1}:(0,G(-d_-)) \to (-d_-,\infty)$.
\item[(b)] The function $R:\left(0,G\left(-d_-\right)\right) \to
\mathbb{R}$ is negative and strictly increasing.
        \item[(c)] Suppose that \Cref{AssumpHighTemp} holds.
Then $G(-d_-) \in [(2\e)^{-1}, \infty]$. 
Furthermore, $\kappa_2\leq \min\{\e^2,d_*\e\}$, and for $k \geq 2$,
\begin{equation}\label{eq:cumulantbound}
|\mu_k| \leq \e^{k-2}\kappa_2 \leq \e^k, \qquad
|\kappa_k| \leq 16^k\e^{k-2}\kappa_2 \leq (16\e)^k.
\end{equation}
For all $z \in (0,(16\e)^{-1})$,
the R-transform admits the convergent series expansion
\begin{equation}\label{eq:Rseries}
R(z) = \sum_{k \geq 1} \kappa_k z^{k-1}.
\end{equation}
    \end{itemize}
\end{Lemma}
\begin{proof}
\begin{itemize}
    \item[(a)] The positivity and monotonicity of $G$ follow directly from the
definition. Since $\lim_{z \to -d_-} G(z)=G(-d_-)$ and $\lim_{z \to \infty}
G(z)=0$, $G$ has an inverse on the stated domain.
   \item[(b)] For any $z\in \left(0, G\left(-d_{-}\right)\right)$,
\[z=G(G^{-1}(z))
=\mathbb{E} \frac{1}{\D^{2}+R(z)+\frac{1}{z}} < \frac{1}{R(z)+\frac{1}{z}}\]
where the inequality is strict because $\D^2$ has strictly positive mean and
variance. Then $R(z)<0$. Furthermore, by Jensen's inequality $\mathbb{E}
\frac{1}{\left(\D^{2}+G^{-1}(z)\right)^{2}} > \left(\mathbb{E}
\frac{1}{\D^{2}+G^{-1}(z)}\right)^{2}=z^{2}$ which also holds strictly since
$\D^2$ has strictly positive variance,
\[R^{\prime}(z)=\frac{1}{G^{\prime}\left(G^{-1}(z)\right)}+\frac{1}{z^{2}}=-\left(\mathbb{E}
\frac{1}{\left(\D^{2}+G^{-1}(z)\right)^{2}}\right)^{-1}+\frac{1}{z^{2}} > 0.\]
Thus $R(z)$ is strictly increasing.
\item[(c)] 
Under \Cref{AssumpHighTemp}, we have both
$d_- \in [d_*-\e,d_*+\e]$ and $\D^2 \in [d_*-\e,d_*+\e]$
almost surely, so $G\left(-d_-\right) \in [(2\e)^{-1}, \infty]$. 
Furthermore $\kappa_2=\V(\D^2)\le \e^2$,
and also $\kappa_2 \le \E \D^2(d_*+\e)-(\E \D^2)^2=d_*\e$.
For any $k \geq 2$, 
\[|\mu_k|=|\E[(-\D^2+d_*)^k]| \le \e^{k-2}
\mathbb{E}\left(-\D^{2}+d_*\right)^{2}=\e^{k-2} \kappa_{2} \leq \e^{k}.\]
The free cumulants $\kappa_k$ for $k \geq 2$ are the same as those of the
centered variable $-\D^2+d_*$. Then, setting $\mu_1=0$,
the non-crossing moment-cumulant relations applied to $-\D^2+d_*$ yield
\[\left|\kappa_{k}\right|=\left| \sum_{\pi \in \operatorname{NC}(k)}
\operatorname{Mobi}\left(\pi, 1_{k}\right) \cdot \prod_{S \in \pi}
\mu_{|S|}\right| \leq 16^{k} \max _{\pi \in \operatorname{NC}(k)} \prod_{S \in
\pi} \left|\mu_{|S|}\right| \leq 16^{k} \e^{k-2} \kappa_{2} \leq(16 \e)^{k}\]
where $\operatorname{NC}(k)$ is the lattice of all non-crossing partitions of
$\{1,...,k\}$, $\operatorname{Mobi}(\cdot, \cdot)$ are the M\"obius functions on
the non-crossing partition lattice, $1_{k}$ is the trivial partition consisting
of the single set $\{1, \ldots, k\}$, and the first inequality applies
$\left|\mathrm{Mobi}\left(\pi, 1_{k}\right)\right| \leq$ $4^{k}$ and
$|\operatorname{NC}(k)| \leq 4^{k}$ \cite[Proposition 13.15]{nica2006lectures}. 
The statement on the R-transform follows from \cite[Notation 12.6, Proposition
13.15]{nica2006lectures}).
\end{itemize}
\end{proof}

\subsection{Prior distribution}\label{appendix:prior}

We verify that Assumption \ref{AssumpPrior} holds for priors having
bounded support or log-concave density.

\begin{Proposition}\label{prop:priorconditions}
Suppose $\pi$ has mean 0 and variance $\Vpi>0$. Then Assumption
\ref{AssumpPrior} holds if
\begin{enumerate}
\item[(a)] $\pi$ has support contained in $[-\sqrt{\pib},\sqrt{\pib}]$, or
\item[(b)] $\pi$ admits a Lebesgue density function $e^{-g(x)}$ for all $x \in
\R$, where $g''(x) \geq 1/\pib$.
\end{enumerate}
\end{Proposition}
\begin{proof}
Under (a), both statements of (\ref{eq:poincare}) and the first statement of
(\ref{eq:subGaussian}) are evident, and the
second statement of (\ref{eq:subGaussian}) follows from Hoeffding's inequality.
Under (b), the first statement of (\ref{eq:subGaussian}) follows from the
Brascamp-Lieb inequality $\V[f(\Xstar)] \leq \pib \cdot \E[f'(\Xstar)^2]$,
see e.g.\ \cite[Theorem 13.13]{erdHos2017dynamical},
and the second from the Bakry-Emery theorem, see
e.g.\ \cite[Theorem 13.6, Proposition 13.8]{erdHos2017dynamical}.
We observe that for any $\Gamma \preceq (4\pib)^{-1}I$, the measure $\mu$
has a density $e^{-g_\mu(x)}$ where $\nabla^2 g_\mu(x) \succeq
3/(4\pib)$. Hence the Brascamp-Lieb inequality applies also to $\mu$,
and both statements of (\ref{eq:poincare}) follow.
\end{proof}

In the proofs of the main results, we use the following implications of
Assumption \ref{AssumpPrior}.

\begin{Proposition}\label{prop:denoiserlipschitz}
Under Assumption \ref{AssumpPrior}, the posterior mean denoiser
$f(y,\gamma)$ defined by (\ref{eq:deff}) is continuously-differentiable
and Lipschitz in $y$, with derivative
$\frac{\partial}{\partial y} f(y,\gamma)=\gamma \cdot \V_\mu[x]$.
In particular, $|\frac{\partial}{\partial y} f(y,\gamma)| \leq C\gamma$
for a constant $C>0$ depending only on $\pib$.
\end{Proposition}
\begin{proof}
In the notation of Assumption \ref{AssumpPrior}, setting $k=1$,
$\Gamma=-\frac{\gamma}{2}$, and $z=\gamma y$, we have
$f(y,\gamma)=\langle x \rangle_\mu$.
A straightforward application of the dominated convergence theorem shows
that $y \mapsto f(y,\gamma)$ is continuously differentiable, with derivative
$\frac{\partial}{\partial y} f(y,\gamma)=\gamma \cdot \V_\mu[x]$.
The inequality $|\frac{\partial}{\partial y} f(y,\gamma)|
\leq C\gamma$ then follows from (\ref{eq:poincare}), and this implies that
$f(y,\gamma)$ is $(C\gamma)$-Lipschitz in $y$.
\end{proof}

\begin{Proposition}\label{prop:concentration}
Under Assumption \ref{AssumpPrior},
suppose $\sigma_1,\ldots,\sigma_n \overset{iid}{\sim}\pi$. Then for a universal
constant $c_0>0$ and any $s>0$,
\[\mathbb{P}\left(\left|\frac{1}{n} \sum_{i=1}^n (\sigma_i^2-\Vpi)
        \right| \geq s \right) \leq 2 \exp \left(-c_0 \min
            \left(\frac{s^2}{\pil^{2}}, \frac{s}{\pil}\right) n\right)\]
\end{Proposition}
\begin{proof}
Under the sub-Gaussian condition (\ref{eq:subGaussian}), the random variables
$\sigma_i^2$ are sub-exponential with mean $\Vpi$,
so the result follows from Bernstein's inequality,
see e.g.\ \cite[Theorem 2.8.2]{vershynin2018high}.
\end{proof}

\begin{Proposition}\label{prop:bayesbnd}
Fix $k \in \{1,2\}$ and let
$\Gamma,z$ and $\langle \cdot \rangle_\mu,\V_\mu[\cdot]$ be
as defined in Assumption \ref{AssumpPrior}. Let $\gamma_{\max},\gamma_{\min}$
be the largest and smallest eigenvalues of $\Gamma$.
Then for any unit vector $v \in \R^k$ and for a
constant $C>0$ depending only on $\pib$,
\[\V_\mu[(v^\top x)^2] \leq C\left(1+\frac{\|z\|^2}{(\pib^{-1}-\gamma_{\max})^2}
+\frac{\max(-\gamma_{\min},0)}{\pib^{-1}-\gamma_{\max}}\right).\]
\end{Proposition}
\begin{proof}
Applying both conditions of (\ref{eq:poincare}),
\begin{equation}\label{eq:bayesbndtmp}
\V_\mu[(v^\top x)^2]
\leq C(1+\V_\mu[v^\top x]+\langle v^\top x\rangle_\mu^2)
\leq C(1+C+\langle v^\top x\rangle_\mu^2)
\end{equation}
so it suffices to bound $\langle v^\top x\rangle_\mu^2$.
We apply an idea similar to
\cite[Proposition 2]{polyanskiy2016wasserstein}: Set $a=3/(4\pib)$ and
denote $\Omega=aI-2\Gamma$, $w=\Omega^{-1}z$,
$\|x-w\|_\Omega^2=(x-w)^\top \Omega (x-w)$, and $d\pi(x)=\prod_{i=1}^k
d\pi(x_i)$. Let $\omega_{\min}=a-2\gamma_{\max}$
be the smallest eigenvalue of $\Omega$, and note
that $\omega_{\min}>0$ because $\gamma_{\max}<(4\pib)^{-1}$
in Assumption \ref{AssumpPrior}. We have
\[\langle v^\top x \rangle_\mu=
\frac{\int v^\top x \cdot e^{\frac{a}{2}\|x\|^2}
e^{-\frac{1}{2}\|x-w\|_{\Omega}^2}d\pi(x)}
{\int e^{\frac{a}{2}\|x\|^2} e^{-\frac{1}{2}\|x-w\|_{\Omega}^2}d\pi(x)}.\]
Denote $c_\Omega(w)=\log \int e^{\frac{a}{2}\|x\|^2}
e^{-\frac{1}{2}\|x-w\|_\Omega^2}d\pi(x)$.
On the event $\|x-w\|_\Omega^2 \geq -2c_\Omega(w)$, we have
\[\frac{e^{-\frac{1}{2}\|x-w\|_\Omega^2}}{\int e^{\frac{a}{2}\|x\|^2}
e^{-\frac{1}{2}\|x-w\|_\Omega^2}d\pi(x)} \leq 1.\]
On the complementary event $\|x-w\|_\Omega^2<-2c_\Omega(w)$, we have
$c_\Omega(w)<0$ and
\[|v^\top x| \leq \|w\|+\|x-w\|
\leq \|w\|+\omega_{\min}^{-1/2}\|x-w\|_\Omega
\leq \omega_{\min}^{-1}\|z\|
+\omega_{\min}^{-1/2}(-2c_\Omega(w))^{1/2}.\]
Thus, combining these bounds,
\[|\langle v^\top x \rangle_\mu| \leq
\int |v^\top x| \cdot e^{\frac{a}{2}\|x\|^2}\cdot 1\,d\pi(x)
+\Big(\omega_{\min}^{-1}\|z\|
+\omega_{\min}^{-1/2}(-2c_\Omega(w))^{1/2}\Big)
\cdot \1\{c_\Omega(w)<0\}.\]
Applying $\frac{a}{2}=\frac{3}{8\pib}$ and the sub-Gaussian tail bound
(\ref{eq:subGaussian}), it is easily checked that the first term satisfies
$\int |v^\top x|\,e^{\frac{a}{2}\|x\|^2}d\pi(x)<C$ for a constant $C>0$
depending only on $\pib$. For the second term,
by Jensen's inequality and the condition that $\pi$ has mean 0,
\begin{align*}
-c_\Omega(w) &\leq \int
\left(-\frac{a}{2}\|x\|^2+\frac{1}{2}\|x-w\|_\Omega^2\right)\,d\pi(x)
\leq -\gamma_{\min}\int \|x\|^2\,d\pi(x)+\frac{1}{2}\|w\|_\Omega^2\\
&\leq C\Big(\max(-\gamma_{\min},0)
+\omega_{\min}^{-1} \|z\|^2\Big).
\end{align*}
Applying this above yields
$|\langle v^\top x \rangle_\mu| \leq
C(1+\omega_{\min}^{-1}\|z\|+\omega_{\min}^{-1/2}\max(-\gamma_{\min},0)^{1/2})$.
Then, applying this to (\ref{eq:bayesbndtmp}) and using $\omega_{\min} \geq
c(\pib^{-1}-\gamma_{\max})$ for a universal constant $c>0$ concludes the
proof.
\end{proof}

\subsection{Varadhan's lemma}

We apply the following version of Varadhan's lemma; the proof is a
straightforward extension of \cite[Lemmas 4.3.4 and 4.3.6]{Dembo1998large} 
and omitted for brevity.

\begin{Lemma}\label{lemma:varadhan}
Let $(X_n)_{n \geq 1}$ be a sequence of random variables taking values in a
regular topological space $\mathcal{X}$, and let $f:\mathcal{X} \to \R$ be a
bounded continuous function.
\begin{enumerate}
\item[(a)] If $\lambda^*:\mathcal{X} \to [0,\infty]$ is such that
$\liminf_{n \to \infty} n^{-1}\log \mathbb{P}[X_n \in G]
\geq -\inf_{x \in G} \lambda^*(x)$ for all open $G \subseteq \mathcal{X}$, then
\[\liminf_{n \to \infty} \frac{1}{n}\log \E[\mathbb{I}(X_n \in G)e^{nf(X_n)}]
\geq \sup_{x \in G} f(x)-\lambda^*(x) \text{ for all open } G \subseteq
\mathcal{X}.\]
\item[(b)] If $\lambda^*:\mathcal{X} \to [0,\infty]$ is lower-semicontinuous,
the level sets $\{x \in \mathcal{X}:\lambda^*(x) \leq K\}$ are compact for
all $K \in [0,\infty)$, and
$\limsup_{n \to \infty} n^{-1}\log \mathbb{P}[X_n \in F]
\leq -\inf_{x \in F} \lambda^*(x)$ for all closed $F \subseteq \mathcal{X}$,
then
\[\limsup_{n \to \infty} \frac{1}{n}\log \E[\mathbb{I}(X_n \in F)e^{nf(X_n)}]
\leq \sup_{x \in F} f(x)-\lambda^*(x) \text{ for all closed } F \subseteq
\mathcal{X}.\]
\end{enumerate}
\end{Lemma}

\subsection{Extension of results to $O \sim \Haar(\O(n))$}\label{appendix:On}

We explain the claim of Remark \ref{remark:On}. Observe first that if $D,Q$
are random and independent of $O,\st,\epsilon$, where (\ref{eq:Assump2D}) holds
almost surely as $n,m \to \infty$,
then the results of Theorems \ref{thm:maintheorem},
\ref{thm:MMSE}, and \ref{thm:TAP} all hold almost surely as $n,m \to \infty$
conditional on $D,Q$, and hence also
unconditionally. If $A=Q^\top DO$ where $O \sim
\Haar(\O(n))$, then we have the equality in law $O \stackrel{L}{=}PO'$
where $O' \sim \operatorname{Haar}(\mathbb{S O}(n))$ and
$P=\operatorname{diag}(1, \ldots, 1, b) \in \R^{n \times n}$
with $b \in\{+1,-1\}$ having equal
probability. If $n>m$, then $DP=D$, so $A=Q^\top DO'$ and the model is
identical to the setting of $O \sim \Haar(\SO(n))$. If $n \leq m$,
then $DP=P'D$ where $P'=\operatorname{diag}(1,\ldots,1,b,1,\ldots,1) \in \R^{m
\times m}$ has $b$ in
the $n^\text{th}$ entry. Setting $Q'=P'Q$, this implies $A=(Q')^\top DO'$.
The asymptotic statements of
Theorems \ref{thm:maintheorem}, \ref{thm:MMSE}, and \ref{thm:TAP} thus
hold almost surely conditional on $b$, and hence also unconditionally.\\\\

\printbibliography

@book{erdHos2017dynamical,
  title={A dynamical approach to random matrix theory},
  author={Erd{\H{o}}s, L{\'a}szl{\'o} and Yau, Horng-Tzer},
  volume={28},
  year={2017},
  publisher={American Mathematical Soc.}
}

@inproceedings{takeuchi2017rigorous,
  title={Rigorous dynamics of expectation-propagation-based signal recovery from unitarily invariant measurements},
  author={Takeuchi, Keigo},
  booktitle={2017 IEEE International Symposium on Information Theory (ISIT)},
  pages={501--505},
  year={2017},
  organization={IEEE}
}

@inproceedings{ding2019capacity,
  title={Capacity lower bound for the {I}sing perceptron},
  author={Ding, Jian and Sun, Nike},
  booktitle={Proceedings of the 51st Annual ACM SIGACT Symposium on Theory of Computing},
  pages={816--827},
  year={2019}
}

@inproceedings{bolthausen2018morita,
  title={A {M}orita type proof of the replica-symmetric formula for {SK}},
  author={Bolthausen, Erwin},
  booktitle={International Conference on Statistical Mechanics of Classical and Disordered Systems},
  pages={63--93},
  year={2018},
  organization={Springer}
}

@article{fan2020approximate,
  title={Approximate message passing algorithms for rotationally invariant matrices},
  author={Fan, Zhou},
  journal={arXiv preprint arXiv:2008.11892},
  year={2020}
}

@article{thouless1977solution,
  title={Solution of 'solvable model of a spin glass'},
  author={Thouless, David J and Anderson, Philip W and Palmer, Robert G},
  journal={Philosophical Magazine},
  volume={35},
  number={3},
  pages={593--601},
  year={1977},
  publisher={Taylor \& Francis}
}

@article{maillard2019high,
  title={High-temperature expansions and message passing algorithms},
  author={Maillard, Antoine and Foini, Laura and Castellanos, Alejandro Lage and Krzakala, Florent and M{\'e}zard, Marc and Zdeborov{\'a}, Lenka},
  journal={Journal of Statistical Mechanics: Theory and Experiment},
  volume={2019},
  number={11},
  pages={113301},
  year={2019},
  publisher={IOP Publishing}
}

@article{marinari1994replica,
  title={Replica field theory for deterministic models. II. A non-random spin glass with glassy behaviour},
  author={Marinari, Enzo and Parisi, Giorgio and Ritort, Felix},
  journal={Journal of Physics A: Mathematical and General},
  volume={27},
  number={23},
  pages={7647},
  year={1994},
  publisher={IOP Publishing}
}

@article{fan2021replica,
  title={The replica-symmetric free energy for Ising spin glasses with orthogonally invariant couplings},
  author={Fan, Zhou and Wu, Yihong},
  journal={arXiv preprint arXiv:2105.02797},
  year={2021}
}

@article{parisi1995mean,
  title={Mean-field equations for spin models with orthogonal interaction matrices},
  author={Parisi, Giorgio and Potters, Marc},
  journal={Journal of Physics A: Mathematical and General},
  volume={28},
  number={18},
  pages={5267},
  year={1995},
  publisher={IOP Publishing}
}

@article{opper2001adaptive,
  title={Adaptive and self-averaging {T}houless-{A}nderson-{P}almer mean-field theory for probabilistic modeling},
  author={Opper, Manfred and Winther, Ole},
  journal={Physical Review E},
  volume={64},
  number={5},
  pages={056131},
  year={2001},
  publisher={APS}
}

@article{opper2001tractable,
  title={Tractable approximations for probabilistic models: {T}he adaptive {T}houless-{A}nderson-{P}almer mean field approach},
  author={Opper, Manfred and Winther, Ole},
  journal={Physical Review Letters},
  volume={86},
  number={17},
  pages={3695},
  year={2001},
  publisher={APS}
}

@article{opper2005expectation,
  title={Expectation consistent approximate inference.},
  author={Opper, Manfred and Winther, Ole and Jordan, Michael J},
  journal={Journal of Machine Learning Research},
  volume={6},
  number={12},
  year={2005}
}

@book{vershynin2018high,
  title={High-dimensional probability: An introduction with applications in data science},
  author={Vershynin, Roman},
  volume={47},
  year={2018},
  publisher={Cambridge university press}
}

@book{Dembo1998large,
  title={Large deviations techniques and applications},
  author={Dembo, Amir and Zeitouni, Ofer},
  volume={38},
  year={1998},
  publisher={Springer}
}

@book{anderson2010introduction,
  title={An introduction to random matrices},
  author={Anderson, Greg W and Guionnet, Alice and Zeitouni, Ofer},
  number={118},
  year={2010},
  publisher={Cambridge university press}
}

@book{rockafellar2015convex,
  title={Convex analysis},
  author={Rockafellar, Ralph Tyrell},
  year={2015},
  publisher={Princeton university press}
}

@article{guerra2003broken,
  title={Broken replica symmetry bounds in the mean field spin glass model},
  author={Guerra, Francesco},
  journal={Communications in mathematical physics},
  volume={233},
  number={1},
  pages={1--12},
  year={2003},
  publisher={Springer}
}

@book{nica2006lectures,
	title={Lectures on the combinatorics of free probability},
	author={Nica, Alexandru and Speicher, Roland},
	volume={13},
	year={2006},
	publisher={Cambridge University Press}
}

@inproceedings{barbier2018mutual,
	title={The mutual information in random linear estimation beyond iid matrices},
	author={Barbier, Jean and Macris, Nicolas and Maillard, Antoine and Krzakala, Florent},
	booktitle={2018 IEEE International Symposium on Information Theory (ISIT)},
	pages={1390--1394},
	year={2018},
	organization={IEEE}
}

@article{tulino2013support,
	title={Support recovery with sparsely sampled free random matrices},
	author={Tulino, Antonia M and Caire, Giuseppe and Verd{\'u}, Sergio and Shamai, Shlomo},
	journal={IEEE Transactions on Information Theory},
	volume={59},
	number={7},
	pages={4243--4271},
	year={2013},
	publisher={IEEE}
}

@article{rangan2019vector,
  title={Vector approximate message passing},
  author={Rangan, Sundeep and Schniter, Philip and Fletcher, Alyson K},
  journal={IEEE Transactions on Information Theory},
  volume={65},
  number={10},
  pages={6664--6684},
  year={2019},
  publisher={IEEE}
}

@article{fan2022tap,
	title={{TAP} equations for orthogonally invariant spin glasses at high temperature},
	author={Fan, Zhou and Li, Yufan and Sen, Subhabrata},
	journal={arXiv preprint arXiv:2202.09325},
	year={2022}
}

@article{takeda2006analysis,
  title={Analysis of {CDMA} systems that are characterized by eigenvalue spectrum},
  author={Takeda, Koujin and Uda, Shinsuke and Kabashima, Yoshiyuki},
  journal={EPL (Europhysics Letters)},
  volume={76},
  number={6},
  pages={1193},
  year={2006},
  publisher={IOP Publishing}
}

@inproceedings{reeves2017additivity,
  title={Additivity of information in multilayer networks via additive {G}aussian noise transforms},
  author={Reeves, Galen},
  booktitle={2017 55th Annual Allerton Conference on Communication, Control, and Computing (Allerton)},
  pages={1064--1070},
  year={2017},
  organization={IEEE}
}

@article{tanaka2002statistical,
  title={A statistical-mechanics approach to large-system analysis of {CDMA} multiuser detectors},
  author={Tanaka, Toshiyuki},
  journal={IEEE Transactions on Information theory},
  volume={48},
  number={11},
  pages={2888--2910},
  year={2002},
  publisher={IEEE}
}

@article{baron2009bayesian,
  title={Bayesian compressive sensing via belief propagation},
  author={Baron, Dror and Sarvotham, Shriram and Baraniuk, Richard G},
  journal={IEEE Transactions on Signal Processing},
  volume={58},
  number={1},
  pages={269--280},
  year={2009},
  publisher={IEEE}
}

@article{guan2011bayesian,
  title={Bayesian variable selection regression for genome-wide association studies and other large-scale problems},
  author={Guan, Yongtao and Stephens, Matthew},
  journal={The Annals of Applied Statistics},
  volume={5},
  number={3},
  pages={1780--1815},
  year={2011},
  publisher={Institute of Mathematical Statistics}
}

@article{guo2005randomly,
  title={Randomly spread {CDMA}: {A}symptotics via statistical physics},
  author={Guo, Dongning and Verd{\'u}, Sergio},
  journal={IEEE Transactions on Information Theory},
  volume={51},
  number={6},
  pages={1983--2010},
  year={2005},
  publisher={IEEE}
}

@inproceedings{barbier2016mutual,
  title={The mutual information in random linear estimation},
  author={Barbier, Jean and Dia, Mohamad and Macris, Nicolas and Krzakala, Florent},
  booktitle={2016 54th Annual Allerton Conference on Communication, Control, and Computing (Allerton)},
  pages={625--632},
  year={2016},
  organization={IEEE}
}

@article{barbier2019optimal,
  title={Optimal errors and phase transitions in high-dimensional generalized linear models},
  author={Barbier, Jean and Krzakala, Florent and Macris, Nicolas and Miolane, L{\'e}o and Zdeborov{\'a}, Lenka},
  journal={Proceedings of the National Academy of Sciences},
  volume={116},
  number={12},
  pages={5451--5460},
  year={2019},
  publisher={National Acad Sciences}
}

@article{barbier2020mutual,
  title={Mutual information and optimality of approximate message-passing in random linear estimation},
  author={Barbier, Jean and Macris, Nicolas and Dia, Mohamad and Krzakala, Florent},
  journal={IEEE Transactions on Information Theory},
  volume={66},
  number={7},
  pages={4270--4303},
  year={2020},
  publisher={IEEE}
}

@inproceedings{reeves2016replica,
  title={The replica-symmetric prediction for compressed sensing with {G}aussian matrices is exact},
  author={Reeves, Galen and Pfister, Henry D},
  booktitle={2016 IEEE International Symposium on Information Theory (ISIT)},
  pages={665--669},
  year={2016},
  organization={IEEE}
}

@inproceedings{montanari2006analysis,
  title={Analysis of belief propagation for non-linear problems: The example of {CDMA} (or: {H}ow to prove {T}anaka's formula)},
  author={Montanari, Andrea and Tse, David},
  booktitle={2006 IEEE Information Theory Workshop-ITW'06 Punta del Este},
  pages={160--164},
  year={2006},
  organization={IEEE}
}

@article{barbier2019adaptive,
  title={The adaptive interpolation method: a simple scheme to prove replica formulas in {B}ayesian inference},
  author={Barbier, Jean and Macris, Nicolas},
  journal={Probability theory and related fields},
  volume={174},
  number={3},
  pages={1133--1185},
  year={2019},
  publisher={Springer}
}

@article{qiu2022tap,
  title={The {TAP} free energy for high-dimensional linear regression},
  author={Qiu, Jiaze and Sen, Subhabrata},
  journal={arXiv preprint arXiv:2203.07539},
  year={2022}
}

@article{ma2017orthogonal,
  title={Orthogonal {AMP}},
  author={Ma, Junjie and Ping, Li},
  journal={IEEE Access},
  volume={5},
  pages={2020--2033},
  year={2017},
  publisher={IEEE}
}

@inproceedings{takeuchi2021bayes,
  title={Bayes-optimal convolutional {AMP}},
  author={Takeuchi, Keigo},
  booktitle={2021 IEEE International Symposium on Information Theory (ISIT)},
  pages={1385--1390},
  year={2021},
  organization={IEEE}
}

@article{liu2022memory,
  title={Memory {AMP}},
  author={Liu, Lei and Huang, Shunqi and Kurkoski, Brian M},
  journal={IEEE Transactions on Information Theory},
  year={2022},
  publisher={IEEE}
}

@article{korada2010tight,
  title={Tight bounds on the capacity of binary input random {CDMA} systems},
  author={Korada, Satish Babu and Macris, Nicolas},
  journal={IEEE Transactions on Information Theory},
  volume={56},
  number={11},
  pages={5590--5613},
  year={2010},
  publisher={IEEE}
}

@article{dudeja2022universality,
  title={Universality of {A}pproximate {M}essage {P}assing with Semi-Random Matrices},
  author={Dudeja, Rishabh and Lu, Yue M and Sen, Subhabrata},
  journal={arXiv preprint arXiv:2204.04281},
  year={2022}
}

@article{dudeja2022spectral,
  title={Spectral Universality of Regularized Linear Regression with Nearly Deterministic Sensing Matrices},
  author={Dudeja, Rishabh and Sen, Subhabrata and Lu, Yue M},
  journal={arXiv preprint arXiv:2208.02753},
  year={2022}
}

@article{wang2022universality,
  title={Universality of approximate message passing algorithms and tensor networks},
  author={Wang, Tianhao and Zhong, Xinyi and Fan, Zhou},
  journal={arXiv preprint arXiv:2206.13037},
  year={2022}
}

@article{guo2005mutual,
  title={Mutual information and minimum mean-square error in {G}aussian channels},
  author={Guo, Dongning and Shamai, Shlomo and Verd{\'u}, Sergio},
  journal={IEEE transactions on information theory},
  volume={51},
  number={4},
  pages={1261--1282},
  year={2005},
  publisher={IEEE}
}

@inproceedings{takeuchi2022convergence,
  title={On the convergence of orthogonal/vector {AMP}: {L}ong-memory message-passing strategy},
  author={Takeuchi, Keigo},
  booktitle={2022 IEEE International Symposium on Information Theory (ISIT)},
  pages={1366--1371},
  year={2022},
  organization={IEEE}
}

@article{polyanskiy2016wasserstein,
  title={Wasserstein continuity of entropy and outer bounds for interference channels},
  author={Polyanskiy, Yury and Wu, Yihong},
  journal={IEEE Transactions on Information Theory},
  volume={62},
  number={7},
  pages={3992--4002},
  year={2016},
  publisher={IEEE}
}

@inproceedings{bolthausen2022gardner,
  title={Gardner formula for {I}sing perceptron models at small densities},
  author={Bolthausen, Erwin and Nakajima, Shuta and Sun, Nike and Xu, Changji},
  booktitle={Conference on Learning Theory},
  pages={1787--1911},
  year={2022},
  organization={PMLR}
}

@article{payaro2009hessian,
  title={Hessian and concavity of mutual information, differential entropy, and entropy power in linear vector {G}aussian channels},
  author={Payar{\'o}, Miquel and Palomar, Daniel P},
  journal={IEEE Transactions on Information Theory},
  volume={55},
  number={8},
  pages={3613--3628},
  year={2009},
  publisher={IEEE}
}

\end{document}